\documentclass[11pt]{article}

\usepackage{amssymb,amsmath,bm}
\usepackage{amsthm}

\usepackage{textcomp}
\usepackage{enumerate}      
\usepackage{graphicx}        
\usepackage{caption}

\usepackage{tabu}

\usepackage[export]{adjustbox}

\usepackage{mathrsfs}

\usepackage{url}

\usepackage{array}
\newcommand{\PreserveBackslash}[1]{\let\temp=\\#1\let\\=\temp}
\newcolumntype{C}[1]{>{\PreserveBackslash\centering}p{#1}}
\newcolumntype{R}[1]{>{\PreserveBackslash\raggedleft}p{#1}}
\newcolumntype{L}[1]{>{\PreserveBackslash\raggedright}p{#1}}

\renewcommand{\setminus}{{\smallsetminus}}

\usepackage{amssymb,amsmath,bm}
\usepackage{amsthm}
\usepackage{hyperref}
\usepackage{mathrsfs}
\usepackage{textcomp}
\usepackage{enumerate}
\usepackage{graphicx}
\usepackage{tikz}
\usepackage{verbatim}

\newtheorem{theorem}{Theorem}[section]
\newtheorem{lemma}[theorem]{Lemma}
\newtheorem{proposition}[theorem]{Proposition}
\newtheorem{condition}[theorem]{Geometric Hypothesis}

\newtheorem{definition}[theorem]{Definition}

\theoremstyle{remark}
\newtheorem{remark}[theorem]{Remark}

\theoremstyle{remark}

\numberwithin{equation}{section}

\begin{document}
\title{\bf Asymptotics of complex $b$-$6j$ symbols}

\author{Yunpeng Meng and Tian Yang}

\date{}

\maketitle

\begin{abstract}
We study the $b$-$6j$ symbols -- an analytic extension of the $6j$-symbols for the principal series of the modular double of $\mathrm U_q\mathfrak{sl}(2;\mathbb R)$ -- with complex index $b,$ refereed to as the \emph{complex $b$-$6j$ symbols}. Then we relate their asymptotics, when the six parameters scale according to the dihedral angles of a hyperideal hyperbolic tetrahedron, to the volume and the determinant of the Gram matrix of the tetrahedron.  In the case $\arg b=\pm \frac{\pi}{4},$ we believe that  this work is closely related to the complex Liouville string\,\cite{CEMR}.
\end{abstract}

\tableofcontents

\section{Introduction}

In \cite{LMSWY}, Liu, Ming, Sun, Wu and the second author considered the $6j$-symbols associated to the principal series (also known as the positive representations) of the modular double of $\mathrm U_{q}\mathfrak{sl}(2;\mathbb R),$ denoted by 
$\mathrm U_{q\tilde{q}}\mathfrak{sl}(2;\mathbb R),$ which is a quantum group of significant interest in mathematical physics \cite{Fad95, Fad16, PT99,TesVar} and representation theory \cite{FI14}. When the six parameters of a $\mathrm U_{q\tilde{q}}\mathfrak{sl}(2;\mathbb R)$  $6j$-symbol scale according to the edge lengths of a truncated hyperideal hyperbolic tetrahedron or a flat tetrahedron, a connection between the semi-classical limit of the $6j$-symbol 
and the hyperbolic volume of the tetrahedron was established. Using these $\mathrm U_{q\tilde{q}}\mathfrak{sl}(2;\mathbb R)$  $6j$-symbols as building blocks, together with their asymptotic properties, the authors of \cite{LMSWY} further defined a family of Turaev-Viro type invariants for hyperbolic $3$-manifolds with totally geodesic boundary, and obtained a thorough understanding of their semi-classical limit and its relationship with the hyperbolic volume and the adjoint twisted Reidemeister torsion of the manifolds. In \cite{LMSWY2}, the same authors completed the study of the asymptotics of the $\mathrm U_{q\tilde{q}}\mathfrak{sl}(2;\mathbb R)$ $6j$-symbols by first showing that if the six parameters of a $\mathrm U_{q\tilde{q}}\mathfrak{sl}(2;\mathbb R)$ $6j$-symbol do not correspond to the edge lengths of a truncated hyperideal hyperbolic tetrahedron nor a flat tetrahedron, then they necessarily correspond to the edge lengths of a truncated hyperideal anti-de Sitter tetrahedron, and then establishing  the connection between the semi-classical limit of the $6j$-symbol 
and the anti-de Sitter volume of the tetrahedron in this case. 
The $\mathrm U_{q\tilde{q}}\mathfrak{sl}(2;\mathbb R)$ $6j$-symbol, as a function of its six parameters, has a meromorphic extension out of the spectrum of the principal series of $\mathrm U_{q\tilde{q}}\mathfrak{sl}(2;\mathbb R),$ called the $b$-$6j$ symbol. 
The authors of \cite{LMSWY2} also studied the asymptotics of these  $b$-$6j$ symbols when the six parameters scale according to the dihedral angles of a truncated hyperideal hyperbolic tetrahedron or a truncated hyperideal anti-de Sitter tetrahedron (both beyond the spectrum of the principal series of $\mathrm U_{q\tilde{q}}\mathfrak{sl}(2;\mathbb R)$), relating them to the volume of the tetrahedron in the corresponding geometry. It is worth emphasizing that all the $b$-$6j$ symbols considered in \cite{LMSWY} and \cite{LMSWY2} are indexed by a real number $b.$  

In this paper, we consider the $b$-$6j$ symbols with a complex index $b,$ refereed to as the \emph{complex $b$-$6j$ symbols} (see Definition \ref{B6j} and Proposition \ref{lem: ab conver}), and relate their asymptotic behavior, in the case that the six parameters scale according to the dihedral angles of a hyperideal hyperbolic tetrahedron, to the volume and the determinant of the Gram matrix of the tetrahedron (see Theorem \ref{vol}).  We hope that these results will contribute to extending the Turaev-Viro type $3$-manifold invariants introduced in \cite{LMSWY} to the setting of complex index $b,$ relating them to the original Turaev-Viro invariants\,\cite{TV}, and ultimately shedding light on the Volume Conjecture for those invariants\,\cite{CY}. We also believe, in the case $\arg b=\pm \frac{\pi}{4},$ that  this work is closely related to the complex Liouville string\,\cite{CEMR}.

\subsection{Complex $b$-$6j$ symbols}

For $b\in\mathbb C$ with $\mathrm{Re}b>0,$ or equivalently with $\arg b\in (-\frac{\pi}{2},\frac{\pi}{2}),$ let $S_b$ be the  \emph{double sine function} given by
\begin{equation}\label{eq:def-S}
S_b(z)=\exp\Bigg(\int_\Omega\frac{\sinh\Big(\big(\frac{Q}{2}-z\big)t\Big)}{4t\sinh(\frac{bt}{2})\sinh(\frac{t}{2b})}dt\Bigg),
\end{equation}
where $Q=b+b^{-1}$ and the contour $\Omega$ goes along the real line and passes above the pole of the integrand at the origin. The integral absolutely converges for $z\in \mathbb C$ with $0<\mathrm{Re}z<\mathrm {Re} Q,$ defining a holomorphic function in such $z.$

For a six-tuple of complex numbers  $(a_1,\dots,a_6),$ we let in the rest of this article \begin{equation*}
\begin{split}
   & t_1=a_1+a_2+a_3,\quad  t_2=a_1+a_5+a_6,\quad t_3=a_2+a_4+a_6,\quad  t_4=a_3+a_4+a_5,\\
   & q_1=a_1+a_2+a_4+a_5,\quad q_2=a_1+a_3+a_4+a_6,\quad  
q_3=a_2+a_3+a_5+a_6\quad\text{and}\quad q_4=2Q.
\end{split}
\end{equation*}
Then the  six-tuple $(a_1,\dots,a_6)\in\mathbb C^6$ is called \emph{$b$-admissible} if 
$$0<\mathrm{Re}(q_j-t_i)<\mathrm{Re}Q$$
for all $i,j\in\{1,2,3,4\}.$ 
In particular, if $(a_1,\dots,a_6)$ is $b$-admissible, then
$$\max\{\mathrm{Re}t_1,\mathrm{Re}t_2,\mathrm{Re}t_3,\mathrm{Re}t_4\}<\min\{\mathrm{Re}q_1,\mathrm{Re}q_2,\mathrm{Re}q_3,\mathrm{Re}q_4\}.$$

\begin{definition}[Complex $b$-$6j$ symbol]\label{B6j}
The \emph{$b$-$6j$ symbol} for a $b$-admissible six-tuple $(a_1,\dots,a_6)\in\mathbb C^6$ is given by 
\begin{equation}\label{b6j}
\bigg\{\begin{matrix} a_1 & a_2 & a_3 \\ a_4 & a_5 & a_6 \end{matrix} \bigg\}_b=\frac{1}{\prod_{i=1}^4\prod_{j=1}^4S_b(q_j-t_i)^{\frac{1}{2}}}\int_\Gamma \prod_{i=1}^4S_b(u-t_i)\prod_{j=1}^4S_b(q_j-u)d u,
\end{equation}
where for $z\in\mathbb C$ with $0<\mathrm{Re}z<\mathrm{Re}Q,$ $S_b(z)^{\frac{1}{2}}$ is defined as $\exp\Big(\frac{1}{2}\int_\Omega\frac{\sinh\big((\frac{Q}{2}-z)t\big)}{4t\sinh(\frac{bt}{2})\sinh(\frac{t}{2b})}dt\Big),$ and the contour $\Gamma$ is any vertical line passing the interval $(\max\{\mathrm{Re}t_i\},\min\{\mathrm{Re}q_j\}).$  See Figure \ref{Gamma}. 
\end{definition}

\begin{figure}[htbp]
\centering
\includegraphics[scale=0.4]{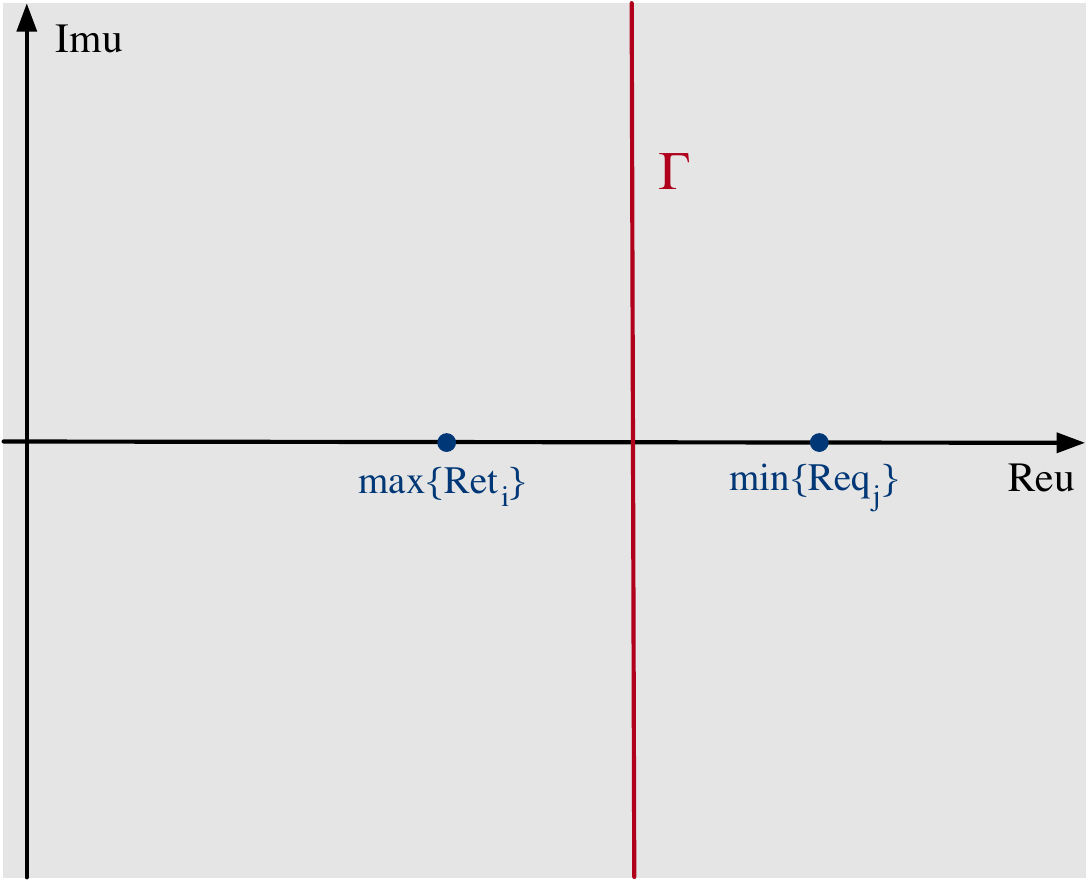}
\caption{The contour $\Gamma.$}
\label{Gamma}
\end{figure}

In Subsection \ref{abconver}, we will prove the following

\begin{proposition}\label{lem: ab conver}
   For any $b$-admissible six-tuple $(a_1,\dots,a_6)\in\mathbb C^6,$ the integral in (\ref{b6j}) converges absolutely, is independent of the choice of the vertical line $\Gamma,$ and depends analytically on $b.$ 
\end{proposition}

\subsection{Asymptotics of complex $b$-$6j$ symbols}

Let $(\theta_1,\dots, \theta_6)$ be the dihedral angles of a truncated hyperideal hyperbolic tetrahedron $\Delta;$ 
and for $k\in\{1,\dots,6\},$ let $a_k =\frac{Q}{2} \pm \frac{\theta_k}{2\pi b}$ with the sign   in $\pm$ chosen arbitrarily.
Then by \cite{BB}, for $i,j\in\{1,2,3,4\},$ 
\begin{equation*}\label{b-adm}
    \mathrm{Re}\bigg(\frac{b}{2}\bigg)< \mathrm{Re}(q_j-t_i) < \mathrm{Re}\bigg(Q-\frac{b}{2}\bigg).
\end{equation*}
In particular, the six-tuple  $(a_1,\dots,a_6)$ is $b$-admissible. 

The main result of this paper is the following 

\begin{theorem}\label{vol} Let $(\theta_1,\dots, \theta_6)$ be the dihedral angles of a truncated hyperideal hyperbolic tetrahedron $\Delta.$ Then as $b\to0$ 
\begin{enumerate}[(1)]
    \item with a fixed $\arg b  \in [-\frac{\pi}{4},\frac{\pi}{4}],$ or
    \item  with a fixed $\arg b  \in (-\frac{\pi}{2},-\frac{\pi}{4})\cup(\frac{\pi}{4},\frac{\pi}{2})$ and with $(\theta_1,\dots,\theta_6)$ satisfying the Geometric Hypothesis \ref{per-min} described in Section \ref{GH},
\end{enumerate}
we have
$$\bigg\{\begin{matrix} \frac{Q}{2} \pm \frac{\theta_1}{2\pi b} & \frac{Q}{2} \pm  \frac{\theta_2}{2\pi b} & \frac{Q}{2} \pm \frac{\theta_3}{2\pi b} \\ \frac{Q}{2} \pm  \frac{\theta_4}{2\pi b} & \frac{Q}{2} \pm  \frac{\theta_5}{2\pi b} & \frac{Q}{2} \pm  \frac{\theta_6}{2\pi b} \end{matrix} \bigg\}_b=\frac{1}{2}\frac{e^{-\frac{\mathrm{Vol}(\Delta)}{\pi b^2}}}{\sqrt[4]{-\det\mathrm{Gram}(\Delta)}} \Big(1 +O\big(b^2\big)\Big),$$
where in each entry of the $b$-$6j$ symbol the sign $+$ or $-$ can be chosen arbitrarily,  $\mathrm{Vol}(\Delta)$ is the hyperbolic volume of $\Delta$ and $\mathrm{Gram}(\Delta)$ is the Gram matrix of $\Delta$ in the dihedral angles. 
\end{theorem} 

\begin{remark}
    Numerical computations suggest that a very large portion of the dihedral angles $(\theta_1,\dots,\theta_6)$ satisfy the Geometric Hypothesis \ref{per-min}. See Remarks \ref{geohyp1} and \ref{geohyp2} for more details. We believe that the Geometric Hypothesis could eventually be removed and Theorem \ref{vol} holds for all the possible dihedral angles $(\theta_1,\dots,\theta_6)$ of the hyperideal hyperbolic tetrahedra. It can also be proved that under the same conditions, the same asymptotic expansion holds as $b\to 0$ with $\arg b$ varying in $(-\frac{\pi}{2}+\delta, \frac{\pi}{2}-\delta)$ for any $\delta>0.$ 
\end{remark}

\bigskip

\noindent\textbf{Acknowledgments.} The authors thank Baojun Wu for drawing their attention to complex Liouville strings, which motivated this work. We also thank Giulio Belletti, Tianyue Liu, Shuang Ming, Xin Sun, J\"org Teschner, Ka Ho Wong and Yuxuan Yang for their interest and discussions. The first author is partially supported, and the second author is supported, by NSF Grant DMS-2505908.

\section{Complex $b$-$6j$ symbols}

\subsection{Proof of Proposition \ref{lem: ab conver}}\label{abconver}

The key ingredient in the proof of Proposition \ref{lem: ab conver} is the following estimate of the double sine function on the suitable regions. 

\begin{lemma}\label{Sb-bound2}
Let $C$ be a compact subset of $\big\{ b\in\mathbb C\ \big|\ \arg b\in (-\frac{\pi}{2},\frac{\pi}{2})\big\},$
 and let $\delta>0$ be sufficiently small. For $b\in C,$ let 
 $$T_{b,\delta}^+=\Big\{ s+\mathbf i t\ \Big|\ s\in[\delta, \mathrm{Re}Q-\delta]\ \ \text{and}\ \  t \geqslant 0\Big\}$$
 and
 $$T_{b,\delta}^-=\Big\{ s - \mathbf i t\ \Big|\ s\in[\delta, \mathrm{Re}Q-\delta]\ \ \text{and}\ \ t \geqslant 0\Big\}.$$
 Then there exists a constant $K=K_{C,\delta}>0$ such that 
$$|S_b(z)| < Ke^{\frac{\pi}{2}\mathrm{Im}\big(z(z-Q)\big)}$$
for all $b\in C$ and $z\in T^+_{b,\delta},$ and 
$$|S_b(z)|< Ke^{-\frac{\pi}{2}\mathrm{Im}\big(z(z-Q)\big)}$$
for all $b\in C$ and $z\in T^-_{b,\delta}.$
 \end{lemma}

 The proof of Lemma \ref{Sb-bound2} relies on the relationship between the double sine function and the \emph{double Gamma function}
$$\Gamma_b(z)\doteq\exp\Bigg(\int_0^{+\infty}\bigg(\frac{e^{\big(\frac{Q}{2}-z\big)t}-1}{4t\sinh(\frac{bt}{2})\sinh(\frac{t}{2b})}-\frac{\big(\frac{Q}{2}-z\big)^2}{2te^{t}}-\frac{\frac{Q}{2}-z}{t^2}\bigg)dt\Bigg)$$
for $b\in\mathbb C$ with $\arg b\in (-\frac{\pi}{2},\frac{\pi}{2}).$ The integral absolutely converges for $z\in \mathbb C$ with $\mathrm{Re}z>0,$ defining a holomorphic function in such $z.$ By e.g. \cite[B. 32]{E}, the double sine function and the double Gamma function are related by 
\begin{equation}\label{SbG}
    S_b(z)=\frac{\Gamma_b(z)}{\Gamma_b(Q-z)}
\end{equation}
for $z\in \mathbb C$ with $0<\mathrm{Re}z<\mathrm{Re}Q.$

\begin{proof}[Proof of Lemma \ref{Sb-bound2}]
For $z\in\mathbb C$ with $\mathrm{Re}z>0,$ we define 
$$\psi_b(z)\doteq \int_0^{+\infty}\Bigg(\frac{e^{-zt}}{\big(1-e^{-bt}\big)\big(1-e^{-\frac{t}{b}}\big)}-\frac{1}{2}\bigg(z^2-Qz+\frac{Q^2+1}{6}\bigg)e^{-t}-\frac{1}{t^2}-\frac{Q-2z}{2t}\Bigg)\frac{dt}{t}.$$
Then by (\ref{SbG}) and a direct computation, we have
\begin{equation}\label{Sbpsi}
    S_b(z)=\exp\Big(\psi_b(z)-\psi_b(Q-z)\Big)
\end{equation}
for $z\in\mathbb C$ with $0<\mathrm{Re}z<\mathrm{Re}Q.$ 
We observe that for $z\in\mathbb C$ with $\mathrm{Re}z>0,$ 
\begin{equation}\label{psiR}
    \psi_b(z)=-\frac{1}{2}\bigg(z^2-Qz+\frac{Q^2+1}{6}\bigg)\log z -\frac{Qz}{2}+\frac{3z^2}{4} +\int_0^{+\infty} h_b(t)e^{-zt}dt,
\end{equation}
where 
$$h_b(t)=\frac{1}{t\big(1-e^{-bt}\big)\big(1-e^{-\frac{t}{b}}\big)}-\frac{1}{t^3}-\frac{Q}{2t^2}-\frac{Q^2+1}{12t}.$$
Indeed, by a direct computation, we have 
\begin{equation*}
\begin{split}
 &\psi_b(z)- \int_0^{+\infty} h_b(t)e^{-zt}dt\\
 = & \int_0^{+\infty} \frac{e^{-zt}-1+zt-\frac{1}{2}z^2t^2e^{-t}}{t^3}dt +\frac{Q}{2}\int_0^{+\infty}\frac{e^{-zt}-1+zte^{-t}}{t^2}dt+\frac{Q^2+1}{12}\int_0^{+\infty}\frac{e^{-zt}-e^{-t}}{t}dt\\
 =& -\frac{1}{2}\bigg(z^2-Qz+\frac{Q^2+1}{6}\bigg)\log z -\frac{Qz}{2}+\frac{3z^2}{4},
\end{split}
\end{equation*}
where the last equality comes from the direct evaluations $\int_0^{+\infty} \frac{e^{-zt}-1+zt-\frac{1}{2}z^2t^2e^{-t}}{t^3}dt=\frac{z^2}{2}\big(\frac{3}{2}-\log z\big),$ $\int_0^{+\infty}\frac{e^{-zt}-1+zte^{-t}}{t^2}dt=z(\log z-1)$ and $\int_0^{+\infty}\frac{e^{-zt}-e^{-t}}{t}dt=-\log z.$
\medskip

As a consequence of (\ref{psiR}), we have 
\begin{equation}\label{psi-psi}
\begin{split}
\psi_b(z)-\psi_b(Q-z)= &-\frac{1}{2}\bigg(z^2-Qz+\frac{Q^2+1}{6}\bigg)\Big(\log z -\log (Q-z)\Big) +\frac{Qz}{2}-\frac{Q^2}{4}\\
&+\int_0^{+\infty} h_b(t)e^{-zt}dt -\int_0^{+\infty} h_b(t)e^{-(Q-z)t}dt
\end{split}
\end{equation}
for $z\in\mathbb C$ with $0<\mathrm{Re}z<\mathrm{Re}Q;$ and we will estimate the terms of the right hand side of (\ref{psi-psi}) below.
\smallskip

First of all, as $C$ is compact, there exists a constant $B_1$ such that 
\begin{equation}\label{B1}
    \bigg|-\frac{Q^2}{4}\bigg|<B_1
\end{equation}
for all $b\in C.$
\smallskip

Next, we estimate the terms $\int_0^{+\infty} h_b(t)e^{-zt}dt$ and $\int_0^{+\infty} h_b(t)e^{-(Q-z)t}dt.$ On the one hand, as $\mathrm{Re}b>0$ and $\mathrm{Re}(b^{-1})>0$ for each $b\in C,$ $\lim_{t\to +\infty}e^{-bt}=\lim_{t\to +\infty}e^{-\frac{t}{b}}=0.$ Then by the compactness of $C,$ there is a $t_C>0$ such that 
$|h_b(t)|<1$
for all $b\in C$ and $t>t_C.$ As a consequence, for each $w\in \mathbb C$ with $\mathrm{Re}w>\delta,$ we have 
\begin{equation}\label{tinf}
    \bigg|\int_{t_C}^{+\infty} h_b(t)e^{-wt}dt\bigg|< \int_{t_C}^{+\infty} e^{-\delta t}dt=\frac{e^{-\delta t_C}}{\delta}.
\end{equation}
On the other hand, one can check that $t=0$ is a removable singularity of $h_b(t),$ and $h_b(t)$ is continuous on $[0,t_C].$ Then by the compactness  of $C\times [0,t_C],$ there is a constant $M_C>0$ such that $|h_b(t)|<M_C$
for all $b\in C$ and $t\in [0,t_C].$ As a consequence, for each $w\in \mathbb C$ with $\mathrm{Re}w>\delta,$ we have 
\begin{equation}\label{0t}
    \bigg|\int_0^{t_C} h_b(t)e^{-wt}dt\bigg|< M_C\int_0^{t_C} e^{-\delta t}dt=M_C\bigg(\frac{1}{\delta}-\frac{e^{-\delta t_C}}{\delta}\bigg).
\end{equation}
Putting (\ref{tinf}) and (\ref{0t}) together, and letting $B_2=\frac{e^{-\delta t_C}}{\delta}+M_C\big(\frac{1}{\delta}-\frac{e^{-\delta t_C}}{\delta}\big),$ we have for each $b\in C$ and $z\in T_{b,\delta}^+\cup T_{b,\delta}^-$ that 
\begin{equation}\label{B2}
\bigg|\int_0^{+\infty} h_b(t)e^{-zt}dt\bigg|<B_2\quad\text{and}\quad \bigg|\int_0^{+\infty} h_b(t)e^{-(Q-z)t}dt\bigg|<B_2.   
\end{equation}
\smallskip

Now we are left to estimate  $-\frac{1}{2}\big(z^2-Qz+\frac{Q^2+1}{6}\big)\big(\log z -\log (Q-z)\big) +\frac{Qz}{2}.$ To do this, let 
$N_C=2\max\big\{ |Q|\ \big|\ b\in C\big\}.$ For $z\in \mathbb C,$ let $\epsilon(z)$ be the sign of $\mathrm{Im}z,$ i.e., $\epsilon(z)=1$ if $\mathrm{Im}z>0,$ $\epsilon(z)=0$ if $\mathrm{Im}z=0,$ and $\epsilon(z)=-1$ if $\mathrm{Im}z<0.$ 

On the one hand, by the compactness of $\bigcup_{b\in C}[\delta,\mathrm{Re}Q-\delta]\times [-N_C,N_C],$ there is a $B_3>0$ depending only on $\delta$ and $N_C$ such that 
\begin{equation}\label{B3}
\mathrm{Re}\Bigg(-\frac{1}{2}\bigg(z^2-Qz+\frac{Q^2+1}{6}\bigg)\Big(\log z -\log (Q-z)\Big) +\frac{Qz}{2}\Bigg)< \frac{\epsilon(z)\pi}{2}\mathrm{Im}\Big(z(z-Q)\Big)+B_3
\end{equation}
for all $b\in C$ and $z\in T_{b,\delta}^+\cup T_{b,\delta}^-$ with $|\mathrm{Im}z|\leqslant N_C.$ 

On the other hand, for $z\in \mathbb C$ with $|\mathrm{Im}z|>N_C,$ we have $\big|\frac{Q}{z}\big|\leqslant \frac{N_C}{2|\mathrm{Im}z|}<\frac{1}{2}.$ As a consequence, if $\mathrm{Im}z>N_C,$
then 
$$\log z-\log (Q-z) = \pi\mathbf i -\log \bigg(1-\frac{Q}{z}\bigg);$$ and if $\mathrm{Im}z<-N_C,$
then 
$$\log z-\log (Q-z) = -\pi\mathbf i -\log \bigg(1-\frac{Q}{z}\bigg).$$
Hence
\begin{equation*}
\begin{split}
&-\frac{1}{2}\bigg(z^2-Qz+\frac{Q^2+1}{6}\bigg)\Big(\log z -\log (Q-z)\Big) +\frac{Qz}{2}\\
=& -\frac{\epsilon(z)\pi\mathbf i}{2}\bigg(z^2-Qz+\frac{Q^2+1}{6}\bigg)+\bigg(\frac{z^2}{2}\log\Big(1-\frac{Q}{z}\Big)+\frac{Qz}{2}\bigg)+\frac{1}{2}\bigg(Qz-\frac{Q^2+1}{6}\bigg)\log\Big(1-\frac{Q}{z}\Big);
\end{split}
\end{equation*}
and we estimate each of the terms on the right hand side as follows. For the first term, by the compactness of $C,$ there is a $B_4>0$ depending only on $C$ such that 
\begin{equation}\label{B4}
\mathrm{Re}\Bigg(-\frac{\epsilon(z)\pi\mathbf i}{2}\bigg(z^2-Qz+\frac{Q^2+1}{6}\bigg)\Bigg)<    \frac{\epsilon(z)\pi}{2}\mathrm{Im}\Big(z(z-Q)\Big)+B_4. 
\end{equation}
For the other two terms, by Taylor's Theorem, we have 
$$\log \bigg(1-\frac{Q}{z}\bigg)=-\frac{Q}{z} + R(Q,z)\frac{Q^2}{z^2}$$
for some analytic function $R(Q,z)$ of $\frac{Q}{z}.$
As a consequence,
$$\frac{z^2}{2}\log\Big(1-\frac{Q}{z}\Big)+\frac{Qz}{2}=R(Q,z)\frac{Q^2}{2},$$
and 
$$\frac{1}{2}\bigg(Qz-\frac{Q^2+1}{6}\bigg)\log\Big(1-\frac{Q}{z}\Big)= -\frac{Q^2}{2}+\frac{Q(Q^2+1)}{12z}+R(Q,z)\bigg(\frac{Q^3}{2z}-\frac{Q^2(Q^2+1)}{12z^2}\bigg).$$
As $\big|\frac{Q}{z}\big|<\frac{1}{2}$ is bounded, there is a constant $M>0$ such that $\big|R(Q,z)\big|<M;$ and as $|Q|<\frac{N_C}{2}$ for all $b\in C$ and $|z|>|\mathrm{Im}z|>N_C,$ there is a $B_5>0$ depending on $M,$ $\delta$ and $N_C$ such that 
\begin{equation}\label{B5}
   \bigg|\frac{z^2}{2}\log\Big(1-\frac{Q}{z}\Big)+\frac{Qz}{2}\bigg|<B_5\quad\text{and}\quad \bigg|\frac{1}{2}\bigg(Qz-\frac{Q^2+1}{6}\bigg)\log\Big(1-\frac{Q}{z}\Big)\bigg|<B_5. 
\end{equation}
\medskip

Finally, putting (\ref{Sbpsi}), (\ref{psi-psi}), (\ref{B1}), (\ref{B2}), (\ref{B3}), (\ref{B4}) and (\ref{B5}) together, and letting 
$$K=e^{B_1+2B_2+\max\{B_3, B_4+2B_5\}},$$
we have 
$$|S_b(z)|< K e^{\frac{\epsilon(z)\pi}{2}\mathrm{Im}\big(z(z-Q)\big)}$$
for all $b\in C$ and $z\in T_{b,\delta}^+\cup T_{b,\delta}^-.$ This completes the proof.
\end{proof}

\begin{proof}[Proof of Proposition \ref{lem: ab conver}]
We first show that the integral in the proposition converges absolutely. Let $L>0$ be sufficiently large and  $\delta>0$ be sufficiently small so that for  $u\in \Gamma$ with $\mathrm{Im}u \geqslant L,$ we have $u-t_i\in T^+_{b,\delta}$ for each $i\in\{1,2,3,4\}$ and $q_j-u\in T^-_{b,\delta}$ for each $j\in\{1,2,3,4\},$ and for $u\in \Gamma$ with $\mathrm{Im}u\leqslant -L,$ we have $u-t_i\in T^-_{b,\delta}$ for each $i\in\{1,2,3,4\}$ and $q_j-u\in T^+_{b,\delta}$ for each $j\in\{1,2,3,4\}.$ Let $\Gamma_+=\{ u\in\Gamma\ |\ \mathrm{Im}u\geqslant L\}$ and $\Gamma_-=\{ u\in\Gamma\ |\ \mathrm{Im}u\leqslant -L\},$  and  for the simplicity of the notations let
$$F(u)\doteq \prod_{i=1}^4S_b(u-t_i)\prod_{j=1}^4S_b(q_j-u)$$ 
be the integrand . Let $s\in (\max\{\mathrm{Re}t_i\},\min\{\mathrm{Re}q_j\})$ so that for each $u\in\Gamma_+,$
$u=s+\mathbf i t$
for some  $t\geqslant L;$ and 
for each $u\in\Gamma_-,$
$u=s-\mathbf i t$
for some $t\geqslant L.$ Then by Lemma \ref{Sb-bound} with $C=\{b\},$ the constant $K$ therein, and a direct computation, we have for $u \in \Gamma_+$ that  
\begin{equation*}
\big|F(u) \big|<  K ^ 8 e^{\sum_{i=1}^4\frac{\pi}{2}\mathrm{Im}\big((u-t_i)(u-t_i-Q)\big)-\sum_{j=1}^4\frac{\pi}{2}\mathrm{Im}\big((q_j-u)(q_j-u-Q)\big)} 
 = K^8Me^ {-2\pi\mathrm{Re}Qt},
\end{equation*}
and for $u\in \Gamma_-$ that 
\begin{equation*}
\big|F(u) \big|<  K ^ 8 e^{-\sum_{i=1}^4\frac{\pi}{2}\mathrm{Im}\big((u-t_i)(u-t_i-Q)\big)+\sum_{j=1}^4\frac{\pi}{2}\mathrm{Im}\big((q_j-u)(q_j-u-Q)\big)} 
 = K^8M^{-1}e^ {-2\pi \mathrm{Re}Qt},
\end{equation*}
where $M=e^{-2\pi s\mathrm{Im}Q+\frac{\pi}{2}\mathrm{Im}\big(\sum_{i=1}^4t_i^2-\sum_{j=1}^4q_j^2+(\sum_{i=1}^4t_i+\sum_{j=1}^4q_j)Q\big)}$ 
is a constant independent of $t.$ Since $\mathrm{Re}Q=\mathrm{Re}b+\mathrm{Re}(b^{-1})>0,$ we have
$$
\bigg|\int_{\Gamma_+\cup \Gamma_-}F(u)du \bigg|< K^8\big(M+M^{-1}\big) \int_L^{+\infty} e^{-2\pi \mathrm{Re}Q t}dt <+\infty.$$
As a consequence, the integral in the proposition converges absolutely . 
\medskip

Next, we show that the integral is independent of the choice of the vertical line $\Gamma.$ Suppose $\Gamma_1$ and $\Gamma_2$ are two vertical lines passing the interval $(\max\{\mathrm{Re}t_i\},\min\{\mathrm{Re}q_j\}).$ Let $L>0$ be sufficiently large, and for $k\in\{1,2\},$ let $\Gamma_{k,+}=\{ u\in\Gamma_k\ |\ \mathrm{Im}u\geqslant L\}$ and let $\Gamma_{k,-}=\{ u\in \Gamma_k\ |\ \mathrm{Im}u\leqslant -L\}.$ We also let $C_+$ be the piece of the horizontal line $\{u\in\mathbb C\ |\ \mathrm{Im}u=L\}$ lying in between $\Gamma_{1,+}$ and $\Gamma_{2,+},$ and let $C_-$ be the piece of the horizontal line $\{u\in\mathbb C\ |\ \mathrm{Im}u=-L\}$ lying in between $\Gamma_{1,-}$ and $\Gamma_{2,-}.$ Remember that the contours $\Gamma_{1,+},$ $\Gamma_{1,-},$ $\Gamma_{2,+},$ $\Gamma_{2,-},$ $C_+$ and $C_-$ all depend on $L.$ 
Let $\Gamma_+=\Gamma_{1,+}\cup C_+\cup \Gamma_{2,+}$ and $\Gamma_-=\Gamma_{1,-}\cup C_-\cup \Gamma_{2,-}$ both being suitably oriented.  
Then by the analyticity of $F$ and Cauchy's Theorem, we have 
$$\int_{\Gamma_1}F(u)du-\int_{\Gamma_2}F(u)du=\int_{\Gamma_+\cup \Gamma_-}F(u)du;$$
and to prove the result, it suffices to prove that 
$\big|\int_{\Gamma_+\cup\Gamma_-}F(u)du\big|\to 0$
as $L\to +\infty.$ As the integrals over $\Gamma_1$ and $\Gamma_2$ converge absolutely, the integrals over $\Gamma_{1,+},$ $\Gamma_{1,-},$ $\Gamma_{2,+}$ and $\Gamma_{2,-}$ converge to $0$
as $L\to +\infty,$
and we are left to show that 
$$\bigg|\int_{C_+\cup C_-}F(u)du\bigg|\to 0$$
 as $L\to +\infty.$ To this end, let  $\delta>0$ be sufficiently small so that for $u\in C_+,$ $u-t_i\in T^+_{b,\delta}$ for each $i\in\{1,2,3,4\}$ and $q_j-u\in T^-_{b,\delta}$ for each $j\in\{1,2,3,4\},$ and that for $u\in C_-,$ $u-t_i\in T^-_{b,\delta}$ for each $i\in\{1,2,3,4\}$ and $q_j-u\in T^+_{b,\delta}$ for each $j\in\{1,2,3,4\}.$ Then by Lemma \ref{Sb-bound} with $C=\{b\},$ the constant $K$ therein, and a direct computation, for $u\in C_+$ we have 
\begin{equation}\label{22.17}
\begin{split}
\big|
 F(u)\big|< & K ^ 8 e^{\sum_{i=1}^4\frac{\pi}{2}\mathrm{Im}\big((u-t_i)(u-t_i-Q)\big)-\sum_{j=1}^4\frac{\pi}{2}\mathrm{Im}\big((q_j-u)(q_j-u-Q)\big)}  =K^8M(u) e^ {-2\pi\mathrm{Re}QL};
 \end{split}
\end{equation}
and for $u\in C_-$ we have 
\begin{equation}\label{22.18}
\begin{split}
\big|
 F(u)\big|< & K ^ 8 e^{-\sum_{i=1}^4\frac{\pi}{2}\mathrm{Im}\big((u-t_i)(u-t_i-Q)\big)+\sum_{j=1}^4\frac{\pi}{2}\mathrm{Im}\big((q_j-u)(q_j-u-Q)\big)} 
 = K^8M(u)^{-1} e^ {-2\pi\mathrm{Re}QL},
 \end{split}
\end{equation}
where $M(u)=e^{-2\pi\mathrm{Re}u\mathrm{Im}Q+\frac{\pi}{2}\mathrm{Im}\big(\sum_{i=1}^4t_i^2-\sum_{j=1}^4q_j^2+(\sum_{i=1}^4t_i+\sum_{j=1}^4q_j)Q\big)}.$
 Now, as $C_\pm$ lie in between the two vertical contours $\Gamma_1$ and $\Gamma_2$ both intersecting the interval $(\max\{\mathrm{Re}t_i\},\min\{\mathrm{Re}q_j\}),$ we have $\max\{\mathrm{Re}t_i\}< \mathrm{Re}u< \min\{\mathrm{Re}q_j\}$ for each $u\in C_+\cup C_-.$ As a consequence, there is an $N>0$ such that 
 $K^8M(u)<N$ for each $u\in  C_+,$ and $K8M(u)^{-1}<N$
 for each $u\in C_-.$  Together with (\ref{22.17}), (\ref{22.18}), we have 
$$\big|F(u)\big|< N e^{-2\pi \mathrm{Re}QL}$$
for each $u\in C_+\cup C_-;$  and since $\mathrm{Re}Q>0,$ we have  
$$\bigg|\int_{C_+\cup C_-}F(u)du\bigg|< 2\big(\min\{\mathrm{Re}q_j\}-\max\{\mathrm{Re}t_i\}\big) N e^{-2\pi \mathrm{Re}QL},$$
which converges to $0$ as $L\to +\infty.$ This shows that the integral is independent of the choice of the vertical line $\Gamma.$
\medskip

Finally, we show that the integral depends analytically on $b.$ Let $C$ be any compact subset of $\big\{ b\in\mathbb C\ \big|\ \arg b\in(-\frac{\pi}{2},\frac{\pi}{2})\big\}$ that is small enough so that 
$$\max_{b\in C}\max\{\mathrm{Re}t_1,\mathrm{Re}t_2,\mathrm{Re}t_3,\mathrm{Re}t_4\}<\min_{b\in C}\min \{\mathrm{Re}q_1,\mathrm{Re}q_2,\mathrm{Re}q_3,\mathrm{Re}q_4\}.$$
Let 
$$\delta=\frac{1}{4}\Big(\min_{b\in C}\min \{\mathrm{Re}q_j\}-\max_{b\in C}\max\{\mathrm{Re}t_i\}\Big),$$
$$s\in \Big(\max_{b\in C}\max\{\mathrm{Re}t_i\}+\delta, \min_{b\in C}\min\{\mathrm{Re}q_j\}-\delta\Big),$$
and $\Gamma=\{ u\in\mathbb C\ |\ \mathrm{Re}u=s\}.$
Let $L>0$ be sufficiently large  so that for each $b\in C$ and $u\in \Gamma$ with $\mathrm{Im}u \geqslant L,$ we have $u-t_i\in T^+_{b,\delta}$ for each $i\in\{1,2,3,4\}$ and $q_j-u\in T^-_{b,\delta}$ for each $j\in\{1,2,3,4\},$ and for each $b\in C$ and $u\in \Gamma$ with $\mathrm{Im}u\leqslant -L,$ we have $u-t_i\in T^-_{b,\delta}$ for each $i\in\{1,2,3,4\}$ and $q_j-u\in T^+_{b,\delta}$ for each $j\in\{1,2,3,4\}.$  Let $\Gamma_+=\{ u\in\Gamma\ |\ \mathrm{Im}u\geqslant L\}$ and $\Gamma_-=\{ u\in\Gamma\ |\ \mathrm{Im}u\leqslant -L\}.$ Then by Lemma \ref{Sb-bound} with the constant $K$ therein and a direct computation, we have for each $b\in C$ and $u=s+\mathbf i t \in \Gamma_+$ that  
\begin{equation*}
\big|F(u) \big|<  K ^ 8 e^{\sum_{i=1}^4\frac{\pi}{2}\mathrm{Im}\big((u-t_i)(u-t_i-Q)\big)-\sum_{j=1}^4\frac{\pi}{2}\mathrm{Im}\big((q_j-u)(q_j-u-Q)\big)} 
 = K^8M(b)e^ {-2\pi \mathrm{Re}Qt},
\end{equation*}
and for each $b\in C$ and $u=s-\mathbf i t\in \Gamma_-$ that 
\begin{equation*}
\big|F(u) \big|<  K ^ 8 e^{-\sum_{i=1}^4\frac{\pi}{2}\mathrm{Im}\big((u-t_i)(u-t_i-Q)\big)+\sum_{j=1}^4\frac{\pi}{2}\mathrm{Im}\big((q_j-u)(q_j-u-Q)\big)} 
 = K^8M(b)^{-1}e^ {-2\pi  \mathrm{Re}Qt},
\end{equation*}
where $M(b)=e^{-2\pi s\mathrm{Im}Q+\frac{\pi}{2}\mathrm{Im}\big(\sum_{i=1}^4t_i^2-\sum_{j=1}^4q_j^2+(\sum_{i=1}^4t_i+\sum_{j=1}^4q_j)Q\big)}.$ Let $$M_{C,1}=K^8\max_{b\in C}\max\big\{M(b),M(b)^{-1}\big\}$$
and $N_C=\min_{b\in C}\mathrm{Re}Q.$ Then 
\begin{equation}\label{22.3}
    |F(u)|< M_{C,1}e^{-2\pi N_Ct}
\end{equation}
for all $b\in C$ and $u\in \Gamma_+\cup\Gamma_-.$ On the other hand, by the compactness of $C\times [-L,L],$
 there is an $M_{C,2}>0$ such that 
\begin{equation}\label{22.4}
    |F(u)|< M_{C,2}e^{-2\pi N_Ct}
\end{equation}
for all $b\in C$ and $u\in \Gamma$ with $|\mathrm{Im}u|\leqslant L.$ Putting (\ref{22.3}) and (\ref{22.4}) together and letting $M_C=\max\{M_{C,1},M_{C,2}\},$ we have
\begin{equation*}
    |F(u)|< M_Ce^{-2\pi N_Ct}
\end{equation*}
for all $b\in C$ and $u\in \Gamma.$ Let $g_C(u)=M_Ce^{-2\pi N_C|\mathrm{Im}u |}.$ As $N_C>0,$ the integral $|\int_\Gamma g_C(u)du|<+\infty;$ and by e.g. \cite[Theorem 2.43]{H}, $\int_\Gamma F(u)du$ is analytic in $b$ on $\big\{ b\in \mathbb C\ \big|\ \arg b\in (-\frac{\pi}{2},\frac{\pi}{2})\big\}.$ Finally, as $\prod_{i=1}^4\prod_{j=1}^4S_b(q_j-t_i)^{-\frac{1}{2}}$ is also analytic in $b,$ the integral in the proposition depends analytically on $b.$
\end{proof}

\subsection{Basic properties of complex $b$-$6j$ symbols}

From \eqref{b6j}, one sees immediately the following tetrahedral symmetry of the b-6j symbols.

\begin{proposition}[Tetrahedral symmetry]\label{tetra}
\begin{equation*}
    \begin{Bmatrix} 
      a_1 & a_2 & a_3 \\
      a_4 & a_5 & a_6
   \end{Bmatrix}_b=\begin{Bmatrix} 
      a_2 & a_1 & a_3 \\
      a_5 & a_4 & a_6
   \end{Bmatrix}_b=\begin{Bmatrix} 
      a_1 & a_3 & a_2 \\
      a_4 & a_6 & a_5
   \end{Bmatrix}_b=\begin{Bmatrix} 
      a_1 & a_5 & a_6 \\
      a_4 & a_2 & a_3
   \end{Bmatrix}_b.
\end{equation*}
\end{proposition}

We also have the following  reflection symmetry of the $b$-$6j$ symbols. 

\begin{proposition}[Reflection symmetry]\label{reflection}
    If $(a_1,\dots,a_6)$ is $b$-admissible, then $(Q-a_1,a_2,\dots,a_6)$ is also $b$-admissible, and   
$$\bigg\{\begin{matrix} a_1 & a_2 & a_3 \\ a_4 & a_5 & a_6 \end{matrix} \bigg\}_b=\bigg\{\begin{matrix} Q-a_1 & a_2 & a_3 \\ a_4 & a_5 & a_6 \end{matrix} \bigg\}_b.$$
Moreover, the same result holds if we change $a_k\mapsto Q-a_k$ for any $k\in\{1,\dots,6\}.$
\end{proposition}

\begin{proof} The $b$-admissibility part follows verbatim that of \cite[Proposition 3.4]{LMSWY2}. The reflection symmetry is proved in \cite[Proposition 3.4]{LMSWY2}  for all real $b,$ and the case for a complex $b$ follows immediately from the analytic dependence on $b$ stated in Proposition \ref{lem: ab conver}. 
\end{proof}

The following Proposition \ref{conjugate} states the behavior of the complex $b$-$6j$ symbols under the simultaneous complex conjugation of  the six-tuple $(a_1,\dots,a_6)$ and the index $b.$

\begin{proposition}[Complex conjugation]\label{conjugate}
$(\overline{a_1},\dots,\overline{a_6})$ is $\overline b$-admissible if and only if $(a_1,\dots,a_6)$ is $b$-admissible, and 
    \begin{equation*}
\bigg\{\begin{matrix} \overline {a_1} & \overline {a_2} & \overline {a_3} \\ \overline {a_4} & \overline {a_5} & \overline {a_6} \end{matrix} \bigg\}_{\overline b} = 
\overline{\bigg\{\begin{matrix} a_1 & a_2 & a_3 \\ a_4 & a_5 & a_6 \end{matrix} \bigg\}_b}.
\end{equation*}
\end{proposition}

\begin{proof}
 For $i,j\in\{1,2,3,4\},$ $0<\mathrm{Re}(\overline{q_j}-\overline{t_i})<\mathrm{Re}\overline Q$ 
if and only if $0<\mathrm{Re}(q_j-t_i)<\mathrm{Re}Q,$ hence $(\overline{a_1},\dots,\overline{a_6})$ is $\overline b$-admissible  if and only if $(a_1,\dots,a_6)$ is $b$-admissible.

To show the equality, we first claim that for   $b\in\mathbb C$ with $\arg b\in (-\frac{\pi}{2},\frac{\pi}{2})$ and $z\in \mathbb C$ that is not a pole of $S_b(z),$
\begin{equation}\label{conj}
S_{\overline b}(\overline z)=\overline{S_b(z)}.
\end{equation}
Indeed, let $F_{z,b}(t)=\frac{\sinh\big((\frac{Q}{2}-z)t\big)}{4t\sinh(\frac{bt}{2})\sinh(\frac{t}{2b})}$ be the integrand  in (\ref{eq:def-S}). Then we have $F_{\overline z,\overline b}(\overline t)=\overline{F_{z,b}(t)}.$ As a consequence, $\int_\Omega F_{\overline z,\overline b}(t)dt=\int_{\overline \Omega} F_{\overline z,\overline b}(\overline s)d\overline s=\int_{\overline \Omega} \overline{F_{z,b}(s)}d\overline s =\overline {\int_\Omega F_{z,b}(s)ds},$ where  the first equality comes from the change of variable $t=\overline s;$ and $S_{\overline b}(\overline z)=\exp\big(\int_\Omega F_{\overline z,\overline b}(t)dt\big)=\exp\big(\overline {\int_\Omega F_{z,b}(s)ds}\big)=\overline{\exp\big(\int_\Omega F_{z,b}(s)ds\big)}=\overline{S_b(z)},$ where again the second equality comes from the change of variable $t=\overline s.$

Then for $\boldsymbol a=(a_1,\dots,a_6),$ we let  $G_{\boldsymbol a,b}(u)=\prod_{i=1}^4S_b(u-t_i)\prod_{j=1}^4S_b(q_j-u)$ be the integrand  in (\ref{b6j}). By (\ref{conj}), we have $G_{\overline{\boldsymbol a},\overline b}(\overline u)=\overline{G_{\boldsymbol a,b}(u)}.$ Therefore, $\int_\Gamma G_{\overline{\boldsymbol a},\overline b}(u)du=\int_{\overline \Gamma} G_{\overline{\boldsymbol a},\overline b}(\overline v)d\overline v=\int_{\overline \Gamma} \overline{G_{\boldsymbol a,b}(v)} d\overline v=\overline {\int_\Gamma G_{\boldsymbol a,b}(v)dv},$ where the first equality comes from the change of variable $u=\overline v;$ and the desired equality follows immediately.
\end{proof}

\subsection{Complex $b$-$6j$ symbols parametrized by the dihedral angles}

In this subsection, we focus on the case that $a_k=\frac{Q}{2}\pm\frac{\theta_k}{2\pi b}$ for each $k\in\{1,\dots,6\}$  with the sign $+$ or $-$ chosen arbitrarily, and with  $(\theta_1,\dots,\theta_6)$ the dihedral angles of a hyperideal hyperbolic tetrahedron. In this case,  the following Proposition \ref{contour}  gives us a flexibility to choose the contour of integral in the computation of the $b$-$6j$ symbols, which is a crucial step in the proof of Theorem \ref{vol}. 

To state the result, we first recall that the double sine function $S_b$ satisfies the functional equation (see e.g \cite[A.15]{TesVar})
\begin{equation}\label{FE1}
S_b(z+b^{\pm 1})=2\sin (\pi b ^{\pm 1} z)S_b(z),
\end{equation}
and extends to a meromorphic function on $\mathbb C$ with the set of poles $\{ -nb-mb^{-1} \ |\ m, n\in \mathbb Z_{\geqslant 0}\}$ and the set of  zeros  $\{ Q + nb + mb^{-1} \ |\ m, n\in \mathbb Z_{\geqslant 0}\}.$ See Figure \ref{pole}. 
\begin{figure}[htbp]
\centering
\includegraphics[scale=0.5]{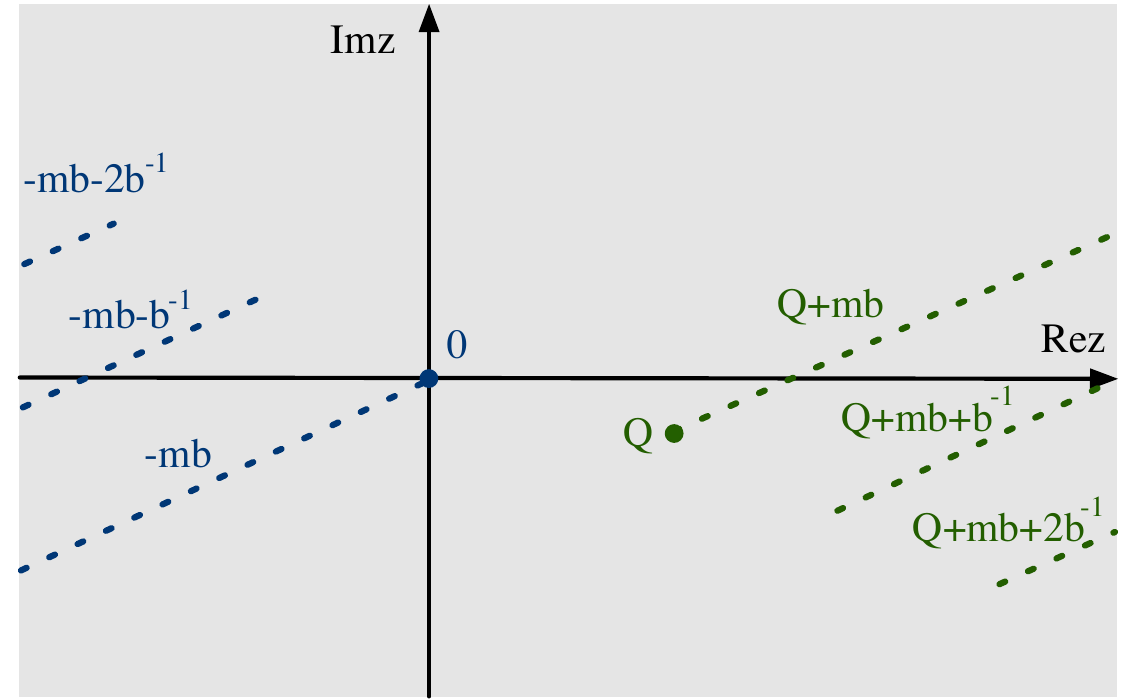}
\caption{The poles and zeros of $S_b(z).$}
\label{pole}
\end{figure}

\begin{proposition}\label{contour} Let $(\theta_1,\dots,\theta_6)$ be the dihedral angles of a hyperideal hyperbolic tetrahedron. For each $k\in\{1,\dots,6\},$ let $a_k=\frac{Q}{2}\pm\frac{\theta_k}{2\pi b}$ with the sign $+$ or $-$ chosen arbitrarily.
Let $t_{\text{max}}\in \{t_1,t_2,t_3,t_4\}$ be the one with  $\mathrm{Re}t_{\text{max}}=\max\{\mathrm{Re}t_i,\mathrm{Re}t_2,\mathrm{Re}t_3,\mathrm{Re}t_4\}$ and let $q_{\text{min}}\in \{q_1,q_2,q_3,q_4\}$ be the one with $\mathrm{Re}q_{\text{min}}=\min\{\mathrm{Re}q_1,\mathrm{Re}q_2,\mathrm{Re}q_3,\mathrm{Re}q_4\}.$ Let 
$\Gamma^*$ be a contour as depicted in (the shaded region of) Figure \ref{Gamma*}, i.e.,
\begin{enumerate}[(1)]
\item for each $u\in\Gamma^*,$ 
$$\arg(u-t_{\text{max}})\notin[\pi-\arg b,\pi+\arg b]\quad\text{and}\quad\arg(u-q_{\text{min}})\notin [-\arg b,\arg b],$$
and 
\item one end of $\Gamma^*$ is a ray of argument in $[\arg b, \pi - \arg b],$ and the other end of $\Gamma^*$ is a ray of argument  in $[\arg b-\pi,-\arg b].$ 
\end{enumerate}
Then for $b\in \mathbb C$ with $\arg b\in \big(0,\frac{\pi}{2}\big),$ the $b$-$6j$ symbol can be computed by
\begin{equation*}
\bigg\{\begin{matrix} a_1 & a_2 & a_3 \\ a_4 & a_5 & a_6 \end{matrix} \bigg\}_b=\Bigg(\frac{1}{\prod_{i=1}^4\prod_{j=1}^4S_b(q_j-t_i)}\Bigg)^{\frac{1}{2}}\int_{\Gamma^*} \prod_{i=1}^4S_b(u-t_i)\prod_{j=1}^4S_b(q_j-u)du.
\end{equation*} 
\end{proposition}

\begin{figure}[htbp]
\centering
\includegraphics[scale=0.4]{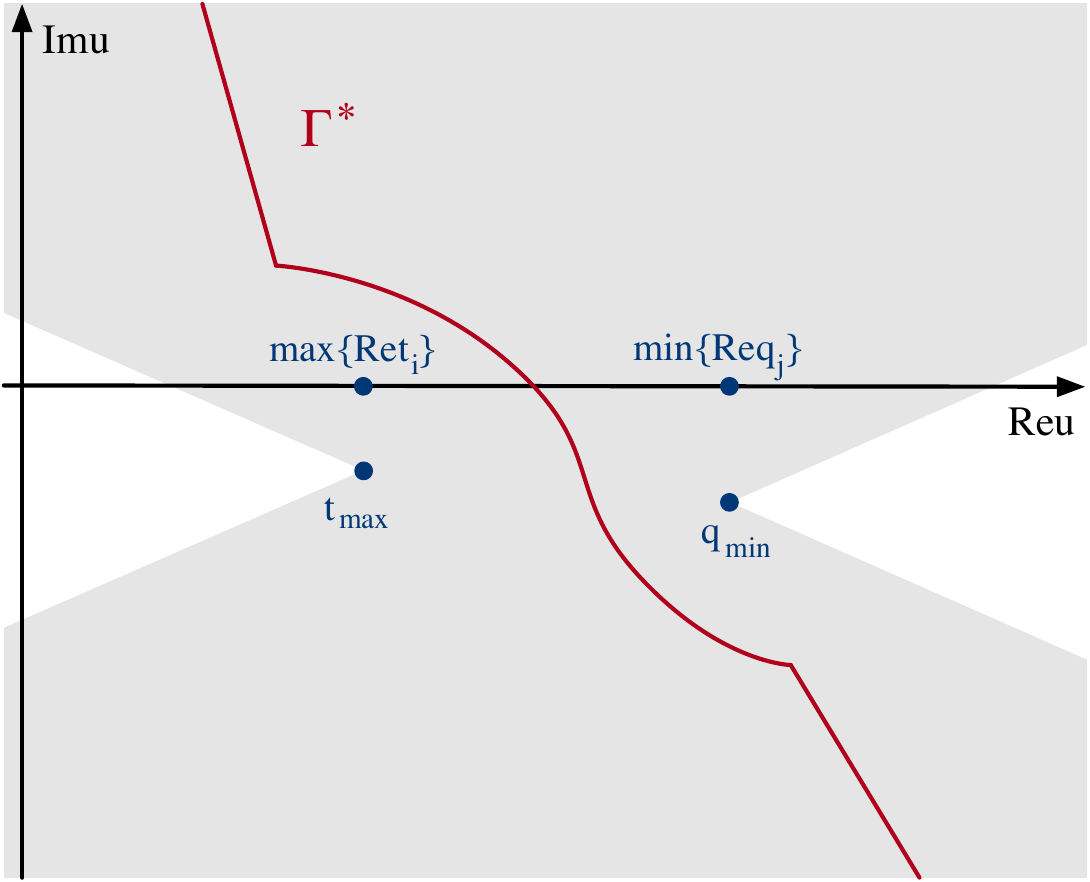}
\caption{The contour $\Gamma^*.$}
\label{Gamma*}
\end{figure}

The proof of Proposition \ref{contour} relies on the relationship between the double sine function and \emph{Faddeev's quantum dilogarithm function}
$$\Phi_b(z)\doteq\exp\bigg(\int_\Omega \frac{e^{-2\mathbf izt}}{4t\sinh(b t)\sinh\big(\frac{t}{b}\big)}dt\bigg)$$
for $b\in\mathbb C$ with $\arg b\in (0,\frac{\pi}{2}).$
The integral absolutely converges for $z\in \mathbb C$ with $-\mathrm{R e}\big(\frac{Q}{2}\big)<\mathrm{Im}z<\mathrm {Re} \big(\frac{Q}{2}\big),$ defining a holomorphic function in such $z;$ and by the functional equation \cite[Formula (48)]{AK}
\begin{equation}\label{phieq}
    \Phi_b\bigg(z-\frac{b^{\pm 1}\mathbf i }{2}\bigg)=\Big(1+e^{2\pi b^{\pm 1}z}\Big)\Phi_b\bigg(z+\frac{b^{\pm 1}\mathbf i }{2}\bigg),
\end{equation}
$\Phi_b(z)$ extends to a meromorphic function on $\mathbb C$ with the set of poles $\big\{\frac{Q}{2}\mathbf i+nb\mathbf i+mb^{-1}\mathbf i \ \big|\ m, n\in \mathbb Z_{\geqslant 0}\big\}$ and the set of  zeros $\big\{-\frac{Q}{2}\mathbf i-nb\mathbf i-mb^{-1}\mathbf i \ \big|\ m, n\in \mathbb Z_{\geqslant 0}\big\}.$ By \cite[A4, A5, A13]{TesVar}, Faddeev's quantum dilogarithm function $\Phi_b$ and the double sine function $S_b$ are related by \begin{equation}\label{PhiS}
    \Phi_b(z)=S_b\bigg(\mathbf iz+\frac{Q}{2}\bigg)e^{\frac{\pi \mathbf i z^2}{2}+\frac{\pi \mathbf i}{24}(b^2+b^{-2})}.
\end{equation}
We need the following infinite product formula for Faddeev's quantum dilogarithm function.

\begin{proposition}\label{Prod}
For $b\in \mathbb C$ with $\arg b\in (0,\frac{\pi}{2}),$ 
$$\Phi_b(z)=\frac{\prod_{n=0}^{+\infty}\Big(1+e^{2\pi b z+(2n+1)\pi \mathbf i b^2}\Big)}{\prod_{n=0}^{+\infty}\Big(1+e^{2\pi b^{-1} z-(2n+1)\pi \mathbf i b^{-2}}\Big)}$$
for each $z\in \mathbb C$ that is not a pole of $\Phi_b(z).$ 
\end{proposition}

\begin{proof} Observe that for $b\in\mathbb C$ with $\arg b\in(0,\frac{\pi}{2}),$ we have $\mathrm{Im}(b^2)>0.$
Then by \cite[Formula (44)]{AK}, the formula holds for all $z\in \mathbb C$ with $-\mathrm{Im}\big(\frac{Q}{2}\big)<\mathrm{Im}z<\mathrm {Im} \big(\frac{Q}{2}\big);$ and to prove the formula for other $z\in\mathbb C,$ it suffices to show that the right hand side of the equality defines a meromorphic function on $\mathbb C$ that shares the same set of poles with $\Phi_b.$

To this end, we let $N_b(z)=\prod_{n=0}^{+\infty}\big(1+e^{2\pi b z+(2n+1)\pi \mathbf i b^2}\big)$ be the numerator and $M_b(z)=\prod_{n=0}^{+\infty}\big(1+e^{2\pi b^{-1} z-(2n+1)\pi \mathbf i b^{-2}}\big)$ be the denominator of the right hand side of the equality. 

We first claim that both $N_b(z)$ and $M_b(z)$ define holomorphic functions on $\mathbb C,$ as a consequence of which $\frac{N_b(z)}{M_b(z)}$ is a meromorphic function on $\mathbb C.$ Indeed, for each  $z\in \mathbb C,$ we have 
$\sum_{n=0}^{+\infty}\big|e^{2\pi b z+(2n+1)\pi \mathbf i b^{2}}\big|\leqslant\big(\sum_{n=0}^{+\infty}e^{-(2n+1)\pi\mathrm{Im}(b^2)}\big) e^{2\pi |b||z|}$
and 
$\sum_{n=0}^{+\infty}\big|e^{2\pi b^{-1} z-(2n+1)\pi \mathbf i b^{-2}}\big|\leqslant \big(\sum_{n=0}^{+\infty}e^{(2n+1)\pi\mathrm{Im}(b^{-2})}\big)e^{2\pi |b^{-1}||z|}.$ 
Since $\mathrm{Im}(b^2)>0$ and $\mathrm{Im}(b^{-2})<0,$ both of the series above converge uniformly on the disk $|z|\leqslant r$ for each $r>0.$ Then by e.g. \cite[Theorem on p. 437]{K}, both of the infinite products $N_b(z)$ and $M_b(z)$ converge, and define  holomorphic functions on $\mathbb C.$ 

Next, we compute the set of poles of $\frac{N_b(z)}{M_b(z)}.$
To see this, we first observe that the set of zeros of
$N_b(z)$ is 
$$Z_{N_b}=\bigg\{-\Big(n+\frac{1}{2}\Big)b\mathbf i + \Big(m+\frac{1}{2}\Big)b^{-1}\mathbf i \ \bigg|\ n\in \mathbb Z_{\geqslant 0}, m\in \mathbb Z\bigg\}$$
and the set of zeros of $M_b(z)$ is 
$$Z_{M_b}=\bigg\{\Big(n+\frac{1}{2}\Big)b\mathbf i + \Big(m+\frac{1}{2}\Big)b^{-1}\mathbf i \ \bigg|\ n\in \mathbb Z, m\in \mathbb Z_{\geqslant 0}\bigg\};$$
and all these zeros are simple zeros. 
As a consequence, all the points in $Z_N\cap Z_M$ are removable singularities of $\frac{N_b(z)}{M_b(z)},$
and hence the set of poles of $\frac{N_b(z)}{M_b(z)}$
 is 
 $$ Z_{M_b}\setminus Z_{N_b}=\bigg\{\Big(n+\frac{1}{2}\Big)b\mathbf i + \Big(m+\frac{1}{2}\Big)b^{-1}\mathbf i \ \bigg|\ m,n\in \mathbb Z_{\geqslant 0}\bigg\}.$$ 

Finally, the set of poles of $\Phi_b$ is $\big\{\frac{Q}{2}\mathbf i+nb\mathbf i+mb^{-1}\mathbf i \ \big|\ m, n\in \mathbb Z_{\geqslant 0}\big\},$ 
which coincides with $Z_{M_b}\setminus Z_{N_b}.$ This completes the proof. 
\end{proof}

Using Proposition \ref{Prod}, we obtain the following estimate of the double sine function near infinity, which is the key ingredient of the proof of Proposition \ref{contour}. A similar  result can also be found in \cite[Formula (A.11)]{TesVar}.

\begin{lemma}\label{Sb-bound} For $b\in \mathbb C$ with $\arg b\in (0,\frac{\pi}{2})$ and $\delta>0$ sufficiently small, as depicted in Figure \ref{R}, let 
$$R^+_\delta = \Big\{ s Q + t e^{\mathbf i\phi}\ \Big|\  s\in (\delta, 1-\delta),\ \ t\geqslant 0\ \ \text{and}\ \ \phi\in [\arg b, \pi-\arg b]\Big\},$$
and 
$$R^-_\delta = \Big\{ s Q - t e^{\mathbf i\phi}\ \Big|\  s\in (\delta, 1-\delta),\ \ t\geqslant 0\ \ \text{and}\ \ \phi\in [\arg b, \pi-\arg b]\Big\}.$$
Then there exists a constant $K=K_{b,\delta}>0$  such that 
$$|S_b(z)| < Ke^{\frac{\pi}{2}\mathrm{Im}\big(z(z-Q)\big)}$$
for all $z\in R^+_\delta,$ and 
$$|S_b(z)|< Ke^{-\frac{\pi}{2}\mathrm{Im}\big(z(z-Q)\big)}$$
for all $z\in R^-_\delta.$
\end{lemma}

\begin{figure}[htbp]
\centering
\includegraphics[scale=0.4]{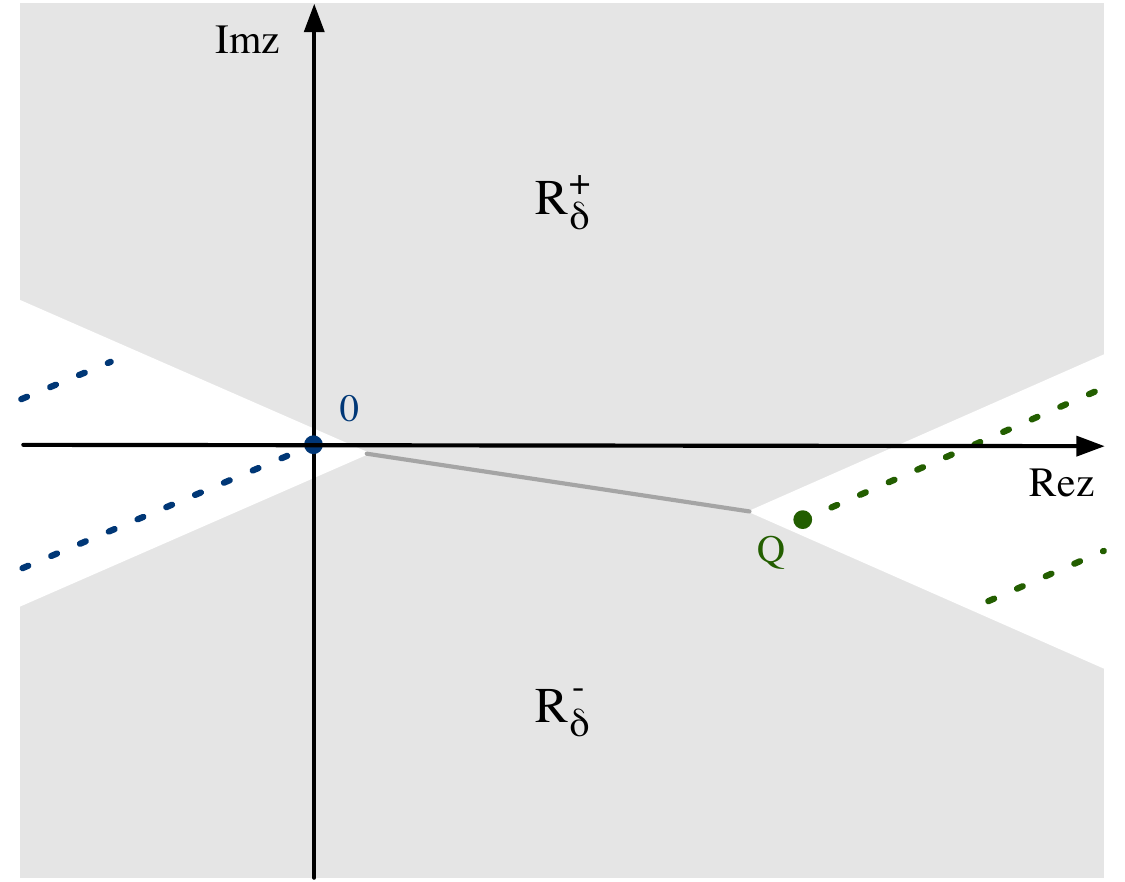}
\caption{The regions $R^+_\delta$ and $R^-_\delta.$}
\label{R}
\end{figure}

\begin{proof}
By (\ref{PhiS}) and Proposition \ref{Prod}, we have 
\begin{equation}\label{SProd}
S_b(z) = e^{\frac{\pi\mathbf i}{2}z(z-Q)+\frac{\pi\mathbf i}{12}(Q^2+1)}\frac{\prod_{n=0}^{+\infty}\Big(1-e^{-2\pi\mathbf i bz+2(n+1)\pi\mathbf i b^2}\Big)}{\prod_{n=0}^{+\infty}\Big(1-e^{-2\pi\mathbf i b^{-1}z-2n\pi\mathbf i b^{-2}}\Big)}
\end{equation}
for each $z\in \mathbb C$ that is not a pole of $S_b(z).$ We also notice that none of the poles nor zeros of $S_b$ lies in the regions $R^+_\delta$ and $R^-_\delta,$ as shown in Figure \ref{R}.

First of all, we  have 
\begin{equation}\label{2.4}
 \Big| e^{\frac{\pi\mathbf i}{2}z(z-Q)+\frac{\pi \mathbf i}{12}(Q^2+1)}\Big|=\Big|e^{\frac{\pi\mathbf i}{12}(Q^2+1)}\Big|e^{-\frac{\pi}{2}\mathrm{Im}\big(z(z-Q)\big)}
\end{equation}
for each $z\in\mathbb C.$ 

Next, for each $z=sQ-te^{\mathbf i\phi}\in R^-_\delta,$ we have
\begin{equation*}
\begin{split}
 \mathrm{Re}\Big(-2\pi\mathbf i bz+2(n+1)\pi\mathbf i b^2\Big)=&-2\pi\Big((n+1-s)\mathrm{Im}\big(b^2\big)+|b|t\sin\big(\phi+\arg b\big)\Big)\\
 <& -2\pi (n+\delta) \mathrm{Im}\big(b^2\big),
\end{split}
\end{equation*}
where the inequality comes from that $s\in (\delta,1-\delta)$ so that $1-s>\delta,$ $\mathrm{Im}(b^2)>0,$ $t\geqslant 0,$ and $\phi \in [\arg b,\pi-\arg b]$ so that $\sin(\phi+\arg b)\geqslant 0.$ As a consequence, we have 
\begin{equation*}
\sum_{n=0}^{+\infty}\Big|e^{-2\pi\mathbf i bz+2(n+1)\pi\mathbf i b^2}\Big|< \sum_{n=0}^{+\infty}e^{-2\pi (n+\delta) \mathrm{Im}(b^2)}<+\infty,
\end{equation*}
where the last inequality comes from that $\mathrm{Im}\big(b^2\big)>0.$ Then by \cite[Theorem 3 and Theorem 4 on pp. 219-220]{K} and the fact that $e^{-2\pi (n+\delta)\mathrm{Im}(b^2)}<1$ for all $n\geqslant 0,$ the infinity product $\prod_{n=0}^{+\infty} \big(1+e^{-2\pi (n+\delta)\mathrm{Im}(b^2)}\big)$ converges, the infinity product $\prod_{n=0}^{+\infty} \big(1-e^{-2\pi (n+\delta)\mathrm{Im}(b^2)}\big)$ converges to a nonzero number, and 
\begin{equation}\label{2.5}
    \prod_{n=0}^{+\infty} \Big(1-e^{-2\pi (n+\delta)\mathrm{Im}(b^2)}\Big)<\bigg|\prod_{n=0}^{+\infty}\Big(1-e^{-2\pi\mathbf i bz+2(n+1)\pi\mathbf i b^2}\Big)\bigg|<\prod_{n=0}^{+\infty} \Big(1+e^{-2\pi (n+\delta)\mathrm{Im}(b^2)}\Big).
\end{equation}
Similarly, for each $z=sQ-te^{\mathbf i\phi}\in R^-_\delta,$ we also have 
\begin{equation*}
\begin{split}
 \mathrm{Re}\Big(-2\pi\mathbf i b^{-1}z-2n\pi\mathbf i b^{-2}\Big)=&2\pi\Big((n+s)\mathrm{Im}\big(b^{-2}\big)-|b|^{-1}t\sin\big(\phi-\arg b\big)\Big)\\
 <& 2\pi (n+\delta) \mathrm{Im}\big(b^{-2}\big),
\end{split}
\end{equation*}
where the inequality comes from that $s>\delta,$ $\mathrm{Im}\big(b^{-2}\big)<0,$ $t\geqslant 0,$ and $\phi \in [\arg b,\pi-\arg b]$ so that  $\sin(\phi-\arg b)\geqslant 0.$ As a consequence, we have 
\begin{equation*}
\sum_{n=0}^{+\infty}\Big|e^{-2\pi\mathbf i b^{-1}z-2n\pi\mathbf i b^{-2}}\Big|< \sum_{n=0}^{+\infty}e^{2\pi (n+\delta) \mathrm{Im}(b^{-2})}<+\infty,
\end{equation*}
where the last inequality comes from that $\mathrm{Im}\big(b^{-2}\big)<0.$  Then by \cite[Theorem 3 and Theorem 4 on pp. 219-220]{K} and the fact that $e^{2\pi (n+\delta)\mathrm{Im}(b^{-2})}<1$ for each $n\geqslant 0,$ the infinity product $\prod_{n=0}^{+\infty} \big(1+e^{2\pi (n+\delta)\mathrm{Im}(b^{-2})}\big)$ converges, the infinity product $\prod_{n=0}^{+\infty} \big(1-e^{2\pi (n+\delta)\mathrm{Im}(b^{-2})}\big)$ converges to a nonzero number, and 
\begin{equation}\label{2.6}
    \prod_{n=0}^{+\infty} \Big(1-e^{2\pi (n+\delta)\mathrm{Im}(b^{-2})}\Big)<\bigg|\prod_{n=0}^{+\infty}\Big(1-e^{-2\pi\mathbf i b^{-1}z-2n\pi\mathbf i b^{-2}}\Big)\bigg|<\prod_{n=0}^{+\infty} \Big(1+e^{2\pi (n+\delta)\mathrm{Im}(b^{-2})}\Big).
\end{equation}
Putting (\ref{SProd}), (\ref{2.4}), (\ref{2.5}) and (\ref{2.6}) together and letting 
$$K_1=\Big|e^{-\frac{\pi\mathbf i}{12}(Q^2+1)}\Big|\frac{\prod_{n=0}^{+\infty} \Big(1+e^{2\pi (n+\delta)\mathrm{Im}(b^{-2})}\Big)}{\prod_{n=0}^{+\infty} \Big(1-e^{-2\pi (n+\delta)\mathrm{Im}(b^2)}\Big)}$$
and
$$K_2=\Big|e^{\frac{\pi\mathbf i}{12}(Q^2+1)}\Big|\frac{\prod_{n=0}^{+\infty} \Big(1+e^{-2\pi (n+\delta)\mathrm{Im}(b^2)}\Big)}{\prod_{n=0}^{+\infty} \Big(1-e^{2\pi (n+\delta)\mathrm{Im}(b^{-2})}\Big)}$$
which are constants depending only on $b$ and $\delta,$
we have 
\begin{equation}\label{2.8}
    \frac{1}{K_1}e^{-\frac{\pi}{2}\mathrm{Im}\big(z(z-Q)\big)}<|S_b(z)| < K_2e^{-\frac{\pi}{2}\mathrm{Im}\big(z(z-Q)\big)}
    \end{equation}
for all $z\in R^-_\delta.$

Now for $z\in R^+_\delta,$ we have $Q-z\in R^-_\delta.$ Then by (\ref{2.8}) and  the fact that $S_b(z)=\frac{1}{S_b(Q-z)}$ for each $z\in\mathbb C$ that is not a pole nor a zero of $S_b(z)$ (see e.g. \cite[Formula (A.16)]{TesVar}), we have 
\begin{equation}\label{2.88}
   \frac{1}{K_2}e^{\frac{\pi}{2}\mathrm{Im}\big(z(z-Q)\big)}<|S_b(z)|< K_1e^{\frac{\pi}{2}\mathrm{Im}\big(z(z-Q)\big)}.
    \end{equation}

Finally, the result follows immediately from (\ref{2.8}) and (\ref{2.88}) with  $K=\max\{K_1,K_2\}.$ 
\end{proof}

\begin{lemma}\label{in}
    Under the assumption of Proposition \ref{contour}, there exists a $\delta>0$ such that for each $u\in\Gamma^*$ and for each $i$ and $j$ in $\{1,2,3,4\},$
$$u-t_i\in R^+_\delta\cup R^-_\delta\quad\text{and}\quad q_j-u\in R^+_\delta\cup R^-_\delta,$$
where $R^+_\delta$ and $R^-_\delta$ are the regions defined in Lemma \ref{Sb-bound}. 
\end{lemma}

\begin{proof} For each $z\in\mathbb C,$ we let 
$$C^+(z)=\Big\{ z + sb+tb^{-1}\ \Big|\ s,t\in\mathbb R_{\geqslant 0}\Big\}\quad\text{and}\quad C^-(z)=\Big\{ z - sb - tb^{-1}\ \Big|\ s,t\in\mathbb R_{\geqslant 0}\Big\}.$$
Then as shown in Figure \ref{R}, the region $R^+_\delta\cup R^-_\delta$ is a subset of $\mathbb C\setminus \big(C^-(0)\cup C^+(Q)\big);$ and as shown in Figure \ref{Gamma*},  the contour $\Gamma^*$ lies in the region $\mathbb C\setminus  \big(C^-(t_\text{max})\cup C^+(q_\text{min})\big).$ 

We first claim that for each $i$ and $j$ in $\{1,2,3,4\},$ 
$$C^-(q_j-Q)\subset C^-(t_i)\subset C^-(t_\text{max})\quad\text{and}\quad C^+(t_i+Q)\subset C^+(q_j)\subset C^+(q_\text{min}),$$
as a direct consequence of which we have for each $u\in \mathbb C\setminus \big( C^-(t_\text{max})\cup C^+(q_\text{min}\big))$ that 
$$u-t_i\in \mathbb C\setminus \big(C^-(0)\cup C^+(Q)\big)\quad\text{and}\quad q_j-u\in \mathbb C\setminus \big(C^-(0)\cup C^+(Q)\big).$$
To prove the claim, we let 
\begin{equation}\label{2.9}
    t_i'=\bigg(t_i-\frac{3Q}{2}\bigg)b
    \end{equation}
for each $i\in\{1,2,3,4\},$
and let 
\begin{equation}\label{2.10}
    q_j'=\big(q_j-2Q\big)b
\end{equation}
for each $j\in\{1,2,3,4\}.$
Then 
\begin{equation*}
\begin{split}
&t_1'=\frac{\pm\theta_1\pm\theta_2\pm\theta_3}{2\pi},\quad t_2'=\frac{\pm\theta_1\pm\theta_5\pm\theta_6}{2\pi},\quad
t_3'=\frac{\pm\theta_2\pm\theta_4\pm\theta_6}{2\pi},\quad
t_4'=\frac{\pm\theta_3\pm\theta_4\pm\theta_5}{2\pi},\\
&q_1'=\frac{\pm\theta_1\pm\theta_2\pm\theta_4\pm\theta_5}{2\pi},\quad q_2'=\frac{\pm\theta_1\pm\theta_3\pm\theta_4\pm\theta_6}{2\pi},\quad 
q_3'=\frac{\pm\theta_2\pm\theta_3\pm\theta_5\pm\theta_6}{2\pi}\quad\text{and}\quad q_4'=0,
\end{split}
\end{equation*}
where the signs $\pm$ above are determined by the choice of the signs $\pm$ in $a_1,\dots,a_6$ in Proposition \ref{contour}.
In particular, $t_i'\in\mathbb R$ for each $i\in\{1,2,3,4\},$  and $q'_j\in\mathbb R$ for each $j\in\{1,2,3,4\}.$  Moreover, letting $t_\text{max}'=\big(t_\text{max}-\frac{3Q}{2}\big)b$ and $q_\text{min}'=(q_\text{min}-2Q)b,$ then we have
\begin{equation}\label{2.11}
t_\text{max}'=\max\{t_1',t_2',t_3',t_4'\}\quad\text{and}\quad q_\text{min}'=\min\{q_1',q_2',q_3',q_4'\};
\end{equation}
and by \cite{BB}, as the dihedral angles of a hyperbolic hyperideal tetrahedron satisfy 
$0<\theta_r+\theta_s+\theta_t<\pi$
for the triples $\{r,s,t\}=\{1,2,3\}, \{1,5,6\}, \{2,4,6\}, \{3,4,5\},$
we have
\begin{equation}\label{2.12}
   q_j'-t_i'<\frac{1}{2}
\end{equation}
for all $i,j\in\{1,2,3,4\}.$ Now for each $i\in\{1,2,3,4\},$ by (\ref{2.11}), we have 
$$t_i=t_\text{max}-(t_\text{max}'-t_i')b^{-1}\in C^-(t_\text{max});$$
and for each $j\in\{1,2,3,4\},$ by (\ref{2.12}) we have 
$$q_j-Q=t_i-\frac{1}{2}b-\Big(\frac{1}{2}-q_j'+t_i'\Big)b^{-1}\in C^-(t_i)$$
for each $i\in\{1,2,3,4\}.$ As a consequence, we have $C^-(q_j-Q)\subset C^-(t_i)\subset C^-(t_\text{max}).$ Similarly, for each $j\in\{1,2,3,4\},$ by (\ref{2.11}), we have 
$$q_j=q_\text{min}+(q_j'-q_\text{min}')b^{-1}\in C^+(q_\text{min});$$
and for each $i\in\{1,2,3,4\},$ by (\ref{2.12}) we have 
$$t_i+Q=q_j+\frac{1}{2}b+\Big(\frac{1}{2}-q_j'+t_i'\Big)b^{-1}\in C^+(q_j)$$
for each $j\in\{1,2,3,4\}.$ As a consequence, we have $C^+(t_i+Q)\subset C^+(q_j)\subset C^+(q_\text{min}).$ This completes the proof of the claim.
\medskip

Next, as the ends of $\Gamma^*$ are two rays  with arguments respectively  in $[\arg b, \pi-\arg b]$ and $[\arg b-\pi, -\arg b],$ and the boundary of the regions $C^-(t_\text{max})$ and $C^+(q_\text{min})$ consist of  rays of arguments respectively $\pm \arg b$ and $\pm (\pi-\arg b),$ there exists a $\delta>0$ such that the distance
$$d\big(\Gamma^*, \partial C^-(t_\text{max})\cup \partial C^+(q_\text{min})\big)>\delta.$$
As a consequence, for this $\delta,$ we have 
$$u-t_i\in R^+_\delta\cup R^-_\delta\quad\text{and}\quad q_j-u\in R^+_\delta\cup R^-_\delta$$
for each $u\in\Gamma^*$ and for each $i\in\{1,2,3,4\}$ and $j\in\{1,2,3,4\}.$
\end{proof}

\begin{proof}[Proof of Proposition \ref{contour}]
We first show that the integral in the proposition converges absolutely. Let $R>0$ be sufficiently large. By Lemma \ref{in}, for $u\in \Gamma^*$ with $|u|>R$ and $\mathrm{Im}u>0,$ we have $u-t_i\in R^+_\delta$ for each $i\in\{1,2,3,4\}$ and $q_j-u\in R^-_\delta$ for each $j\in\{1,2,3,4\};$ and for $u\in \Gamma^*$ with $|u|>R$ and $\mathrm{Im}u<0,$ we have $u-t_i\in R^-_\delta$ for each $i\in\{1,2,3,4\}$ and $q_j-u\in R^+_\delta$ for each $j\in\{1,2,3,4\}.$ Let $\Gamma^*_+=\big\{ u\in \Gamma^*\ \big|\ |u|>R\ \ \text{and}\ \ \mathrm{Im}u>0\big\},$  $\Gamma^*_-=\big\{ u\in \Gamma^*\ \big|\ |u|>R\ \ \text{and}\ \ \mathrm{Im}u<0\big\},$ and let $\phi_+$ and $\phi_-$ respectively be the arguments of  $\Gamma^*_+$ and $\Gamma^*_-.$  We for the simplicity of the notations let
$$F(u)\doteq \prod_{i=1}^4S_b(u-t_i)\prod_{j=1}^4S_b(q_j-u)$$ 
be the integrand; and let $s_+, s_-\in\mathbb R$ so that for each $u\in\Gamma_+^*,$
$$u=s_+Q+te^{\mathbf i\phi_+}$$
for some $t>0,$ and 
for each $u\in\Gamma_-^*,$
$$u=s_-Q-te^{\mathbf i\phi_-}$$
for some $t>0.$ Then by Lemma \ref{Sb-bound} with the constant $K$ therein and a direct computation, we have for $u \in \Gamma^*_+$ that  
\begin{equation}\label{2.13}
\begin{split}
\big|F(u) \big|< & K ^ 8 e^{\sum_{i=1}^4\frac{\pi}{2}\mathrm{Im}\big((u-t_i)(u-t_i-Q)\big)-\sum_{j=1}^4\frac{\pi}{2}\mathrm{Im}\big((q_j-u)(q_j-u-Q)\big)} \\
 =& K^8M_+ e^ {-2\pi \big(|b|\sin(\phi_++\arg b) + |b|^{-1} \sin (\phi_+-\arg b) \big)t},
 \end{split}
\end{equation}
where $M_+=e^{\frac{\pi}{2}(7-4s_+)\big(\mathrm{Im}(b^2)+\mathrm{Im}(b^{-2})\big)+\frac{\pi}{2}\big(\sum_{i=1}^4(t_i'^2+t_i')-\sum_{j=1}^4q_j'^2\big)\mathrm{Im}(b^{-2})}$
is a constant independent of $u,$ and $t_i'$ and $q_j'$ are as defined in (\ref{2.9}) and (\ref{2.10});
and  we have for $u\in \Gamma^*_-$ that 
\begin{equation}\label{2.14}
\begin{split}
\big|F(u) \big|< & K ^ 8 e^{-\sum_{i=1}^4\frac{\pi}{2}\mathrm{Im}\big((u-t_i)(u-t_i-Q)\big)+\sum_{j=1}^4\frac{\pi}{2}\mathrm{Im}\big((q_j-u)(q_j-u-Q)\big)} \\
 =&K^8 M_- e^ {-2\pi \big(|b|\sin(\phi_-+\arg b) + |b|^{-1} \sin (\phi_--\arg b) \big)t}
 \end{split}
\end{equation}
where $M_-=e^{-\frac{\pi}{2}(7-4s_-)\big(\mathrm{Im}(b^2)+\mathrm{Im}(b^{-2})\big)-\frac{\pi}{2}\big(\sum_{i=1}^4(t_i'^2+t_i')-\sum_{j=1}^4q_j'^2\big)\mathrm{Im}(b^{-2})}$ 
is a constant independent of $u.$ As both $\phi_+$ and $\phi_-$ are in $[\arg b, \pi-\arg b],$ we have 
\begin{equation}\label{2.15}
   \big(|b|\sin(\phi_++\arg b) + |b|^{-1} \sin (\phi_+-\arg b) \big)>0
\end{equation}
and 
\begin{equation}\label{2.16}
    \big(|b|\sin(\phi_-+\arg b) + |b|^{-1} \sin (\phi_--\arg b) \big)>0.
\end{equation}
Putting (\ref{2.13}), (\ref{2.14}), (\ref{2.15}) and (\ref{2.16}) together, we have
$$\bigg|\int_{\Gamma^*_+}F(u)du \bigg|< K^8M_+\int_R^{+\infty} e^{-2\pi \big(|b|\sin(\phi_++\arg b) + |b|^{-1} \sin (\phi_+-\arg b) \big) t}dt <+\infty$$
and
$$
\bigg|\int_{\Gamma^*_-}F(u)du \bigg|< K^8M_-\int_R^{+\infty} e^{-2\pi \big(|b|\sin(\phi_-+\arg b) + |b|^{-1} \sin (\phi_--\arg b) \big) t}dt <+\infty.$$
As a consequence, the integral in the proposition converges absolutely. 
\medskip

Next, we show that the integral does not depend on the choice of the contour $\Gamma^*.$ Suppose $\Gamma^*_1$ and $\Gamma^*_2$ are two contours satisfying the conditions in the proposition. Let $R>0$ be sufficiently large so that the circle $C_R=\{ u\in\mathbb C\ |\ |u|=R\}$ intersects each of  $\Gamma^*_1$ and $\Gamma^*_2$ at exactly two points, one with a positive imaginary part and one with a negative imaginary part. For $k\in\{1,2\},$ let $\Gamma^*_{k,+}=\big\{ u\in \Gamma^*_k\ \big|\ |u|>R\ \ \text{and}\ \ \mathrm{Im}u>0\big\},$ and let $\Gamma^*_{k,-}=\big\{ u\in \Gamma^*_k\ \big|\ |u|>R\ \ \text{and}\ \ \mathrm{Im}u<0\}.$ We also let $C_+$ be the piece of the circle $C_R$ lying in between $\Gamma^*_{1,+}$ and $\Gamma^*_{2,+},$ let $C_-$ be the piece of the circle $C_R$ lying in between $\Gamma^*_{1,-}$ and $\Gamma^*_{2,-},$ and remember that all the contours $\Gamma^*_{1,+},$ $\Gamma^*_{1,-},$ $\Gamma^*_{2,+},$ $\Gamma^*_{2,-},$ $C_+$ and $C_-$ depend on $R.$ 
Let $\Gamma_+=\Gamma^*_{1,+}\cup C_+\cup \Gamma^*_{2,+}$ and $\Gamma_-=\Gamma^*_{1,-}\cup C_-\cup \Gamma^*_{2,-}$ both being suitably oriented.  See Figure \ref{Gamma3}. 
\begin{figure}[htbp]
\centering
\includegraphics[scale=0.3]{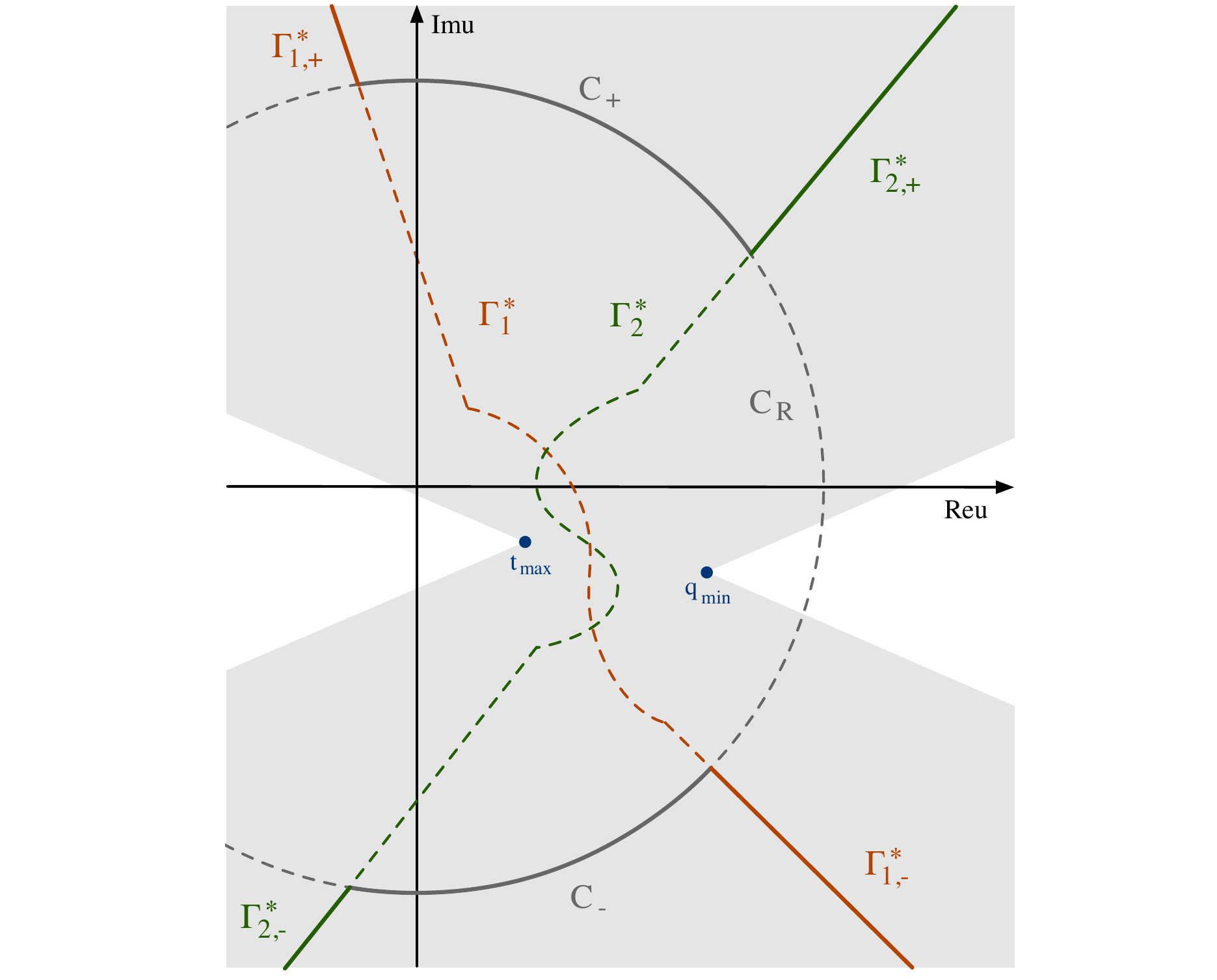}
\caption{The contour $\Gamma_+=\Gamma^*_{1,+}\cup C_+\cup \Gamma^*_{2,+}$ and the contour $\Gamma_-=\Gamma^*_{1,-}\cup C_-\cup \Gamma^*_{2,-}.$}
\label{Gamma3}
\end{figure}
Then by the analyticity of $F$ and Cauchy's Theorem, we have 
$$\int_{\Gamma^*_1}F(u)du-\int_{\Gamma^*_2}F(u)du=\int_{\Gamma_+\cup \Gamma_-}F(u)du;$$
and to prove the result, it suffices to prove that 
$\big|\int_{\Gamma_+\cup\Gamma_-}F(u)du\big|\to 0$
as $R\to +\infty.$ As the integrals over $\Gamma^*_1$ and $\Gamma^*_2$ converge absolutely, the integrals over $\Gamma^*_{1,+},$ $\Gamma^*_{1,-},$ $\Gamma^*_{2,+}$ and $\Gamma^*_{2,-}$ converge to $0$
as $R\to +\infty,$
and we are left to show that 
$$\bigg|\int_{C_+\cup C_-}F(u)du\bigg|\to 0$$
 as $R\to +\infty.$ To this end, for $R$ sufficiently large, we have  for $u\in C_+$ that $u-t_i\in R^+_\delta$ for each $i\in\{1,2,3,4\}$ and $q_j-u\in R^-_\delta$ for each $j\in\{1,2,3,4\};$ and we have for $u\in C_-$ that $u-t_i\in R^-_\delta$ for each $i\in\{1,2,3,4\}$ and $q_j-u\in R^+_\delta$ for each $j\in\{1,2,3,4\}.$ Then by Lemma \ref{Sb-bound} with the constant $K$ therein and a direct computation, for $u=Re^{\mathbf i \phi} \in C_+,$ we have 
\begin{equation}\label{2.17}
\begin{split}
\big|
 F(u)\big|< & K ^ 8 e^{\sum_{i=1}^4\frac{\pi}{2}\mathrm{Im}\big((u-t_i)(u-t_i-Q)\big)-\sum_{j=1}^4\frac{\pi}{2}\mathrm{Im}\big((q_j-u)(q_j-u-Q)\big)} \\
 =& K^8M e^ {-2\pi \big(|b|\sin(\phi+\arg b) + |b|^{-1} \sin (\phi-\arg b) \big)R};
 \end{split}
\end{equation}
and for $u=-Re^{\mathbf i \phi}\in C_-,$ we have 
\begin{equation}\label{2.18}
\begin{split}
\big|
 F(u)\big|< & K ^ 8 e^{-\sum_{i=1}^4\frac{\pi}{2}\mathrm{Im}\big((u-t_i)(u-t_i-Q)\big)+\sum_{j=1}^4\frac{\pi}{2}\mathrm{Im}\big((q_j-u)(q_j-u-Q)\big)} \\
 =& K^8M^{-1} e^ {-2\pi  \big(|b|\sin(\phi+\arg b) + |b|^{-1} \sin (\phi-\arg b) \big)R},
 \end{split}
\end{equation}
where $M=e^{\frac{7\pi}{2}\big(\mathrm{Im}(b^2)+\mathrm{Im}(b^{-2})\big)+\frac{\pi}{2}\big(\sum_{i=1}^4(t_i'^2+t_i')-\sum_{j=1}^4q_j'^2\big)\mathrm{Im}(b^{-2})}$
is a constant independent of $u.$ Now, as $ |b|\sin(\phi+\arg b)+|b|^{-1}\sin(\phi-\arg b)>0$
 for each $\phi \in [\arg b, \pi-\arg b],$ there is  an $\epsilon >0$ and a $c>0$ independent of $\phi$ such that 
 \begin{equation}\label{2.19}
     |b|\sin(\phi+\arg b)+|b|^{-1}\sin(\phi-\arg b)\geqslant c 
 \end{equation}
for each 
$\phi \in [\arg b -\epsilon, \pi-\arg b+\epsilon].$ 
As the arguments of the ends of $\Gamma^*_1$ and $\Gamma^*_2$ lie in $[\arg b, \pi-\arg b]$ and $[\arg b-\pi, -\arg b],$ for  $R$ sufficiently large and for each $u=Re^{\mathbf i\phi}\in C_+$ and $u=-Re^{\mathbf i\phi}\in C_-,$ we have  $\phi \in [\arg b -\epsilon, \pi-\arg b+\epsilon].$ 
Together with (\ref{2.17}), (\ref{2.18}) and (\ref{2.19}) and letting $N= K^8 \max\{M,M^{-1}\},$ we have 
$$\big|F(u)\big|< Ne^{-2\pi cR}$$
for each $u\in C_+\cup C_-.$  As a consequence, 
$$\bigg|\int_{C_+\cup C_-}F(u)du\bigg|< 2\pi RN e^{-2\pi cR},$$
which converges to $0$ as $R\to +\infty.$ 
\end{proof}


\section{Double sine function and complexified Lobachevsky function}\label{SbLi}

The asymptotics of complex $b$-$6j$ symbols rely heavily on the relationship between the double sine function and the complexified Lobachevsky function. The main result of this section is Proposition \ref{EST} that computes the classical limit of the function  $\log S_b\Big(\frac{x}{\pi b}+\frac{b}{2}\Big).$

For $x\in \mathbb C$ with 
$0<\mathrm{Re}\Big(\frac{x}{\pi b}+\frac{b}{2}\Big)<\mathrm{Re}Q,$ we define 
$$\log S_b\Big(\frac{x}{\pi b}+\frac{b}{2}\Big)\doteq\int_\Omega\frac{\sinh\Big(\big(\frac{1}{2b}-\frac{x}{\pi b}\big)t\Big)}{4t\sinh(\frac{bt}{2})\sinh(\frac{t}{2b})}dt;$$
and by the functional equation (\ref{FE1}), $\log S_b\Big(\frac{x}{\pi b}+\frac{b}{2}\Big)$ meromorphically extends to the $\mathbb C$ with the set of poles 
$$\bigg\{-\Big(m+\frac{1}{2}\Big)\pi b^2-n\pi\ \bigg|\ m,n\in\mathbb Z_{\geqslant 0}\bigg\}\cup\bigg\{\Big(m+\frac{1}{2}\Big)\pi b^2+(n+1)\pi\ \bigg|\ m,n\in\mathbb Z_{\geqslant 0}\bigg\}.$$
In particular, it is holomorphic on the region $H_\theta$ depicted in Figure \ref{HHd} (a), where $\theta=\arg b,$ consisting of $x$ such that $\arg \big(\frac{x}{\pi b}+\frac{b}{2}\big) \notin [-\pi,-\pi+\theta]\cup [\pi-\theta,\pi]$ and $\arg \big(\frac{x}{\pi b}+\frac{b}{2}-Q\big) \notin [-\theta,\theta],$ i.e., 
\begin{equation*}
H_\theta=\bigg\{ x\in \mathbb C \ \bigg | \arg \Big(x+\frac{\pi b^2}{2}\Big) \notin [-\pi,-\pi+2\theta]\ \ \text{and} \ \ \arg \Big(x-\pi-\frac{\pi b^2}{2}\Big) \notin [0,2\theta] \bigg\}.
\end{equation*}
\begin{figure}[htbp]
\centering
\includegraphics[scale=0.3]{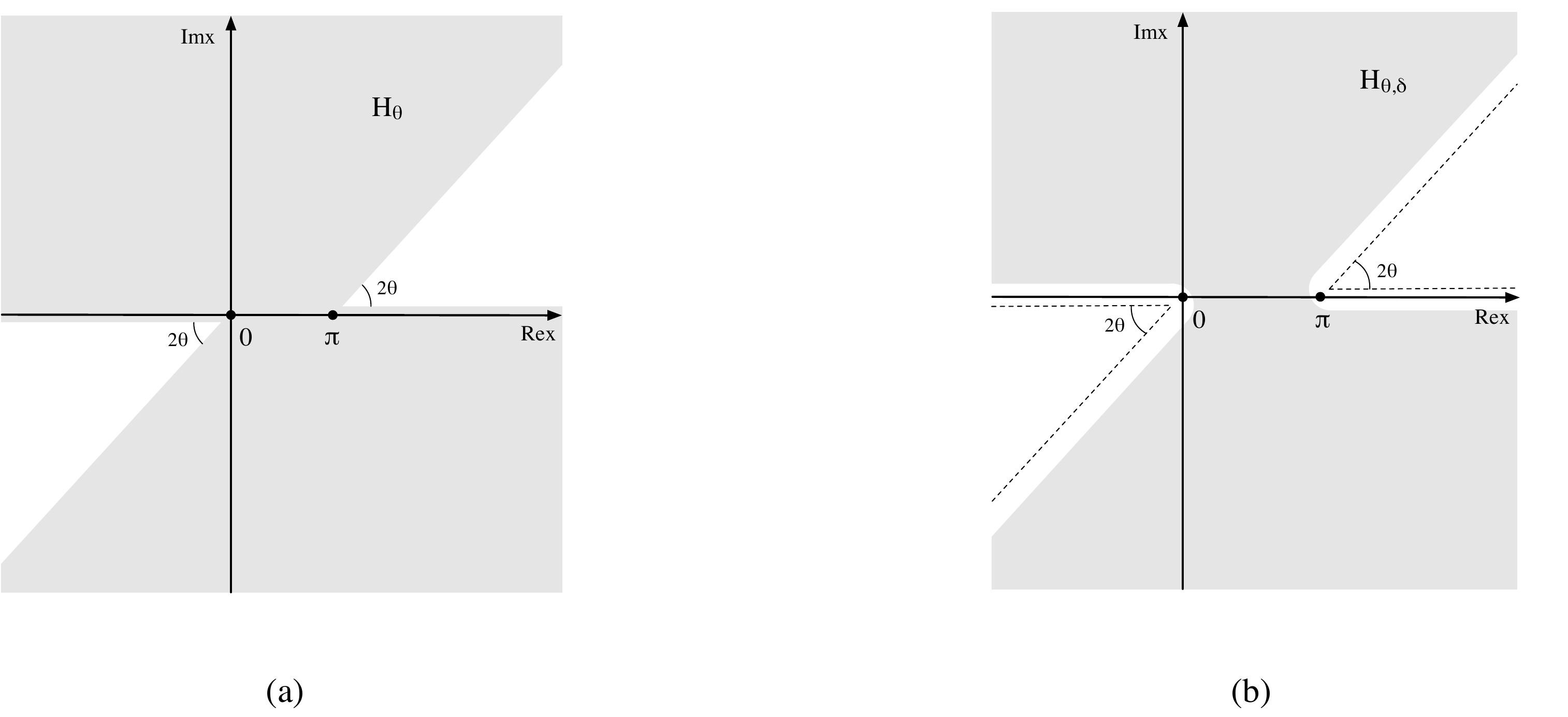}
\caption{The region $H_\theta$ in (a), and the region $H_{\theta,\delta}$ in (b).}
\label{HHd}
\end{figure}

The classical limit of $\log S_b\Big(\frac{x}{\pi b}+\frac{b}{2}\Big)$ as $b\to 0$ is closely related to  the following \emph{complexified Lobachevsky function}
\begin{equation}\label{eq:Lx}
L(x)= x^2-\pi x +\frac{\pi^2}{6}-\mathrm{Li}_2\big(e^{2\mathbf ix}\big)
\end{equation}
for $x\in\mathbb C$ with $0<\mathrm{Re}x<\pi.$ Here $\mathrm{Li}_2$ is the dilogarithm function defined on $\mathbb C\setminus (1,+\infty)$ by
$$\mathrm{Li}_2(z)=-\int_0^z\frac{\log (1-u)}{u}du,$$
where the integral is along any path in $\mathbb C\setminus (1,\infty)$ connecting $0$ and $z.$
We call the function $L$ the complexified Lobachevsky function because for each $x\in (0,\pi),$
$$L(x)=-2\mathbf i \Lambda(x),$$
where $\Lambda:\mathbb R \to \mathbb R $ is the Lobachevsky function defined by 
$$\Lambda(x)=-\int_0^x \log|2\sin t|dt.$$
It is proved in \cite[Section 2.2]{LMSWY} that the function $L(x)$ extends continuously to $\mathbb C\setminus \big((-\infty,0)\cup(\pi,\infty)\big)$ and  holomorphically to the region $$H=\mathbb C\setminus \big((-\infty,0]\cup[\pi,\infty)\big),$$ by the following functional equations:
 If $x\in H$ with $\mathrm{Im}x>0,$ then
\begin{equation}\label{period3}
L(x+\pi)=L(x)+2\pi x;
\end{equation}
and if $x\in H$ with $\mathrm{Im}x< 0,$ then
\begin{equation}\label{period4}
L(x+\pi)=L(x)-2\pi x.
\end{equation}


We  now approximate $S_b\Big(\frac{x}{\pi b} +\frac{b}{2}\Big)$ using $L(x).$ Let $\nu_b(x)$ be such that
$$2\pi \mathbf i b^2 \log S_b\Big(\frac{x}{\pi b} +\frac{b}{2}\Big)-L(x)=\nu_b(x)b^4.$$ 
For  $\delta>0$ sufficiently small, as depicted in Figure~\ref{HHd} (b), let 
$$H_{\theta,\delta}=\Big\{ x\in H_\theta\ \Big|\ d(x,\partial H_\theta) \geqslant \delta \Big\}.$$

\begin{proposition}\label{EST}  For $\theta\in (0,\frac{\pi}{2}),$ there exists a $B=B_{\theta,\delta}>0$  independent of $b$ such that for all $b\in \mathbb C$ with  $|b|$ sufficiently small and $\arg b=\theta,$ and for all  $x\in H_{\theta,\delta},$ 
$$|\nu_b(x)|< B,$$
i.e., 
$$\bigg|2\pi \mathbf i b^2 \log S_b\Big(\frac{x}{\pi b} +\frac{b}{2}\Big)-L(x)\bigg|<B|b|^4.$$
\end{proposition}

The key ingredient in the proof of Proposition \ref{EST} is Proposition \ref{EST2} below that estimates the difference between Faddeev's quantum dilogarithm function $\Phi_b$ and the classical dilogarithm function $\mathrm{Li}_2$ on the suitable region. To state the result, for $\theta=\arg b\in(0,\frac{\pi}{2}),$  let 
$$V_\theta=\bigg\{ y\in \mathbb C\ \bigg|\ \arg\Big(y-\pi\mathbf i-\pi b^2\mathbf i\Big)\notin \Big[\frac{\pi}{2},\frac{\pi}{2}+2\theta\Big]\ \ \text{and}\ \  \arg\Big(y +\pi\mathbf i+\pi b^2 \mathbf i\Big)\notin \Big[-\frac{\pi}{2},-\frac{\pi}{2}+2\theta\Big]\bigg\}$$
as depicted in Figure \ref{V} (a). 
\begin{figure}[htbp]
\centering
\includegraphics[scale=0.3]{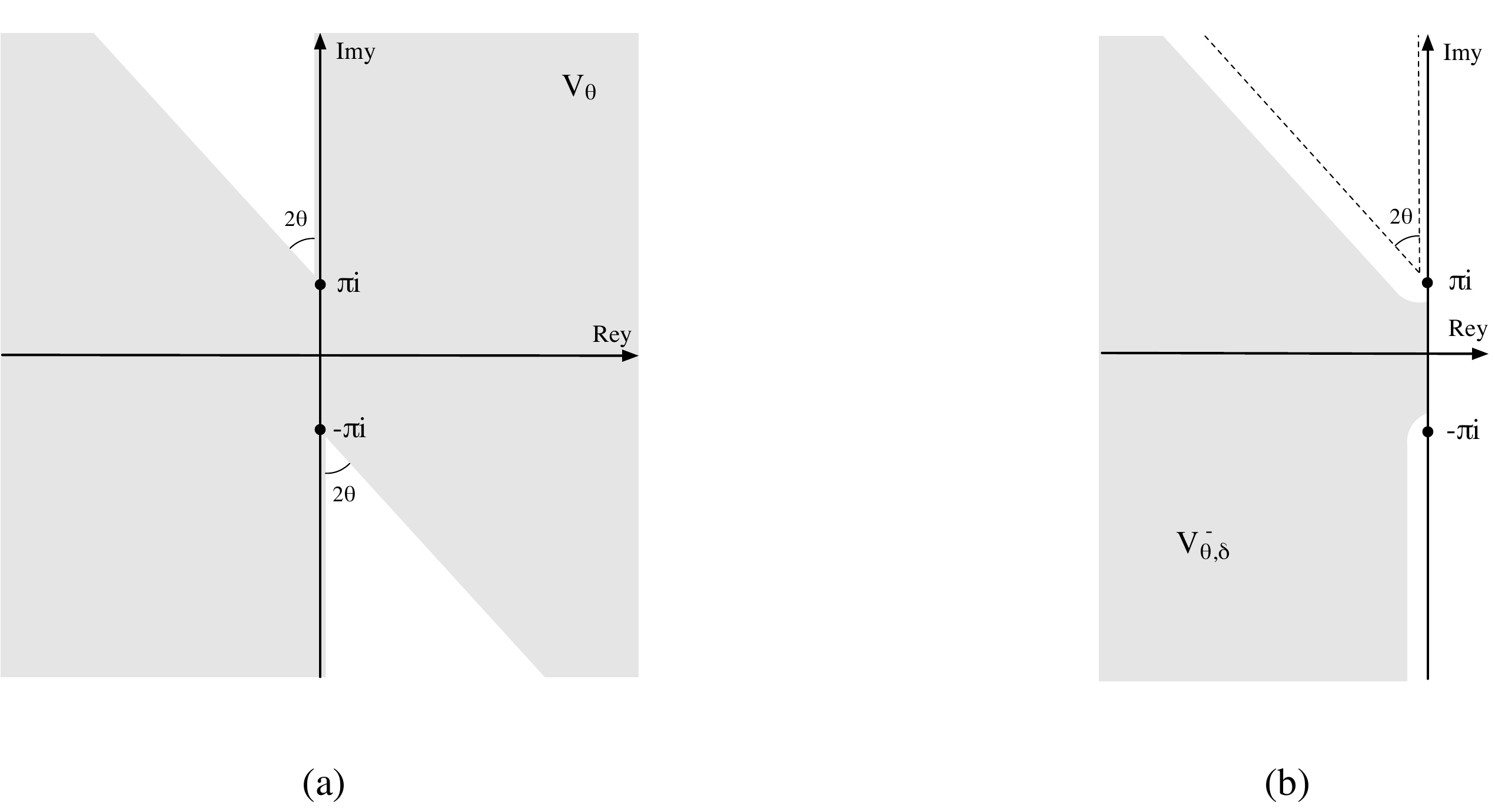}
\caption{The region $V_\theta$ in (a), and the region $V_{\theta,\delta}^-$ in (b).}
\label{V}
\end{figure}
Then $\Phi_b\big(\frac{y}{2\pi b}\big)$ is holomorphic on $V_\theta.$ Now for $y\in V_\theta$ satisfying $-\mathrm{Re}\big(\frac{Q}{2}\big)< \mathrm{Im}\big(\frac{y}{2\pi b}\big)<\mathrm{Re}\big(\frac{Q}{2}\big),$ we define
$$\log\Phi_b\Big(\frac{y}{2\pi b}\Big)\doteq\int_\Omega \frac{e^{-\mathbf i yt}}{4t\sinh(\pi t)\sinh(\pi t b^2)}dt,$$
and extend it analytically to $V_\theta$ using the functional equation (\ref{phieq}). 
For 
$\delta>0$ sufficiently small, we let 
$$V_{\theta,\delta}^-=\Big\{ y\in V_\theta\ \Big|\ \mathrm{Re}y\leqslant 0\ \ \text{and}\ \  d(y,\partial V_{\theta})>\delta\Big\}$$
as depicted in Figure \ref{V} (b). 

\begin{proposition}\label{EST2}
For $\theta\in (0,\frac{\pi}{2}),$ there exists a $B=B_{\theta,\delta}>0$ independent of $b$ such that for all $b\in\mathbb C$ with $|b|$ sufficiently small and $\arg b=\theta,$ and for all $y\in V_{\theta,\delta}^-,$ 
$$\Bigg|\log\Phi_b\bigg(\frac{y}{2\pi b}\bigg)-\frac{1}{2\pi \mathbf i b^2}\mathrm{Li}_2\big(-e^y\big)\Bigg|< B|b|^2.$$
\end{proposition}

To prove Proposition \ref{EST2}, we need Lemma \ref{log=} and Lemma \ref{N} below. Observe that for $y\in V_{\theta,\delta}^-,$ $\frac{y}{2\pi b}$ is neither a pole nor a zero of $\Phi_b.$ Then by Proposition \ref{Prod}, 
\begin{equation}\label{prod}
\Phi_b\Big(\frac{y}{2\pi b}\Big)=\frac{\prod_{n=0}^{+\infty}\Big(1+e^{y+(2n+1)\pi \mathbf i b^2}\Big)}{\prod_{n=0}^{+\infty}\Big(1+e^{\big(y-(2n+1)\pi \mathbf i \big)b^{-2}}\Big)}.
\end{equation}
Denote
$$\log N_b\Big(\frac{y}{2\pi b}\Big)\doteq \sum_{n=0}^{+\infty}\log \Big(1+e^{y+(2n+1)\pi \mathbf i b^2}\Big) $$
and
$$\log M_b\Big(\frac{y}{2\pi b}\Big)\doteq \sum_{n=0}^{+\infty}\log \Big(1+e^{\big(y-(2n+1)\pi \mathbf i \big)b^{-2}}\Big).$$

\begin{lemma}\label{log=} 
For $b\in\mathbb C$ with $|b|$ sufficiently small and $\arg b=\theta,$ both
$\log N_b\big(\frac{y}{2\pi b}\big)$ and $\log M_b\big(\frac{y}{2\pi b}\big)$ define holomorphic functions on $V_{\theta,\delta}^-;$ and
\begin{equation}\label{log}
    \log\Phi_b\Big(\frac{y}{2\pi b}\Big)=\log N_b\Big(\frac{y}{2\pi b}\Big)-\log M_b\Big(\frac{y}{2\pi b}\Big)
\end{equation}
for each $y\in V_{\theta,\delta}^-.$
\end{lemma}

\begin{proof} First of all, since $\mathrm{Re} y\leqslant 0$ for each $y\in V_{\theta,\delta}^-$ and $\mathrm{Im}(b^2)>0,$ we have
$$\Big|e^{y+(2n+1)\pi\mathbf ib^2}\Big|= e^{\mathrm{Re} y-(2n+1)\pi\mathrm{Im}(b^2)}\leqslant e^{-(2n+1)\pi\mathrm{Im}(b^2)}<1 $$
for each $n\geqslant 0.$ Therefore, for each $n\geqslant 0,$ $\arg \big(1+e^{y+(2n+1)\pi \mathbf i b^2}\big)\in (-\pi,\pi)$ and $\log \big(1+e^{y+(2n+1)\pi \mathbf i b^2}\big)$ is a well defined holomorphic function on $V_{\theta,\delta}^-;$ 
 and there is a $K>0$ such that 
 $$\bigg|\log \Big(1+e^{y+(2n+1)\pi \mathbf i b^2}\Big)\bigg|\leqslant K \Big|e^{y+(2n+1)\pi\mathbf ib^2}\Big|\leqslant K  e^{-(2n+1)\pi\mathrm{Im}(b^2)} $$ for each $y\in V_{\theta,\delta}^-$ and $n\geqslant 0.$ Aa a consequence, 
 $$\sum_{n=0}^{+\infty} \bigg|\log \Big(1+e^{y+(2n+1)\pi \mathbf i b^2}\Big)\bigg| \leqslant K \sum_{n=0}^{+\infty} e^{-(2n+1)\pi\mathrm{Im}(b^2)}$$
converges, and $\sum_{n=0}^{+\infty} \log \big(1+e^{y+(2n+1)\pi \mathbf i b^2}\big)$ converges uniformly on $V_{\theta,\delta}^-.$ Then by Weierstrass' Theorem, 
$\log N_b\big(\frac{y}{2\pi b}\big)$ defines a holomorphic function on $V_{\theta,\delta}^-$ for any $b\in \mathbb C$ with $\arg b=\theta$ and with an arbitrary $|b|.$

Next, let $R$ be the ray starting from $\pi\mathbf i$ with angle $2\theta$ from the positive $\mathrm{Im}y$-axis, and let $b\in\mathbb C$ with $\pi|b|^2<\delta $ so that $\pi\mathbf i$ does not lie in $V_{\theta,\delta}^-.$  Then as shown in  Figure \ref{VV} (a), there is a $d\in (0,\delta)$ such that $d(y, R)> d$ for each $y\in V_{\theta,\delta}^-.$  
\begin{figure}[htbp]
\centering
\includegraphics[scale=0.25]{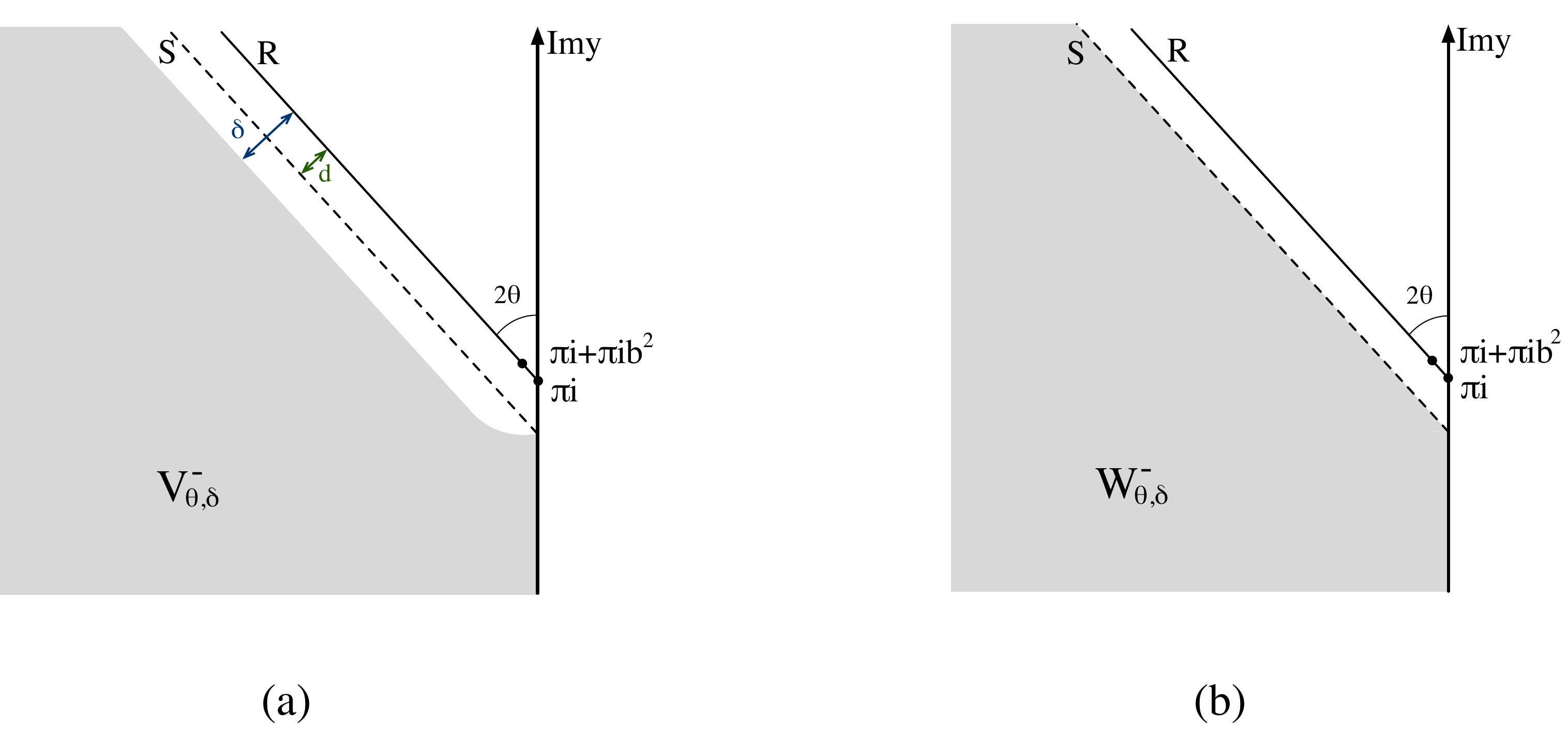}
\caption{The distance $d$ in (a), and the region $W_{\theta,\delta}^-$ in (b).}
\label{VV}
\end{figure}
This is  equivalent to that $\mathrm{Re}\big((y-\pi \mathbf i)e^{-2\mathbf i\theta}\big)< -d$
for each $y\in V_{\theta,\delta}^-.$ We also observe that as $|b|$ approaches $ 0,$ $d$ approaches $\delta\sin(2\theta).$ As a consequence, when $|b|$ is sufficiently small, we have $d>\frac{\delta\sin(2\theta)}{2}$ and 
$$\mathrm{Re}\Big(\big(y-\pi \mathbf i\big)b^{-2}\Big)<-\frac{\delta\sin(2\theta)|b|^{-2}}{2}=\frac{\delta}{2}\mathrm{Im}(\big(b^{-2}\big)$$
for each $y\in V_{\theta,\delta}^-.$
Together with the fact that $\mathrm{Im}(b^{-2})<0,$ we have 
$$\Big|e^{\big(y-(2n+1)\pi \mathbf i \big)b^{-2}}\Big|= e^{\mathrm{Re}\big((y-\pi\mathbf i)b^{-2}\big)+2n\pi\mathrm{Im}(b^{-2})}<e^{\frac{\delta}{2}\mathrm{Im}((b^{-2})+2n\pi\mathrm{Im}(b^{-2})}<1$$
for each $n\geqslant 0.$ Therefore, for each $n\geqslant 0,$ $\arg \big(1+e^{\big(y-(2n+1)\pi \mathbf i\big) b^{-2}}\big)\in (-\pi,\pi)$ and $\log \big(1+e^{\big(y-(2n+1)\pi \mathbf i \big)b^{-2}}\big)$ is a well defined holomorphic function on $V_{\theta,\delta}^-;$ 
 and there is a $K>0$ such that 
 \begin{equation}\label{K}
 \bigg|\log \Big(1+e^{\big(y-(2n+1)\pi \mathbf i \big)b^{-2}}\Big)\bigg|\leqslant K \Big|e^{\big(y-(2n+1)\pi\mathbf i\big)b^{-2}}\Big|\leqslant K e^{\frac{\delta}{2}\mathrm{Im}((b^{-2})+2n\pi\mathrm{Im}(b^{-2})}
 \end{equation}
 for each $y\in V_{\theta,\delta}^-$ and  $n\geqslant 0.$ Aa a consequence, 
 \begin{equation}\label{M}
     \sum_{n=0}^{+\infty} \Bigg|\log \bigg(1+e^{\big(y-(2n+1)\pi \mathbf i \big)b^{-2}}\bigg)\Bigg| \leqslant K e^{\frac{\delta}{2}\mathrm{Im}((b^{-2})}\sum_{n=0}^{+\infty} e^{2n\pi\mathrm{Im}(b^{-2})}
     \end{equation}
converges, and $\sum_{n=0}^{+\infty} \log \Big(1+e^{(y-(2n+1)\pi \mathbf i) b^{-2}}\big)$ converges uniformly on $V_{\theta,\delta}^-.$ Then by Weierstrass' Theorem, $\log M_b\big(\frac{y}{2\pi b}\big)$ defines a holomorphic function on $V_{\theta,\delta}^-$ for $b\in\mathbb C$ with $\arg b=\theta$ and a sufficiently small $|b|.$  By the same argument, we can also show that $\log M_b\big(\frac{y}{2\pi b}\big)$ defines a holomorphic function on a small neighborhood of $0$ in  $V_{\theta,\delta}^-$ for any $b\in \mathbb C$ with $\arg b=\theta$ and with an arbitrary $|b|.$ 

Finally, by (\ref{prod}), we have 
$$\log\Phi_b\Big(\frac{y}{2\pi b}\Big)=\log N_b\Big(\frac{y}{2\pi b}\Big)-\log M_b\Big(\frac{y}{2\pi b}\Big)+2k\pi\mathbf i$$
for some integer $k.$ To prove (\ref{log}), at $y=0$ and $b_0=e^{\frac{\pi\mathbf i}{4}},$ we have by a direct computation that 
$$\log \Phi_{b_0}(0)=\frac{\big(b_0^2+b_0^{-2}\big)\pi\mathbf i}{24}=0,$$
and that  $\log N_{b_0}(0)=\log M_{b_0}(0)=\sum_{n=0}^{+\infty}\log\big(1+e^{-(2n+1)\pi}\big),$
hence 
$$\log N_{b_0}(0)-\log M_{b_0}(0)=0.$$
This implies that $k=0,$ and (\ref{log}) holds. 
\end{proof}

\begin{lemma}\label{N}
\begin{enumerate}[(1)]
    \item There exists a $C=C_{\theta,\delta}>0$ independent of $b$ such that for all $b\in \mathbb C$ with $|b|$ sufficiently small and $\arg b=\theta,$ and for all $y\in V_{\theta,\delta}^-,$ 
$$\bigg|\log N_b\Big(\frac{y}{2\pi b}\Big)-\frac{1}{2\pi \mathbf i b^2}\mathrm{Li}_2\big(-e^y\big)\bigg|< C|b|^2.$$

\item There exists a $D=D_{\theta,\delta}>0$ independent of $b$ such that for all $b\in \mathbb C$ with $|b|$ sufficiently small and $\arg b=\theta,$ and for all $y\in V_{\theta,\delta}^-,$ 
$$\bigg|\log M_b\Big(\frac{y}{2\pi b}\Big)\bigg|< D |b|^2.$$
\end{enumerate}
\end{lemma}

The proof of Lemma \ref{N} relies on the following Lemma \ref{EM}, which is a consequence of Euler’s Summation Formula. 
\begin{lemma}\label{EM}
Let $f:[0,+\infty)\to \mathbb C$ be a $C^2$-smooth function with the series $\sum_{n=0}^{+\infty} f(n)$ converges, the integral $\int_0^{+\infty}f(x)dx$ converges, and the limits $\lim_{x\to +\infty}f(x)=\lim_{x\to +\infty}f'(x)=\lim_{x\to +\infty}f''(x)=0.$ Then 
\begin{equation*}
\sum_{n=0}^{+\infty}f(n)=\int_0^{+\infty}f(x)dx +\frac{f(0)}{2}-\frac{f'(0)}{12}-\frac{1}{2}\int_0^{+\infty}f''(x)P_2(x)dx,
\end{equation*}
where $P_2(x)=\big(x-[x]\big)^2 - \big(x-[x]\big) +\frac{1}{6}$ is the second Bernoulli polynomial of $x-[x],$ and the Gauss symbol $[x]$ is the largest integer less than or equal to $x.$
\end{lemma}

\begin{proof}
By the Second-derivative form of Euler’s Summation Formula (see e.g. \cite[Theorem 3]{A}), we have for each $k\in \mathbb Z_{\geqslant 0}$ that 
\begin{equation}\label{em}
\sum_{n=0}^kf(n)=\int_0^kf(x)dx+\frac{f(k)+f(0)}{2}+\frac{f'(k)-f'(0)}{12}-\frac{1}{2}\int_0^kf''(x)P_2(x)dx.
\end{equation}
Indeed, the original statement of \cite[Theorem 3]{A} is for real valued $C^2$-smooth functions, applying which respectively to the real and the imaginary parts of $f$ we get (\ref{em}). To prove the result, since $\sum_{n=0}^{+\infty}f(n)$ and $\int_0^{+\infty}f(x)dx$ converge, and $f(x)\to 0$ and $f'(x)\to 0$ as $x\to +\infty,$  it suffices to show that the integral 
$$F(x)\doteq\int_0^xf''(t)P_2(t)dt$$ converges as $x\to +\infty.$

To this end, as  $\sum_{n=0}^{k}f(n)$ and $\int_0^{k}f(x)dx$ converge as $k\to+\infty,$ and $f(k)\to 0$ and $f'(k)\to 0$ as $k\to +\infty,$  we have directly from (\ref{em}) that the integral $F(k)$ converges as  $k\in \mathbb Z_{\geqslant 0}$ and  $k\to +\infty.$ In particular, $F([x])$ converges as $x\to +\infty.$

Next, for each $x\in [0,+\infty),$ we have
$F(x)-F([x])=\int_{[x]}^xf''(t)P_2(t)dt.$
Since $\big|P_2(x)\big|\leqslant \frac{1}{6}$  for all $x\in [0,+\infty),$ we have  
$$\big|F(x)-F([x])\big|\leqslant \frac{1}{6}\int_{[x]}^x\big|f''(t)|dt\leqslant \frac{1}{6}\sup_{[x]\leqslant t<[x]+1}\big|f''(t)\big|;$$
and since $f''$ is continuous and $f''(x)\to 0$ as $x\to +\infty,$
 we have $\sup_{[x]\leqslant t<[x]+1}\big|f''(t)\big| \to 0$ as $x\to +\infty.$ This implies that $F(x)-F([x])$ converges to $0$ as $x\to +\infty.$ 
 
 Finally, as $F([x])$ converges as $x\to +\infty,$ we have that $F(x)$ converges as $x\to +\infty.$ 
\end{proof}

\begin{proof}[Proof of Lemma \ref{N}]
For (1), we let 
\begin{equation}\label{fy}
f_y(x)\doteq \log\Big(1+e^{y+(2x+1)\pi\mathbf i b^2}\Big)
\end{equation}
for each $x\in [0,+\infty).$ Then for each $y\in V_{\theta,\delta}^-,$
$$\log N_b\Big(\frac{y}{2\pi b}\Big)=\sum_{n=0}^{+\infty}f_y(n).$$

We first verify that for each $y\in V_{\theta,\delta}^-,$ $f_y$ satisfies the conditions of Lemma \ref{EM}. Clearly, $f_y(x)$ is $C^2$-smooth on $[0,+\infty).$  The convergence of the series 
$\sum_{n=0}^{+\infty}f_y(n)$ is given by Lemma \ref{log=}. For the converges of the integral $\int_0^{+\infty}f_y(x)dx,$ as 
$\lim_{x\to +\infty}\frac{f_y(x)}{e^{y+(2x+1)\pi\mathbf i b^2}}=1,$ there exists a $x_0>0$ such that for all $x>x_0,$ $\big|f_y(x)\big|< 2\big|e^{y+(2x+1)\pi\mathbf i b^2}\big|=2e^{\mathrm{Re}y-(2x+1)\pi\mathrm{Im}(b^2)
}.$ As a consequence, $\big|\int_0^{+\infty}f_y(x)dx\big|<\big|\int_0^{x_0}f_y(x)dx\big|+2e^{\mathrm{Re}y}\int_{x_0}^{+\infty}e^{-(2x+1)\pi\mathrm{Im}(b^2)}dx,$
which converges by the fact that $\mathrm{Im}(b^2)>0.$ For the limits of $f_y^{(i)}(x)$ as $x\to +\infty,$ $i=0,1,2,$ by a direct computation, we have 
\begin{equation}\label{f'y}
f_y'(x)=2\pi\mathbf i b^2\frac{e^{y+(2x+1)\pi\mathbf ib^2}}{1+e^{y+(2x+1)\pi\mathbf ib^2}}\quad\text{and}\quad f_y''(x)=-4\pi^2 b^4\frac{e^{y+(2x+1)\pi\mathbf ib^2}}{\big(1+e^{y+(2x+1)\pi\mathbf ib^2}\big)^2};
\end{equation}
and as $\mathrm{Im}(b^2)>0,$ we have  $\lim_{x\to +\infty}e^{y+(2x+1)\pi\mathbf ib^2}=0$ and
$\lim_{x\to +\infty}f_y(x)=\lim_{x\to +\infty}f'_y(x)=\lim_{x\to +\infty}f''_y(x)=0.$
Therefore, all the conditions of Lemma \ref{EM} are satisfied by $f_y(x),$ from which we have 
\begin{equation}\label{3.10}
    \log N_b\Big(\frac{y}{2\pi b}\Big)=\int_0^{+\infty}f_y(x)dx +\frac{f_y(0)}{2}-\frac{f_y'(0)}{12}-\frac{1}{2}\int_0^{+\infty}f_y''(x)P_2(x)dx.
\end{equation}

Next, we relate the first term on a right hand side above to the function $\frac{1}{2\pi\mathbf ib^2}\mathrm{Li}_2(-e^y).$ To do this, 
we consider a larger region $W_{\theta,\delta}^-$ which is obtained from $V_{\theta,\delta}^-$ by joining in the points  below the dotted ray $S.$ See  Figure \ref{VV} (b). More precisely, $$W_{\theta,\delta}^-=V_{\theta,\delta}^-\cup \bigg\{ y\in \mathbb C\ \bigg|\ \mathrm{Re}\Big((y-z)e^{2\mathbf i\theta}\Big)<0\text{ for all }z\in S,\ \ \mathrm{Re} y <0 \ \ \text{and}\ \  \mathrm{Im}y >0 \bigg\}.$$ 
We notice that when $\theta\in[\frac{\pi}{4},\frac{\pi}{2}),$ $W_{\theta,\delta}^-$ coincides with $V_{\theta,\delta}^-.$ 
We also observe that, for each $y\in V_{\theta,\delta}^-,$ the ray $\big\{ y+(2x+1)\pi\mathbf i b^2\ \big|\ x\in [0,+\infty)\big\}$ lies entirely in the region $W_{\theta,\delta}^-.$ As a consequence, 
$\big|-e^{y+(2x+1)\pi \mathbf ib^2}\big|=e^{\mathrm{Re}y-(2x+1)\pi \mathrm{Im}(b^2)}<1,$ and hence $-e^{y+(2x+1)\pi \mathbf ib^2}\in\mathbb C\setminus [1,+\infty)$ for each $y\in V_{\theta,\delta}^-$ and $x\in [0,+\infty).$  Then by the change of variable $z=-e^{y+(2x+1)\pi\mathbf ib^2},$ we have
 \begin{equation}\label{3.11}
 \begin{split}
     \int_0^{+\infty}f_y(x)dx = & \lim_{x_0\to +\infty} \int_0^{x_0}\log\Big(1+e^{y+(2x+1)\pi\mathbf ib^2}\Big)dx\\
=& \lim_{x_0\to +\infty} \frac{1}{2\pi \mathbf ib^2} \int_{-e^{y+\pi\mathbf i b^2}}^{-e^{y+(2x_0+1)\pi\mathbf i b^2}}\frac{\log(1-z)}{z}dz\\
=&\lim_{x_0\to +\infty} \frac{1}{2\pi \mathbf ib^2} \bigg(\mathrm{Li}_2\Big(-e^{y+\pi\mathbf i b^2}\Big)- \mathrm{Li}_2\Big(-e^{y+(2x_0+1)\pi\mathbf i b^2}\Big)\bigg)\\
=&\frac{1}{2\pi \mathbf ib^2}\mathrm{Li}_2\Big(-e^{y+\pi\mathbf i b^2}\Big),
 \end{split}
 \end{equation}
where 
the last equality comes from that $\mathrm{Im}(b^2)>0,$ and hence $\lim_{x_0\to +\infty} \mathrm{Li}_2\Big(-e^{y+(2x_0+1)\pi\mathbf i b^2}\Big)=\mathrm{Li}_2(0)=0.$ Applying the integral form of Taylor's Remainder Theorem that 
$f(y+h)=f(y)+f'(y)h+\int_y^{y+h}f''(t)(y+h-t)dt$
to the function 
\begin{equation}\label{f}
f(y)\doteq \mathrm{Li}_2(-e^y)
\end{equation}
with $h=\pi \mathbf i b^2,$ we have
\begin{equation}\label{3.12}
\frac{1}{2\pi\mathbf i b^2}\mathrm{Li}_2\Big(-e^{y+\pi\mathbf i b^2}\Big)  = \frac{1}{2\pi\mathbf i b^2}\mathrm{Li}_2\big(-e^{y}\big) + \frac{f'(y)}{2} + \frac{1}{2\pi\mathbf i b^2}\int_y^{y+\pi\mathbf ib^2}f''(t) \Big(y+\pi\mathbf ib^2-t\Big)dt.
\end{equation}

Putting (\ref{3.10}),  (\ref{3.11}), (\ref{f})  and (\ref{3.12}) together, 
we have 
\begin{equation}\label{3.14}
\begin{split}
&\log N_b\Big(\frac{y}{2\pi b}\Big)-\frac{1}{2\pi \mathbf i b^2}\mathrm{Li}_2\big(-e^y\big)
\\
=&\,\frac{1}{2}\Big(f_y(0)+f'(y)\Big) -\frac{f_y'(0)}{12}-\frac{1}{2}\int_0^{+\infty}f_y''(x)P_2(x)dx + \frac{1}{2\pi \mathbf i b^2}\int_y^{y+\pi\mathbf ib^2}f''(t) \Big(y+\pi\mathbf ib^2-t\Big)dt;
\end{split}    
\end{equation}
and we will estimate each of the terms on the right hand side. To this end, we let 
$\epsilon= \min\big\{\frac{\delta\tan(2\theta)}{4},\frac{\delta}{4}\big\}$ when $\theta\in(0,\frac{\pi}{4}),$ 
and let 
$\epsilon = \frac{\delta}{4}$ when  $\theta\in[\frac{\pi}{4},\frac{\pi}{2}).$ 
Then as shown in Figure \ref{delta}, 
\begin{figure}[htbp]
\centering
\includegraphics[scale=0.2]{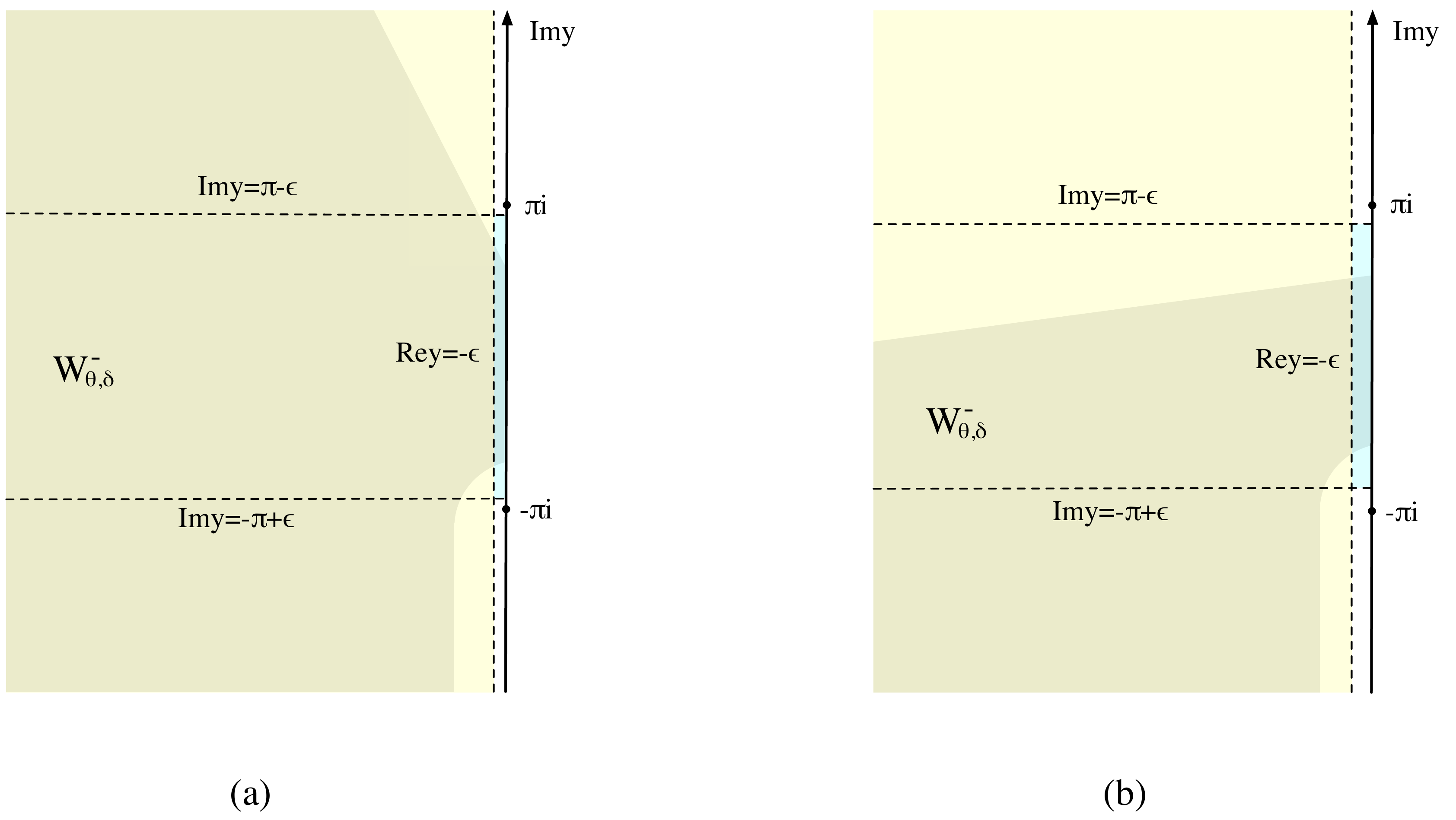}
\caption{The case $\theta\in (0,\frac{\pi}{4})$ in (a), and the case $\theta\in[\frac{\pi}{4},\frac{\pi}{2})$ in (b).}
\label{delta}
\end{figure}
for each $u \in W_{\theta,\delta}^-$ and $b\in\mathbb C$ with $\arg b=\theta$ and with $|b|$ sufficiently small,  either $u$ lies in the light yellow region, i.e.,
$$\mathrm{Re}u<-\epsilon,$$
or $u$ lies in the light blue region, i.e., 
$$-\epsilon\leqslant \mathrm{Re}u\leqslant 0 \quad\text{and}\quad -\pi + \epsilon < \mathrm{Im}u < \pi -\epsilon.$$ 
In the former case, we have $\big|1+e^u\big|\geqslant 1-e^{\mathrm{Re}u}>1-e^{-\epsilon}.$
In the latter case, if $\mathrm{Im}u \in (-\frac{\pi}{2}, \frac{\pi}{2}),$ then
$|1+e^u|\geqslant \mathrm{Re}(1+e^u)=1+e^{\mathrm{Re}u}\cos(\mathrm{Im}u)> 1;$ and if $\mathrm{Im}u \in (-\pi+\epsilon,\pi-\epsilon)\setminus (-\frac{\pi}{2}, \frac{\pi}{2}),$ then
$|1+e^u|\geqslant \big|\mathrm{Im}(1+e^u)\big|=\big|e^{\mathrm{Re}u}\sin(\mathrm{Im}u)\big|>e^{-\epsilon}\sin \epsilon.$ Putting  these together, we have 
\begin{equation}\label{bound}
   \big|e^u\big|\leqslant 1\quad\text{and}\quad \frac{1}{|1+e^u|}< N \doteq  \max \bigg\{\frac{1}{1-e^{-\epsilon}}, 1,  \frac{1}{e^{-\epsilon}\sin\epsilon}\bigg\}
\end{equation}
for each $u\in W_{\theta,\delta}^-.$ Now, for the first term on the right hand side of (\ref{3.14}), by (\ref{fy})  we have $f_y(0)\doteq \log\Big(1+e^{y+\pi\mathbf i b^2}\Big),$ by (\ref{f}) we have $f'(y)=-\log\big(1+e^y\big),$ and by (\ref{bound}) we have 
\begin{equation}\label{e1}
\begin{split}
    \bigg|\frac{1}{2}\Big(f_y(0)+f'(y)\Big)\bigg| = & \frac{1}{2}\bigg|\log\Big(1+e^{y+\pi\mathbf ib^2}\Big)-\log\big(1+e^y\big)\bigg|\\
    =& \frac{1}{2} \bigg|\int_y^{y+\pi \mathbf i b^2} \frac{e^u}{1+e^u}du\bigg|\\
    \leqslant & \frac{1}{2}\int_y^{y+\pi\mathbf ib^2} \frac{|e^u|}{|1+e^u|} |du|< \frac{1}{2}\cdot  N \cdot  \big|\pi \mathbf ib^2\big|= \frac{N\pi}{2} |b|^2.
    \end{split}
\end{equation}
For the second term on the right hand side of (\ref{3.14}), by (\ref{f'y}) we have $f_y'(0)=2\pi\mathbf i b^2\frac{e^{y+\pi\mathbf ib^2}}{1+e^{y+\pi\mathbf ib^2}},$ and by (\ref{bound}) we have 
\begin{equation}\label{e2}
\begin{split}
 \bigg| -\frac{f'_y(0)}{12}\bigg|=\frac{1}{12}\big|2\pi\mathbf i b^2\big| \frac{\big|e^{y+\pi\mathbf ib^2}\big|}{\big|1+e^{y+\pi\mathbf ib^2}\big|}< \frac{N\pi}{6}|b|^2. 
\end{split}
\end{equation}
For the third term on the right hand side of (\ref{3.14}), by (\ref{f'y}) and (\ref{bound}) we have 
\begin{equation}\label{e3}
\begin{split}
 \bigg|-\frac{1}{2}\int_0^{+\infty}f_y''(x)P_2(x)dx \bigg| \leqslant &\frac{1}{2} \int_0^{+\infty} \big| - 4\pi^2 b^4 \big| \frac{\big|e^{y+(2x+1)\pi\mathbf ib^2}\big|}{\big|1+e^{y+(2x+1)\pi\mathbf ib^2}\big|^2}\big|P_2(x)\big|dx\\
 < &\frac{1}{2}\cdot \big|-4\pi^2b^4\big|\cdot N^2 \cdot \frac{1}{6}\cdot  \int_0^{+\infty} e^{\mathrm{Re}y-(2x+1)\pi\mathrm{Im}(b^2)}dx \\
 = & \frac{N^2 \pi^2 |b|^4}{3}  \frac{e^{\mathrm{Re}y-\pi\mathrm{Im}(b^2)}}{2\pi \mathrm{Im}(b^2)} <\frac{N^2\pi}{6\sin(2\theta)} |b|^2,
    \end{split}
\end{equation}
where the second inequality comes from  the fact that $\big|P_2(x)\big|\leqslant \frac{1}{6}$ for each $x\in [0,+\infty),$ and the last inequality comes from that $\mathrm{Re}y\leqslant 0$ for all $y\in W_{\theta,\delta}^-,$ $\mathrm {Im}(b^2)>0,$ and $\mathrm{Im}(b^2)=|b|^2\sin(2\theta).$ For the last term on the right hand side of (\ref{3.14}), by (\ref{f}) we have $f''(y)= -\frac{e^y}{1+e^y},$ and by (\ref{bound}) we have 
\begin{equation}\label{e4}
\begin{split}
 \bigg|\frac{1}{2\pi \mathbf i b^2}\int_y^{y+\pi\mathbf ib^2}f''(t) \Big(y+\pi\mathbf ib^2-t\Big)dt\bigg| \leqslant & \frac{1}{\big|2\pi \mathbf ib^2\big|} \int_y^{y+\pi\mathbf ib^2}\frac{\big|e^t\big|}{\big|1+e^t\big|}\Big|y+\pi\mathbf i b^2-t\Big| |dt| \\
 <&\frac{1}{\big|2\pi \mathbf ib^2\big|}\cdot  N \cdot \big|\pi\mathbf i b^2\big| \cdot \big|\pi\mathbf i b^2\big|  = \frac{N\pi}{2}|b|^2,
  \end{split}
\end{equation}
where the last inequality comes from the fact that $\big|y+\pi\mathbf ib^2-t\big|\leqslant \big|\pi\mathbf i b^2\big| $ for each $t$ on the line segment connecting $y$ and $y+\pi\mathbf ib^2.$

Finally, putting (\ref{3.14}), (\ref{e1}), (\ref{e2}), (\ref{e3}) and (\ref{e4}) together, we have (1) with $C=\frac{7N\pi}{6}+\frac{N^2\pi}{6\sin(2\theta)}.$
\\

For (2), we first observe that when $|b|$ is sufficiently small, we can let the constant $K=2$ in (\ref{K}), which is independent of $b;$ and by (\ref{M}) we have 
\begin{equation*}
\begin{split}
 \bigg|\log M_b\Big(\frac{y}{2\pi b}\Big)\bigg|\leqslant &\, 2 e^{\frac{\delta}{2}\mathrm{Im}(b^{-2})}\sum_{n=0}^{+\infty} e^{2n\pi\mathrm{Im}(b^{-2})}\\
< &\, 2 e^{\frac{\delta}{2}\mathrm{Im}(b^{-2})}\sum_{n=0}^{+\infty} e^{-2n\pi}
 = \frac{2e^{\frac{\delta}{2}\mathrm{Im}(b^{-2})}}{1-e^{-2\pi}},
 \end{split}
\end{equation*}
where the second inequality comes from that when $|b|$ is sufficiently small, $\mathrm{Im}(b^{-2})=-\sin(2\theta)|b|^{-2}<-1.$ As
$\lim_{|b|\to 0} \frac{ e^{\frac{\delta}{2}\mathrm{Im}(b^{-2})}}{|b|^2}=0,$ there is an $M=M_{\theta,\delta}>0$ such that 
$$ e^{\frac{\delta}{2}\mathrm{Im}(b^{-2})}< M|b|^2$$
for $b\in \mathbb C$ with $\arg b=\theta$ and $|b|$ sufficiently small, from which (2) follows with $D=\frac{2M}{1-e^{-2\pi}}.$
\end{proof}

\begin{proof}[Proof of Proposition \ref{EST2}]
Let $C$ and $D$ respectively be the constants in (1) and (2) of Lemma  \ref{N}. Then Proposition \ref{EST2} follows immediately from Lemma \ref{log=} and Lemma \ref{N} with $B=C+D.$
\end{proof}

Now we are ready to prove Proposition \ref{EST}.

\begin{proof}[Proof of Proposition \ref{EST}]

We first observe that, for each $x\in H_\theta,$
\begin{equation}\label{S-}
    \log S_b\bigg(\frac{\pi-x}{\pi b}+\frac{b}{2}\bigg)=-\log S_b\bigg(\frac{x}{\pi b}+\frac{b}{2}\bigg).
\end{equation}
Indeed, by e.g. \cite[Formula (A.16)]{TesVar}, $S_b(\frac{\pi-x}{\pi b}+\frac{b}{2})S_b(\frac{x}{\pi b}+\frac{b}{2})=1$ for each $x\in H_\theta,$ hence $\log S_b\big(\frac{\pi-x}{\pi b}+\frac{b}{2}\big)+\log S_b\big(\frac{x}{\pi b}+\frac{b}{2}\big)=2k\pi \mathbf i$ for some integer $k;$ and a direct computation at $x=\frac{\pi}{2}$ shows that $k=0.$

We also observe that, for each $x\in H=\mathbb C\setminus\big((-\infty,0]\cup [\pi,\infty)\big),$ 
\begin{equation}\label{L-}
    L(\pi-x)=-L(x).
\end{equation}
Indeed, for each $x\in (0,\pi),$ $L(\pi-x)+L(x)=-2\mathbf i\big(\Lambda(\pi-x)+\Lambda(x)\big)=0.$ Then (\ref{L-}) follows from the analyticity of $L$ on $H.$

In the rest of the proof, we will use the change of variable $y=-\mathbf i(2x-\pi).$ Then $x\in H_\theta$ if and only if $y\in V_\theta;$ and by (\ref{PhiS}) and a direct computation at $x=\frac{\pi}{2},$ we have
\begin{equation}\label{323}
2\pi\mathbf ib^2\log\Phi_b\bigg(\frac{y}{2\pi b}\bigg)=2\pi\mathbf ib^2\log S_b\bigg(\frac{x}{\pi b}+\frac{b}{2}\bigg)+x^2-\pi x +\frac{\pi^2}{6}-\frac{\pi^2b^4}{12}.    
\end{equation}

Now for the case that $x \in H_{\theta,\delta}$ with $\mathrm{Im}x \leqslant 0,$ we have $y\in V_{\theta,2\delta}^-.$ Then by (\ref{L-}), (\ref{323}), (\ref{eq:Lx})  and Proposition \ref{EST2}, we have
\begin{equation*}
\begin{split}
    \bigg|2\pi \mathbf i b^2 \log S_b\Big(\frac{x}{\pi b} +\frac{b}{2}\Big)-L(x)\bigg| = &\bigg|2\pi \mathbf i b^2 \log S_b\Big(\frac{x}{\pi b} +\frac{b}{2}\Big)+L(\pi-x)\bigg|\\
=& \Bigg|2\pi \mathbf i b^2\log\Phi_b\Big(\frac{y}{2\pi b}\Big)-\mathrm{Li}_2\big(-e^y\big)+\frac{\pi^2b^4}{12}\Bigg|
< \bigg(2\pi B'+\frac{\pi^2}{12}\bigg)|b|^4,
\end{split}
\end{equation*}
where $B'$ is the constant in Proposition \ref{EST2} determined by $\theta$ and $2\delta.$

For the case that  $x \in H_{\theta,\delta}$ with $\mathrm{Im}x \geqslant 0,$ we have $\pi - x \in H_{\theta,\delta}$ with $\mathrm{Im}(\pi-x)\leqslant 0;$ and the result follows from  (\ref{S-}), (\ref{L-}) and the previous case. 

Therefore, Proposition \ref{EST} holds for all $x\in  H_{\theta,\delta}$ with $B=2\pi B'+\frac{\pi^2}{12}.$
\end{proof}


\section{Asymptotics of complex $b$-$6j$ symbols}

The goal of this section is to prove  Theorem \ref{vol}. The main tool is the following Saddle Point Approximation from \cite[Proposition 6.1]{WY} (see also \cite[Proposition 2.22 and Remark 2.23]{LMSWY}).

\begin{proposition}\label{saddle}
Let $D$ be a region in $\mathbb C$ and let $f(z)$ and $g(z)$ be holomorphic functions on $D.$ Let $f_\hbar(z)$ be a holomorphic function of the form
$$ f_\hbar(z) = f(z) + \upsilon_\hbar(z)\hbar^2.$$
Let $S$ be a curve in $D$ and let $c$ be a point on $S.$ Suppose
\begin{enumerate}[(i)]
\item $c$ is a critical point of $f$ in $D,$
\item $\mathrm{Re}(f)(c) > \mathrm{Re}(f)(z)$ for all $z \in S\setminus \{c\},$
\item $f''(c)\neq 0,$ 
\item $g(c) \neq 0,$ 
\item $|\upsilon_\hbar(z)|$ is bounded from above by a constant independent of $\hbar$ in $D,$ and
\item $S$ is smooth  around $c.$ 
\end{enumerate}
Then
\begin{equation*}
\int_S g(z) e^{\frac{f_\hbar(z)}{\hbar}} dz= \big(2\pi \hbar\big)^{\frac{1}{2}}\frac{g(c)}{\sqrt{-f''(c)}}e^{\frac{f(c)}{\hbar}} \Big(1 + O \big(\hbar\big)\Big).
 \end{equation*}
\end{proposition}

\subsection{Properties of the relevant functions}

To prove Theorems \ref{vol}, following \cite{LMSWY2}, we introduce the following a new set of variables $\alpha_k,$ $\xi,$ $\tau_i,$ $\eta_j$:
\begin{align}
&\alpha_k=\pi b a_k-\frac{\pi b^2}{2} \textrm{ for }k\in\{1,\dots, 6\},\quad \xi=\pi b u, \label{alpha-a}\\
 &\tau_i=\pi b t_i-\frac{3\pi b^2}{2}\textrm{ for }i\in\{1,2,3,4\},\quad\text{and}\quad \eta_j=\pi b q_j-2\pi b^2 \textrm{ for }j\in\{1,2,3,4\}.\label{tau-t}
\end{align}
Then 
\begin{equation*}
\begin{split}
&\tau_1=\alpha_1+\alpha_2+\alpha_3,\quad\tau_2=\alpha_1+\alpha_5+\alpha_6,\quad
\tau_3=\alpha_2+\alpha_4+\alpha_6,\quad\tau_4=\alpha_3+\alpha_4+\alpha_5,\\
&\eta_1=\alpha_1+\alpha_2+\alpha_4+\alpha_5,\quad \eta_2=\alpha_1+\alpha_3+\alpha_4+\alpha_6,\quad 
\eta_3=\alpha_2+\alpha_3+\alpha_5+\alpha_6\quad\text{and}\quad \eta_4=2\pi;
\end{split}
\end{equation*}
and for $i,j\in\{1,2,3,4\},$ $\tau_i\in\mathbb R,$ $\eta_j\in\mathbb R,$ and
$$0<\eta_j-\tau_i < \pi.$$ 
In particular,
$$\max\{\tau_1, \tau_2, \tau_3, \tau_4\}<\min\{ \eta_1, \eta_2, \eta_3, \eta_4\}.$$

Let $\boldsymbol\alpha=(\alpha_1,\dots,\alpha_6)$ and let
\begin{equation*}
\begin{split}
U_{\boldsymbol \alpha,b}(\xi)= & -\pi \mathbf i b^2 \sum_{i=1}^4\sum_{j=1}^4 \log S_b(q_j-t_i) + 2\pi \mathbf i b^2 \sum_{i=1}^4  \log S_b(u-t_i) + 2\pi \mathbf i b^2 \sum _{j=1}^4\log S_b(q_j-u).
\end{split}
\end{equation*}
Then by Proposition \ref{contour}, we have 
\begin{equation}\label{fk}
\bigg\{\begin{matrix} a_1 & a_2 & a_3 \\ a_4 & a_5 & a_6 \end{matrix} \bigg\}_b=\frac{1}{\pi b}\int_{\Gamma}\exp\bigg(\frac{U_{\boldsymbol \alpha,b}(\xi)}{2\pi \mathbf ib^2} \bigg)d\xi,
\end{equation}
where $\Gamma$ is a contour as depicted in Figure \ref{Dtheta},  i.e., with $\theta=\arg b,$ 
\begin{enumerate}[(1)]
\item $\Gamma$ lies in the region 
\begin{equation*}
D_\theta=\Big\{\xi\in \mathbb C \ \Big|\ \xi-\max\{\tau_i\}\in H_\theta\ \ \text{and}\ \ \min\{\eta_j\}-\xi\in H_\theta \Big\},
\end{equation*}
which is also the region of $\xi\in \mathbb C$ with $\xi-\tau_i\in H_\theta$ for all $i\in\{1,2,3,4\}$ and $\eta_j-\xi\in H_\theta$ for all $j\in\{1,2,3,4\},$
and 
\item one end of $\Gamma$ is a ray of argument in $[2\theta, \pi],$ and the other end of $\Gamma$ is a ray of argument in $[2\theta-\pi,0].$ 
\end{enumerate}

\begin{figure}[htbp]
\centering\includegraphics[scale=0.30]{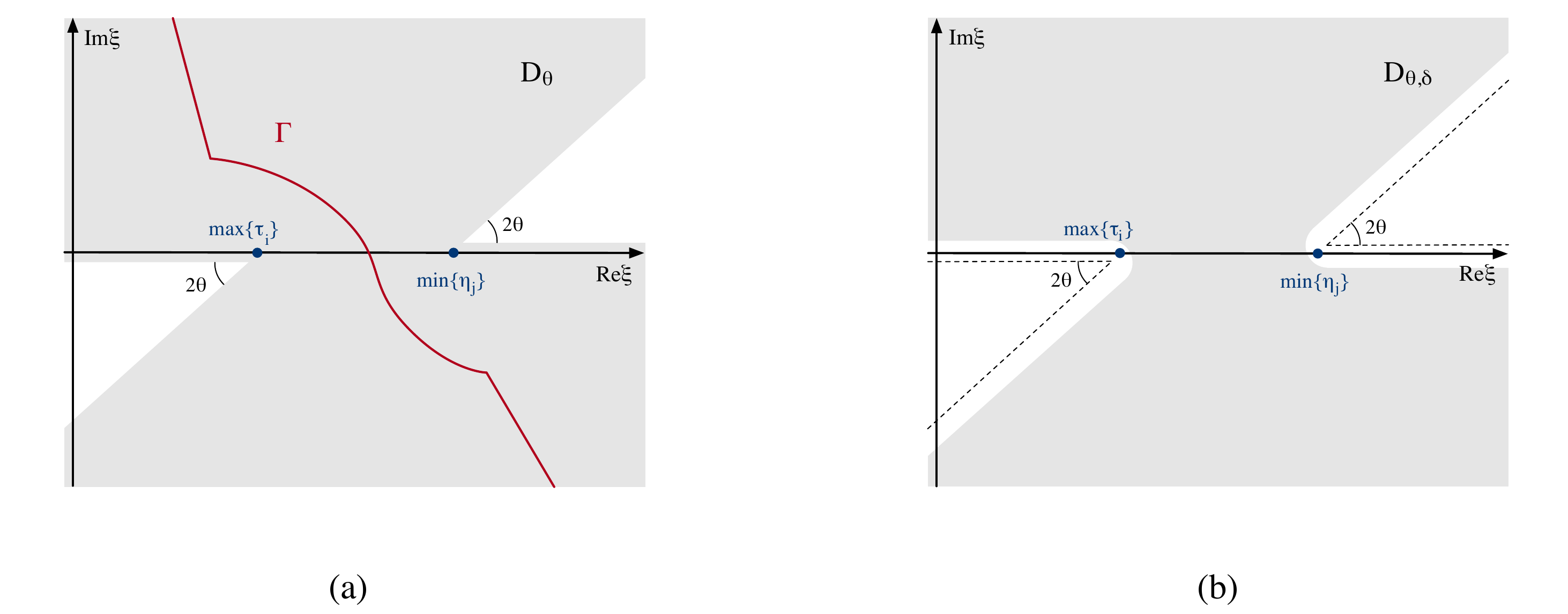}
\caption{The region $D_\theta$ and the contour $\Gamma$ of the integrand in~\eqref{fk} in (a), and the region $D_{\theta,\delta}$ in (b).}
\label{Dtheta}
\end{figure}

Define the function $U_{\boldsymbol \alpha}$ by 
\begin{equation}\label{U}
U_{\boldsymbol \alpha}(\xi)=-\frac{1}{2}\sum_{i=1}^4\sum_{j=1}^4L(\eta_j-\tau_i)+\sum_{i=1}^4L(\xi-\tau_i)+\sum_{j=1}^4L(\eta_j-\xi),
\end{equation}
where $L(x)$ is as  defined in Section \ref{SbLi}. Define also
\begin{equation}\label{kappa}
\kappa_{\boldsymbol \alpha}(\xi)=8\pi^2+14\pi\sum_{k=1}^6\alpha_k-28\pi\xi-4\pi \mathbf i\sum_{i=1}^4\log\Big(1-e^{2\mathbf i(\xi-\tau_i)}\Big)+3\pi \mathbf i\sum_{j=1}^4\log\Big(1-e^{2\mathbf i(\eta_j-\xi)}\Big),
\end{equation}
and 
\begin{equation}\label{nuU}\nu_{\boldsymbol\alpha,b}(\xi)=\frac{U_{\boldsymbol \alpha,b}(\xi)-U_{\boldsymbol \alpha}(\xi)-\kappa_{\boldsymbol \alpha}(\xi)b^2}{b^4}.
\end{equation} 
Then by (\ref{fk}) and (\ref{nuU}), the 
$b$-$6j$ symbol of $(a_1,\dots,a_6)$ is computed by
\begin{equation}\label{fk3} 
\bigg\{\begin{matrix} a_1 & a_2 & a_3 \\ a_4 & a_5 & a_6 \end{matrix} \bigg\}_b=\frac{1}{\pi b}\int_{\Gamma}\exp\bigg(\frac{U_{\boldsymbol \alpha}(\xi)+\kappa_{\boldsymbol \alpha}(\xi)b^2+
\nu_{\boldsymbol \alpha,b}(\xi)b^4}{2\pi \mathbf ib^2} \bigg)d\xi.
\end{equation}
\\

We first recall two important properties of the functions $U_{\boldsymbol\alpha}$ and $\kappa_{\boldsymbol\alpha}$ 
from \cite{LMSWY2}.

\begin{proposition}\cite[Proposition 3.11]{LMSWY2}\label{critical1} Let $(\theta_1,\dots, \theta_6)$ be the dihedral angles of a truncated hyperideal hyperbolic  tetrahedron $\Delta.$ Then the function $U_{\boldsymbol \alpha}(\xi)$ has a unique critical point $\xi^*$ on the interval $(\max\{\tau_i\},\min\{\eta_j\}),$ and 
$$U_{\boldsymbol \alpha}(\xi^*)=-2\mathbf i\mathrm{Vol}(\Delta),$$
where $\mathrm{Vol}(\Delta)$ is the hyperbolic volume of $\Delta.$
\end{proposition}

 \begin{proposition}\cite[Proposition 3.16]{LMSWY2}\label{Hess} At the critical point  $\xi^*$ of $U_{\boldsymbol\alpha},$ we have
 \begin{equation*}
\frac{-U''_{\boldsymbol \alpha}(\xi^*)}{\exp\big(\frac{\kappa_{\boldsymbol \alpha}(\xi^*)}{\pi \mathbf i}\big)}=16\sqrt{\det\mathrm{Gram}(\Delta)},
\end{equation*}
where $\mathrm{Gram}(\Delta)$ is the Gram matrix of $\Delta$ in the dihedral angles. As a consequence, $\xi^*$ is a non-degenerate critical point of $U_{\boldsymbol\alpha}.$
\end{proposition}
\bigskip

To understand the asymptotics of the complex $b$-$6j$ symbols, in below Propositions \ref{period}, \ref{ri}, \ref{limder1} and \ref{limder2}, we study properties of the function $U_{\boldsymbol\alpha}$ on the region 
$$D=\mathbb C \setminus (-\infty, \max\{\tau_i\})\cup(\min\{\eta_j\},\infty).$$

\begin{proposition}\label{period}
\begin{enumerate}[(1)]
    \item For $\xi\in D$ with $\mathrm{Im}\xi>0,$ we have
    $$U_{\boldsymbol\alpha}(\xi+\pi)=U_{\boldsymbol\alpha}(\xi)-4\pi^2.$$

    \item  For $\xi\in D$ with $\mathrm{Im}\xi<0,$ we have
   $$U_{\boldsymbol\alpha}(\xi+\pi)=U_{\boldsymbol\alpha}(\xi)+4\pi^2.$$
\end{enumerate}
As a consequence, the derivative
$U_{\boldsymbol\alpha}'$ is $\pi$-periodic in both cases.
\end{proposition}
\begin{proof} 
For (1), we have $\mathrm{Im}(\xi-\tau_i)>0$
for each $i\in \{1,2,3, 4\}$ 
and
$\mathrm{Im}(\eta_j-\xi)<0.$ Then by (\ref{period3}) and (\ref{period4}), we have 
$$L(\xi+\pi-\tau_i)=L(\xi-\tau_i)+2\pi(\xi-\tau_i)$$
for each $i,$ and 
$$L(\eta_j-\xi-\pi)=L(\eta_j-\xi)+2\pi(\eta_j-\xi-\pi)$$
for each $j.$
As a consequence, we have
\begin{equation*}
\begin{split}
U_{\boldsymbol\alpha}(\xi+\pi)=&-\frac{1}{2}\sum_{i=1}^4\sum_{j=1}^4L(\eta_j-\tau_i)+\sum_{i=1}^4L(\xi+\pi-\tau_i)+\sum_{j=1}^4L(\eta_j-\xi-
\pi)\\
=&-\frac{1}{2}\sum_{i=1}^4\sum_{j=1}^4L(\eta_j-\tau_i)+\sum_{i=1}^4L(\xi-\tau_i)+\sum_{j=1}^4L(\eta_j-\xi)\\
&+2\pi \sum_{i=1}^4 (\xi-\tau_i) + 2\pi \sum_{j=1}^4 (\eta_j-\xi-\pi)\\
=&\,U_{\boldsymbol\alpha}(\xi)-4\pi^2,
\end{split}   
\end{equation*}
where the last equality comes from that $\sum_{j=1}^4\eta_j=\sum_{i=1}^4\tau_i+2\pi.$
\\

For (2), by (\ref{period3}) and (\ref{period4}), we have 
$$L(\xi+\pi-\tau_i)=L(\xi-\tau_i)-2\pi(\xi-\tau_i)$$
for each $i\in\{1,2,3,4\},$
and 
$$L(\eta_j-\xi-\pi)=L(\eta_j-\xi)-2\pi(\eta_j-\xi-\pi)$$
for each $j\in\{1,2,3,4\};$
and the rest of the proof follows verbatim  that of (1).
\end{proof}

For the proof of  Propositions \ref{ri} and \ref{limder1}, we need the following two lemmas from \cite{LMSWY2}.

\begin{lemma}\cite[Lemma 3.13]{LMSWY2}\label{d} 
\begin{enumerate}[(1)] 
\item For $\alpha,\beta\in(0,\pi),$ we have $$ \frac{\partial}{\partial l}\Big|_{l=0}\Big( \mathrm{Im}L(\alpha+\mathbf il)+ \mathrm{Im}L(\beta-\mathbf il)\Big)=0.$$

\item For $\alpha,\beta\in(0,\pi)$ with $\alpha+\beta<\pi,$  we have $$\lim_{l\to +\infty} \frac{\partial}{\partial l}\Big( \mathrm{Im}L(\alpha+\mathbf il) +\mathrm{Im}L(\beta-\mathbf il)\Big)<0,$$
and
  $$ \lim_{l\to -\infty} \frac{\partial }{\partial l}\Big(\mathrm{Im}L(\alpha+\mathbf il)+\mathrm{Im}L(\beta-\mathbf il)\Big)>0.$$
\end{enumerate}
\end{lemma}

\begin{lemma}\cite[Lemma 3.14]{LMSWY2}\label{dd} 
Let $\alpha,\beta\in (0,\pi)$ such that $\alpha+\beta<\pi.$ 
\begin{enumerate}[(1)]
\item If $|\alpha-\beta|\leqslant \frac{\pi}{2},$ then
 $$\frac{\partial ^2}{\partial l^2}\Big(\mathrm{Im}L(\alpha+\mathbf il)+\mathrm{Im}L(\beta-\mathbf il)\Big) <0.$$

\item  If $|\alpha-\beta|>\frac{\pi}{2},$ then there is a unique $l_0>0$ such that:
\begin{enumerate}[(i)]
\item For $-l_0<l<l_0,$ 
$$\frac{\partial ^2}{\partial l^2}\Big(\mathrm{Im}L(\alpha+\mathbf il)+\mathrm{Im}L(\beta-\mathbf il)\Big)<0.$$ 

\item For $l<-l_0\ \text{or}\ l>l_0,$ 
$$\frac{\partial ^2}{\partial l^2}\Big(\mathrm{Im}L(\alpha+\mathbf il)+\mathrm{Im}L(\beta-\mathbf il)\Big)>0.$$
\end{enumerate}

\end{enumerate}
\end{lemma}

\begin{proposition}\label{ri}
\begin{enumerate}[(1)]
\item For  $\xi\in (\max\{\tau_i\},\min\{\eta_j\}),$ we have 
$$\frac{\partial\mathrm{Re}U_{\boldsymbol\alpha}(\xi)}{\partial\mathrm{Re}\xi}=0.$$
As a consequence, $\mathrm{Re}U_{\boldsymbol\alpha}(\xi)\equiv 0$ on $(\max\{\tau_i\},\min\{\eta_j\}).$

    \item Let $\xi^*\in (\max\{\tau_i\},\min\{\eta_j\})$  be the  critical point of $U_{\boldsymbol\alpha}.$ Then for $\xi\in(\max\{\tau_i\},\xi^*),$ we have 
$$-\frac{\partial\mathrm{Re}U_{\boldsymbol\alpha}(\xi)}{\partial\mathrm{Im}\xi}=\frac{\partial\mathrm{Im}U_{\boldsymbol\alpha}(\xi)}{\partial\mathrm{Re}\xi}<0;$$
and for  $\xi\in(\xi^*,\min\{\eta_j\}),$ we have 
$$-\frac{\partial\mathrm{Re}U_{\boldsymbol\alpha}(\xi)}{\partial\mathrm{Im}\xi}=\frac{\partial\mathrm{Im}U_{\boldsymbol\alpha}(\xi)}{\partial\mathrm{Re}\xi}>0.$$
\end{enumerate}
\end{proposition}

\begin{proof} For (1), since $\xi-\tau_i\in(0,\pi)$ and $\eta_j-\xi\in(0,\pi)$ for all  $i,j\in \{1,2,3,4\}$  and for all $\xi\in(\max\{\tau_i\},\min\{\eta_j\}),$ by the Cauchy-Riemann equation and Lemma \ref{d} (1), we have
\begin{equation*}
\begin{split}
    \frac{\partial\mathrm{Re}U_{\boldsymbol\alpha}(\xi)}{\partial\mathrm{Re}\xi} = \frac{\partial\mathrm{Im}U_{\boldsymbol\alpha}(\xi)}{\partial\mathrm{Im}\xi}
    = \sum_{i=1}^4 \frac{\partial}{\partial \mathrm{Im}\xi}\Big|_{\mathrm{Im}\xi=0}\Big(\mathrm{Im}L(\xi-\tau_i)+\mathrm{Im}L(\eta_i-\xi)\Big)=0.
    \end{split}
\end{equation*}
Since $\mathrm{Re}U_{\boldsymbol\alpha}(\xi^*)=0$ by Proposition \ref{critical1}, we have $\mathrm{Re}U_{\boldsymbol\alpha}(\xi)\equiv 0$ on $(\max\{\tau_i\},\min\{\eta_j\}).$
\medskip

For (2), since $\xi-\tau_i\in(0,\pi),$ $\eta_i-\xi\in(0,\pi)$  and $(\xi-\tau_i)+(\eta_j-\xi)=\eta_j-\tau_i<\pi$ for all  $i,j\in \{1,2,3,4\}$  and for all $\xi\in(\max\{\tau_i\},\min\{\eta_j\}),$ by Lemma \ref{dd}, we have 
$$\frac{\partial^2\mathrm{Im}U_{\boldsymbol\alpha}(\xi)}{\partial\mathrm{Re}\xi^2}=-\frac{\partial^2\mathrm{Im}U_{\boldsymbol\alpha}(\xi)}{\partial\mathrm{Im}\xi^2}=-\sum_{i=1}^4\frac{\partial^2}{\partial \mathrm{Im}\xi^2}\Big|_{\mathrm{Im}\xi=0}\Big(\mathrm{Im}L(\xi-\tau_i)+\mathrm{Im}L(\eta_i-\xi)\Big)>0.$$
As $\frac{\partial\mathrm{Im}U_{\boldsymbol\alpha}(\xi)}{\partial\mathrm{Re}\xi}=0$ at $\xi=\xi^*,$ we have 
$$\frac{\partial\mathrm{Im}U_{\boldsymbol\alpha}(\xi)}{\partial\mathrm{Re}\xi}<0$$ for $\xi\in (\max\{\tau_i\},\xi^*),$ and $$\frac{\partial\mathrm{Im}U_{\boldsymbol\alpha}(\xi)}{\partial\mathrm{Re}\xi}>0$$ for $\xi\in (\xi^*,\min\{\eta_j\}).$
Then the result follows from the Cauchy-Riemann equation.
\end{proof}

\begin{proposition} \label{limder1}
\begin{enumerate}[(1)]
\item For $\xi\in D$ with $\mathrm{Im}\xi>0,$ we have
$$\frac{\partial \mathrm{Re}U_{\boldsymbol\alpha}(\xi)}{\partial\mathrm{Re}\xi}=\frac{\partial\mathrm{Im}U_{\boldsymbol\alpha}(\xi)}{\partial\mathrm{Im}\xi}<0.$$
As a consequence,  on the region $\{ \xi \in D\ |\ \mathrm{Im}\xi >0\},$ $\mathrm{Im}U_{\boldsymbol\alpha}(\xi)$ is decreasing in $\mathrm{Im}\xi$ and $\mathrm{Re}U_{\boldsymbol\alpha}(\xi)$ is decreasing in $\mathrm{Re}\xi.$

\item For $\xi\in D$ with $\mathrm{Im}\xi<0,$ we have
$$\frac{\partial \mathrm{Re}U_{\boldsymbol\alpha}(\xi)}{\partial\mathrm{Re}\xi}=\frac{\partial\mathrm{Im}U_{\boldsymbol\alpha}(\xi)}{\partial\mathrm{Im}\xi}>0.$$
As a consequence, on the region $\{ \xi \in D\ |\ \mathrm{Im}\xi <0\}, $ $\mathrm{Im}U_{\boldsymbol\alpha}(\xi)$ is increasing in $\mathrm{Im}\xi$ and $\mathrm{Re}U_{\boldsymbol\alpha}(\xi)$ is increasing in $\mathrm{Re}\xi.$ 
\end{enumerate}
\end{proposition}

For the proof of Proposition \ref{limder1},  we also need the following lemma.

\begin{lemma} \label{limder3}
\begin{enumerate}[(1)]
    \item For $\alpha\in (0,\pi),$ we have
    $$-\pi <\frac{\partial}{\partial l} \mathrm{Im}L(\alpha \pm \mathbf il)< \pi.$$

    \item  For $\alpha\in [-\pi,0]$ and $l>0,$ we have
 $$\frac{\partial}{\partial l} \mathrm{Im}L(\alpha \pm \mathbf il)\leqslant-\pi;$$
and for $\alpha\in [-\pi,0]$ and $l<0,$ we have
 $$\frac{\partial}{\partial l} \mathrm{Im}L(\alpha \pm \mathbf il)\geqslant \pi.$$
\end{enumerate}
\end{lemma}

\begin{proof}
For (1) that $\alpha\in(0,\pi),$ by a direct computation, we have
\begin{equation}\label{limp+}
\frac{\partial}{\partial l} \mathrm{Im}L(\alpha+\mathbf il)=2\alpha-\pi-2\arg\Big(1-e^{-2l+2\mathbf i\alpha}\Big),
\end{equation}
and
\begin{equation}\label{limp-}
\frac{\partial}{\partial l} \mathrm{Im}L(\alpha-\mathbf il)=-2\alpha+\pi+2\arg\Big(1-e^{2l+2\mathbf i\alpha}\Big).
\end{equation}
Now, as  $\arg\Big(1-e^{\pm2l+2\mathbf i\alpha}\Big)$ are monotonic in $l,$ and 
\begin{equation}\label{lim}
\lim_{l\to \pm\infty}\arg\Big(1-e^{\mp2l+2\mathbf i\alpha}\Big)=0\quad\text{and}\quad \lim_{l\to \pm\infty}\arg\Big(1-e^{\pm2l+2\mathbf i\alpha}\Big)=2\alpha-\pi
\end{equation}
for $\alpha\in(0,\pi),$ we have
$$ 2\alpha-\pi-2\max\{0,2\alpha-\pi\}\leqslant \frac{\partial}{\partial l} \mathrm{Im}L(\alpha+\mathbf il)\leqslant 2\alpha-\pi-2\min\{0,2\alpha-\pi\},$$
and 
$$ -2\alpha+\pi+2\min\{0,2\alpha-\pi\}\leqslant \frac{\partial}{\partial l} \mathrm{Im}L(\alpha-\mathbf il)\leqslant -2\alpha+\pi+2\max\{0,2\alpha-\pi\}.$$
Finally, since $2\alpha-\pi-2\max\{0,2\alpha-\pi\}=\min\{2\alpha-\pi,\pi-2\alpha\}>-\pi,$ and $2\alpha-\pi-2\min\{0,2\alpha-\pi\}=\max\{2\alpha-\pi,\pi-2\alpha\}<\pi,$ the result follows.
\medskip

For (2), it suffices to consider the case that $\alpha\in (-\pi,0),$ then by continuity, the case that $\alpha\in\{-\pi,0\}$ follows. Let $\beta=\alpha+\pi \in(0,\pi).$ Then by (\ref{period3}) and (\ref{period4}),
\begin{equation}\label{per1}
    \mathrm{Im}L(\alpha\pm \mathbf il)=\mathrm{Im}L(\beta\pm\mathbf il)-2\pi l
    \end{equation}
if $l>0,$ and 
\begin{equation}\label{per2}
\mathrm{Im}L(\alpha\pm \mathbf il)=\mathrm{Im}L(\beta\pm\mathbf il)+2\pi l
\end{equation}
if $l<0.$
Together with (1), we have in the former case that 
$$\frac{\partial}{\partial l} \mathrm{Im}L(\alpha\pm\mathbf il) = \frac{\partial}{\partial l} \mathrm{Im}L(\beta\pm\mathbf il)-2\pi <\pi-2\pi=-\pi;$$
and in the latter case that 
$$\frac{\partial}{\partial l} \mathrm{Im}L(\alpha\pm\mathbf il) = \frac{\partial}{\partial l} \mathrm{Im}L(\beta\pm\mathbf il)+2\pi >-\pi+2\pi=\pi.$$
This completes the proof of (2).
\end{proof}

\begin{proof}[Proof of Proposition \ref{limder1}] 
Let $\{i_1,i_2,i_3,i_4\}$ and $\{j_1,j_2,j_3,j_4\}$ be  permutations of $\{1,2,3,4\}.$ For each $k\in\{1,2,3,4\},$ on each straight line in $D$ consisting of $\xi$  with a  fixed $\mathrm{Re}\xi,$  let 
$$\psi_k(\mathrm{Im}\xi)=\mathrm{Im}L(\xi-\tau_{i_k})+\mathrm{Im}L(\eta_{j_k}-\xi).$$
Then
$$\frac{\partial\mathrm{Im}U_{\boldsymbol \alpha}(\xi)}{\partial \mathrm{Im}\xi}=\sum_{k=1}^4\psi_k'(\mathrm{Im}\xi),$$
and it suffices to prove that for each $k\in\{1,2,3,4\},$ $\psi_k'(\mathrm{Im}\xi)<0$ when $\mathrm{Im}\xi>0,$ and $\psi_k'(\mathrm{Im}\xi)>0$ when $\mathrm{Im}\xi<0.$ 

To this end, for each $k,$ we observe that  $\mathrm{Re}(\xi-\tau_{i_k})+\mathrm{Re}(\eta_{j_k}-\xi)=\eta_{j_k}-\tau_{i_k}<\pi,$ and we consider the following cases:
\begin{enumerate}[(a)]
    \item $\mathrm{Re}(\xi-\tau_{i_k})\in(0,\pi)$ and $\mathrm{Re}(\eta_{j_k}-\xi)\in (0,\pi).$ In this case, by Lemma \ref{d} (1) and Lemma \ref{dd} with $\alpha=\mathrm{Re}(\xi-\tau_{i_k})$ and $\beta=\mathrm{Re}(\eta_{j_k}-\xi),$ we have $\psi_k'(0)=0$ and $\psi_k''(\mathrm{Im}\xi)<0$ for $\mathrm{Im}\xi\in [-l_0,l_0],$ with the understanding that $l_0=+\infty$ if $|\alpha-\beta|\leqslant \frac{\pi}{2}.$ Therefore, $\psi_k'(\mathrm{Im}\xi)>0$ for $\mathrm{Im}\xi\in [-l_0,0)$ and $\psi_k'(\mathrm{Im}\xi)<0$ for $\mathrm{Im}\xi\in (0, l_0].$  
Next, by Lemma \ref{d} (2) and Lemma \ref{dd}, 
$\psi_k'(\mathrm{Im}\xi)>\lim_{\mathrm{Im}\xi\to-\infty}\psi_k'(\mathrm{Im}\xi)>0$ for $\mathrm{Im}\xi<-l_0;$  and 
$\psi_k'(\mathrm{Im}\xi)<\lim_{\mathrm{Im}\xi\to+\infty}\psi_k'(\mathrm{Im}\xi)<0$ for $\mathrm{Im}\xi>l_0.$ 
As a consequence, $\psi_k'(\mathrm{Im}\xi)>0$ when $\mathrm{Im}\xi<0,$ and $\psi_k'(\mathrm{Im}\xi)<0$ when $\mathrm{Im}\xi>0.$ 

 \item $\mathrm{Re}(\xi-\tau_{i_k})\in(0,\pi)$ and $\mathrm{Re}(\eta_{j_k}-\xi)\in (-\pi,0].$ In this case, by Lemma \ref{limder3} (1) with 
$\alpha=\mathrm{Re}(\xi-\tau_{i_k}),$ we have $-\pi<\frac{\partial\mathrm{Im}L(\xi-\tau_{i_k})}{\partial \mathrm{Im}\xi}<\pi;$ and
 by Lemma \ref{limder3} (2) with $\alpha=\mathrm{Re}(\eta_{j_k}-\xi),$ we have 
    $\frac{\partial\mathrm{Im}L(\eta_{j_k}-\xi)}{\partial \mathrm{Im}\xi}\leqslant -\pi$ when
    $\mathrm{Im}\xi>0,$ and  $\frac{\partial\mathrm{Im}L(\eta_{j_k}-\xi)}{\partial \mathrm{Im}\xi}\geqslant \pi$ when
    $\mathrm{Im}\xi<0.$ As a consequence, we have $\psi_k'(\mathrm{Im}\xi)<0$ when $\mathrm{Im}\xi>0,$ and $\psi_k'(\mathrm{Im}\xi)>0$ when $\mathrm{Im}\xi<0.$

    \item   $\mathrm{Re}(\xi-\tau_{i_k})\in(-\pi,0]$ and $\mathrm{Re}(\eta_{j_k}-\xi)\in (0,\pi).$ In this case, by Lemma \ref{limder3} (2) with $\alpha=\mathrm{Re}(\xi-\tau_{i_k}),$ we have 
    $\frac{\partial\mathrm{Im}L(\xi-\tau_{i_k})}{\partial \mathrm{Im}\xi}\leqslant -\pi$ when
    $\mathrm{Im}\xi>0,$ and   $\frac{\partial\mathrm{Im}L(\xi-\tau_{i_k})}{\partial \mathrm{Im}\xi}\geqslant \pi$ when
    $\mathrm{Im}\xi<0;$ and by Lemma \ref{limder3} (1) with 
$\alpha=\mathrm{Re}(\eta_{j_k}-\xi),$ we have $-\pi<\frac{\partial\mathrm{Im}L(\eta_{j_k}-\xi)}{\partial \mathrm{Im}\xi}<\pi.$ As a consequence, we have $\psi_k'(\mathrm{Im}\xi)<0$ when $\mathrm{Im}\xi>0,$ and $\psi_k'(\mathrm{Im}\xi)>0$ when $\mathrm{Im}\xi<0.$ 
\end{enumerate}

Now by Proposition \ref{period}, it suffices to prove the result for $\xi$ in the region $\{ \xi\in D\ |\ c\leqslant \mathrm{Re}\xi\leqslant c+\pi\}$ for any $c\in\mathbb R.$ Let $\{i_1,i_2,i_3,i_4\}=\{j_1,j_2,j_3,j_4\}=\{1,2,3,4\}$ such that 
$\tau_{i_1}\leqslant \tau_{i_2}\leqslant \tau_{i_3}\leqslant \tau_{i_4}\leqslant \eta_{j_1}\leqslant \eta_{j_2}\leqslant \eta_{j_3}\leqslant \eta_{j_4.}$
Since $\eta_{j_k}-\tau_{i_k}<\pi$ for each $k,$ we have $\eta_{j_4}-\pi<\tau_{i_1},$ and we choose  $c\in (\eta_{j_4}-\pi, \tau_{i_1}).$ Let $\tau_{i_0}=c$ and $\eta_{j_5}=c+\pi.$ Then we consider the following cases:

\begin{enumerate}[(1)]
\item  $\mathrm{Re}\xi \in (\tau_{i_4},\eta_{j_1})=(\max\{\tau_i\},\min\{\eta_j\}).$ In this case, $\psi_k$ is in case (a) for each $k\in\{1,2,3,4\},$ and the result follows. 

\item $\mathrm{Re}\xi \in [\eta_{j_s},\eta_{j_{s+1}})$
for some $s\in\{1,2,3,4\}$ or $\mathrm{Re}\xi = \eta_{j_5}.$  In the former case, $\psi_k$ is in case (b) for $k\leqslant s,$ and is in case (a) for $k>s;$ and in the latter case, $\psi_k$ is in case (b) for all $k\in\{1,2,3,4\},$ and the result follows. The result follows in both cases.

\item $\mathrm{Re}\xi = \tau_{i_0}$ or $\mathrm{Re}\xi \in (\tau_{i_{s-1}},\tau_{i_s}]$
for some $s\in\{1,2,3,4\}.$ In the former case, $\psi_k$ is in case (c) for all $k\in\{1,2,3,4\};$ and  the latter case, $\psi_k$ is in case (c) for $k\geqslant s,$ and is in case (a) for $k<s.$ The result follows in both cases.
\end{enumerate}

Finally, by the Cauchy-Riemann equation, this immediately implies the estimate for $\frac{\partial\mathrm{Re}U_{\boldsymbol\alpha}}{\partial \mathrm{Re}\xi}.$
\end{proof}

\begin{proposition} \label{limder2} 
On $\Gamma_c=\big\{\xi\in D\ \big|\ \mathrm{Re}\xi =c  \big\}$ for each $c\in \mathbb R,$ we have 
$$\lim_{\mathrm{Im}\xi\to +\infty}\frac{\partial \mathrm{Im}U_{\boldsymbol \alpha}(\xi)}{\partial \mathrm{Im}\xi} = -4\pi\quad\text{and}\quad \lim_{\mathrm{Im}\xi\to -\infty}\frac{\partial \mathrm{Im}U_{\boldsymbol \alpha}(\xi)}{\partial \mathrm{Im}\xi} = 4\pi.$$
\end{proposition}


\begin{proof}
We first extend the arguments $\arg\big(1-e^{\mp2l+2\mathbf i\alpha}\big)$ in (\ref{limp+}) and (\ref{limp-}) continuously to $\alpha\in[0,\pi].$ Then by (\ref{lim}), 
we have 
\begin{equation}\label{eq1}
\lim_{l\to +\infty} \frac{\partial}{\partial l} \mathrm{Im}L(\alpha\pm\mathbf il)= 2\alpha-\pi,
\end{equation}
and
\begin{equation}\label{eq2}
\lim_{l\to -\infty} \frac{\partial}{\partial l} \mathrm{Im}L(\alpha\pm\mathbf il)=\pi-2\alpha
\end{equation}
for $\alpha\in [0,\pi].$ Next, for $\alpha\in[-\pi,0],$ let $\beta=\alpha+\pi\in[0,\pi].$ Then by (\ref{per1}) and (\ref{per2}), we have
\begin{equation}\label{eq3}
\lim_{l\to +\infty} \frac{\partial}{\partial l} \mathrm{Im}L(\alpha\pm\mathbf il)= \lim_{l\to +\infty} \frac{\partial}{\partial l} \mathrm{Im}L(\beta\pm\mathbf il)-2\pi= 2\beta-\pi-2\pi=2\alpha-\pi,
\end{equation}
and
\begin{equation}\label{eq4}
\lim_{l\to -\infty} \frac{\partial}{\partial l} \mathrm{Im}L(\alpha\pm\mathbf il)= \lim_{l\to -\infty} \frac{\partial}{\partial l} \mathrm{Im}L(\beta\pm\mathbf il)+2\pi= \pi-2\beta+2\pi=\pi-2\alpha.
\end{equation}
Now by Proposition \ref{period}, it suffices to prove the result for $\xi$ in the region $\{ \xi\in D\ |\ c\leqslant \mathrm{Re}\xi\leqslant c+\pi\}$ for some $c\in\mathbb R.$ We choose $c$ from the interval $(\max\{\eta_j\}-\pi, \min\{\tau_i\}).$ In this case, each $\mathrm{Re}(\xi-\tau_i)$ and $\mathrm{Re}(\eta_j-\xi)$ is either in $[-\pi,0]$ or in $[0,\pi].$ Then by (\ref{eq1}), (\ref{eq2}), (\ref{eq3}), (\ref{eq4}), we have
$$\lim_{\mathrm{Im}\xi\to +\infty}\frac{\partial \mathrm{Im}U_{\boldsymbol \alpha}(\xi)}{\partial \mathrm{Im}\xi} = \sum_{i=1}^4 \Big(2\mathrm{Re}(\xi- \tau_i)-\pi\Big)+\sum_{j=1}^4\Big(2\mathrm{Re}(\eta_j-\xi)-\pi\Big) =-4\pi,$$
and 
$$\lim_{\mathrm{Im}\xi\to -\infty}\frac{\partial \mathrm{Im}U_{\boldsymbol \alpha}(\xi)}{\partial \mathrm{Im}\xi} = \sum_{i=1}^4 \Big(\pi-2\mathrm{Re}(\xi- \tau_i)\Big)+\sum_{j=1}^4\Big(\pi-2\mathrm{Re}(\eta_j-\xi)\Big)  =4\pi,$$
where in each equation, the last equality comes from that $\sum_{j=1}^4\eta_j=\sum_{i=1}^4\tau_i+2\pi.$
\end{proof}
\bigskip


The next two propositions respectively analyze the functions $\kappa_{\boldsymbol\alpha}$ and $\nu_{\boldsymbol\alpha,b}$ on the suitable region, which will be needed in the proof of Theorem \ref{vol}. 

For $\theta \in (0, \frac{\pi}{2})$ and $\delta> 0$ sufficiently small, as depicted in Figure \ref{Dtheta} (b), let 
\begin{equation*}
D_{\theta,\delta}
=\Big\{\xi\in \mathbb C \ \Big|\ \xi-\max\{\tau_i\}\in H_{\theta,\delta}\ \ \text{and}\ \ \min\{\eta_j\}-\xi\in H_{\theta,\delta} \Big\},
\end{equation*}
which is also the region of $\xi\in \mathbb C$ with $\xi-\tau_i\in H_{\theta,\delta}$  for all $i\in\{1,2,3,4\}$ and $\eta_j-\xi\in H_{\theta,\delta}$  for all $j\in\{1,2,3,4\},$ and is also the region of $\xi\in D_\theta$ with $d(\xi,\partial D_\theta)>\delta.$   

\begin{proposition}\label{bound1} 
There exists a constant $K=K_{\delta}>0$ such that 
$$\Bigg|\frac{\partial \mathrm{Im}\kappa_{\boldsymbol\alpha}(\xi)}{\partial\mathrm{Im}\xi}\Bigg|<K$$
for all $\xi\in D_{\theta,\delta}.$

\end{proposition}

\begin{proof}
By a direct computation, we have
\begin{equation*}
\begin{split}
\big|\kappa_{\boldsymbol\alpha}'(\xi)\big|=&\,\Bigg|-28\pi-4\pi\mathbf i\sum_{i=1}^4\frac{d \log \big(1-e^{2\mathbf i(\xi-\tau_i)}\big)}{d\xi}+3\pi\mathbf i\sum_{j=1}^4\frac{d\log \big(1-e^{2\mathbf i(\eta_j-\xi)}\big)}{d\xi}\Bigg|\\
\leqslant & \,28\pi+4\pi\sum_{i=1}^4\bigg|\frac{2}{1-e^{-2\mathbf i(\xi-\tau_i)}}\bigg|+3\pi\sum_{j=1}^4\bigg|\frac{2}{1-e^{-2\mathbf i(\eta_j-\xi)}}\bigg|.
\end{split}
\end{equation*}
As for $\xi\in D_{\theta,\delta},$ each $\xi-\tau_i$ and $\eta_j-\xi$ is at least $\delta$-away from the set $\{ k\pi\ |\ k\in\mathbb Z\},$ the denominators in the terms above are bounded from below by some $\epsilon>0$ depending on $\delta.$ Letting $K=28\big(\pi+\frac{2}{\epsilon}\big),$ we have 
$$\bigg|\frac{\partial \mathrm{Im}\kappa_{\boldsymbol\alpha}(\xi)}{\partial\mathrm{Im}\xi}\bigg|\leqslant |\kappa_{\boldsymbol\alpha}'(\xi)|<K.$$
\end{proof}

\begin{proposition}\label{bound2} 
For $b\in\mathbb C$ with $|b|$ sufficiently small and $\arg b=\theta,$ there exists a constant $N=N_{\theta,\delta}>0$ independent of $b$ such that 
$$ \big|\nu_{\boldsymbol\alpha,b}(\xi)\big|<N$$
for all $\xi\in D_{\theta,\delta}.$
\end{proposition}

The proof of Proposition \ref{bound2} follows verbatim that of \cite[Proposition 3.11]{LMSWY} with only minor changes. For the readers' convenience, we include a proof here. 

\begin{proof}
    Recall from (\ref{nuU}) that \(\nu_{\boldsymbol\alpha,b}(\xi)=\frac{U_{\boldsymbol \alpha,b}(\xi)-U_{\boldsymbol \alpha}(\xi)-\kappa_{\boldsymbol \alpha}(\xi)b^2}{b^4}.
\)
Since $\eta_j-\tau_i$'s are in $H_{\theta,\delta}$ for all $i,j\in\{1,2,3,4\}$ and $\delta$ sufficiently small,  Proposition \ref{EST} with the choice of $B$ therein tells us
$$\Big| 2\pi \mathbf i b^2 \log S_b(q_j-t_i) - L(\eta_j-\tau_i)\Big|<B|b|^4.$$
 Next, for $\xi\in D_{\theta,\delta},$ all $\xi-\tau_i$'s  and $\eta_j - \xi$'s  are in the region $H_{\theta,\delta}.$ Then all  $\xi-\tau_i-2\pi b^2$'s and  $\eta_j-\xi+\frac{3\pi b^2}{2}$'s are in $H_{\theta,\frac{\delta}{2}}$ when $|b|$ is small enough,  and Proposition \ref{EST} with the choice of $B$ therein gives 
 $$\Big|2\pi \mathbf i b^2  \log  S_b(u-t_i) -L\big(\xi-\tau_i-2\pi b^2\big)\Big|<B|b|^4$$
and
$$\bigg|2\pi \mathbf i b^2  \log S_b(q_j-u) - L\Big(\eta_j-\xi+\frac{3\pi b^2}{2}\Big)\bigg|<B|b|^4.$$
As a consequence,  we have
\begin{equation}\label{1}
\Big| U_{\boldsymbol \alpha,b}(\xi)- V_{\boldsymbol \alpha,b}(\xi)\Big|<16B|b|^4,
\end{equation}
where
\begin{equation*}
V_{\boldsymbol \alpha,b}(\xi)=-\frac{1}{2}\sum_{i=1}^4\sum_{j=1}^4L(\eta_j-\tau_i)+\sum_{i=1}^4L\big(\xi-\tau_i-2\pi b^2\big)+\sum_{j=1}^4L\Big(\eta_j-\xi+\frac{3\pi b^2}{2}\Big).
\end{equation*}

We are now left to  estimate the difference between $V_{\boldsymbol \alpha,b}$ and $U_{\boldsymbol \alpha}+\kappa_{\boldsymbol \alpha}(\xi)b^2.$ On the one hand, we have
\begin{equation}\label{V-U-k}
\begin{split}
V_{\boldsymbol \alpha,b}-U_{\boldsymbol \alpha}-\kappa_{\boldsymbol \alpha}(\xi)b^2=& \sum_{i=1}^4\Big(L\big(\xi-\tau_i-2\pi b^2\big)-L\big(\xi-\tau_i\big)+ 2\pi b^2L'\big(\xi-\tau_i\big)\Big)\\
&+\sum_{j=1}^4\bigg(L\Big(\eta_j-\xi+\frac{3\pi b^2}{2}\Big)-L\big(\eta_j-\xi\big)-\frac{3\pi b^2}{2}L'\big(\eta_j-\xi\big)\bigg).
\end{split}
\end{equation}
On the other hand, we have
 $$L''(x)=2-\frac{4}{1-e^{-2\mathbf ix}},$$
 whose norm is bounded from above on the region $H_{\theta,\delta}.$ As a consequence, 
 $\big|L''\big(\xi-\tau_i\big)\big|$ and $\big|L''\big(\eta_j-\xi\big)\big|$ for all $i,j \in \{1,2,3,4\}$ are bounded from above on the region $D_{\theta,\delta}$ for any sufficiently small $\delta >0.$   For $b\in C$ with $|b|$ sufficiently small so that $\xi-2\pi b^2$ and $\xi+\frac{3\pi b^2}{2}$ are in $D_{\theta,\frac{\delta}{2}},$  we let $C$ be an upper bound of $|L''|$ on $D_{\theta,\frac{\delta}{2}}.$ Then  we have 
 \begin{equation*}
 \begin{split}
 &\Big|L\big(\xi-\tau_i-2\pi b^2\big)-L\big(\xi-\tau_i\big)+ 2\pi b^2L'\big(\xi-\tau_i\big)\Big|\\
 =&\bigg|\int_{\xi-\tau_i}^{\xi-\tau_i-2\pi b^2}L''(t)\big(\xi-\tau_i-2\pi b^2-t\big)dt\bigg|<4\pi^2|b|^4C
 \end{split}
 \end{equation*}
 for each $i\in\{1,2,3,4\},$ and 
  \begin{equation*}
 \begin{split}
&\bigg|L\Big(\eta_j-\xi+\frac{3\pi b^2}{2}\Big)-L\big(\eta_j-\xi\big)-\frac{3\pi b^2}{2} L'\big(\eta_j-\xi\big)\bigg|\\
=&\bigg|\int_{\eta_j-\xi}^{\eta_j-\xi+\frac{3\pi b^2}{2}}L''(t)\bigg(\eta_j-\xi+\frac{3\pi b^2}{2}-t\bigg)dt\bigg|<\frac{9}{4}\pi^2|b|^4C
 \end{split}
 \end{equation*}
 for each $j\in\{1,2,3,4\}.$
 Together with (\ref{V-U-k}), we have
 \begin{equation}\label{3}
 \Big|V_{\boldsymbol\alpha,b}(\xi)-U_{\boldsymbol \alpha}(\xi)-\kappa_{\boldsymbol \alpha}(\xi)b^2\Big|<25\pi^2|b|^4C.
 \end{equation}

Putting (\ref{1}) and (\ref{3}) together and letting $N =16B+25\pi^2C,$ we have
\begin{equation}\label{2}
\Big| U_{\boldsymbol \alpha,b}(\xi)- U_{\boldsymbol \alpha}(\xi)-\kappa_{\boldsymbol \alpha}(\xi)b^2\Big|<N |b|^4
\end{equation}
for all $\xi\in D_{\theta,\delta},$ and the result follows from (\ref{nuU}).
\end{proof}

\subsection{The Geometric Hypothesis}\label{GH}

Let $(\theta_1,\dots,\theta_6)$ be the dihedral angles of a hyperideal hyperbolic tetrahedron $\Delta.$ For each $k\in\{1,\dots,6\},$ let $\alpha_k=\frac{\pi}{2}\pm\frac{\theta_k}{2}$ with the  choice of the sign $+$ or $-$ arbitrarily, and let $\boldsymbol\alpha=(\alpha_1,\dots,\alpha_6).$ Let $U_{\boldsymbol\alpha}(\xi)$ be as defined in (\ref{U}). Then by Proposition \ref{critical1} and Proposition \ref{ri} (2), 
$$\mathrm{Im}U_{\boldsymbol\alpha}(\xi)\geqslant \mathrm{Im}U_{\boldsymbol\alpha}(\xi^*)=-2\mathrm{Vol}(\Delta)$$
for all $\xi$ in the interval  $(\max\{\tau_i\},\min\{\eta_j\}),$
 where $\xi^*$ is the  critical point of $U_{\boldsymbol\alpha}$ in  $(\max\{\tau_i\},\min\{\eta_j\});$ 
 and the equality holds if and only if $\xi=\xi^*.$ The Geometric Hypothesis needed in Theorem \ref{vol} (2) considers the behavior of $\mathrm{Im}U_{\boldsymbol\alpha}(\xi)$ on the whole $\mathbb R.$ To state the hypothesis, recall  that $\mathrm{Im}L(x)=-2\Lambda(x)$ for all $x\in (0,\pi),$ where $\Lambda:\mathbb R\to\mathbb R$ 
is the Lobachevsky function defined by
$$\Lambda(x)=-\int_0^x\log|2\sin t|dt.$$
Then by (\ref{period3}), (\ref{period4}), we can  continuously (but not $C^1$-smoothly) extend $\mathrm{Im}L(x)$ from $H=\mathbb C\setminus \big((-\infty,0)\cup(\pi,\infty)\big)$ to $\mathbb C$ by letting 
$$\mathrm{Im}L(x)=-2\Lambda(x)$$
for each $x\in\mathbb R.$
As a consequence, we can continuously extend $\mathrm{Im}U_{\boldsymbol\alpha}$  from $D=\mathbb C \setminus (-\infty, \max\{\tau_i\})\cup(\min\{\eta_j\},\infty)$ to $\mathbb C$ by letting 
$$\mathrm{Im}U_{\boldsymbol\alpha}(\xi)=\sum_{i=1}^4\sum_{j=1}^4\Lambda(\eta_j-\tau_i)-2\sum_{i=1}^4\Lambda(\xi-\tau_i)-2\sum_{j=1}^4\Lambda(\eta_j-\xi)$$
for each $\xi\in\mathbb R.$

\begin{condition} \label{per-min} Let $(\theta_1,\dots,\theta_6)$ be the dihedral angles of a hyperideal hyperbolic tetrahedron. For each $k\in\{1,\dots,6\},$ there is a choice of the sign $+$ or $-$ in $\alpha_k=\frac{\pi}{2}\pm\frac{\theta_k}{2}$ such that, with this choice of $\boldsymbol\alpha=(\alpha_1,\dots,\alpha_6),$
$$\mathrm{Im}U_{\boldsymbol\alpha}(\xi)\geqslant \mathrm{Im}U_{\boldsymbol\alpha}(\xi^*)$$
for all $\xi\in \mathbb R;$ and the equality holds if and only if $\xi=\xi^*+k\pi$ for some $k\in\mathbb Z.$
\end{condition}

\begin{remark}\label{geohyp1}
Numerical computations show that this Geometric Hypothesis is not always satisfied. For example, for a sufficiently small $\delta>0,$ $(\theta_1,\dots,\theta_6)=\big(\delta,\delta,\delta,\frac{\pi}{2}-\delta,\frac{\pi}{2}-\delta,\frac{\pi}{2}-\delta\big)$ provides a counter example. 
\end{remark}

\begin{remark}\label{geohyp2}
On the other hand, numerical computations suggest that  a very large portion of the dihedral angles $(\theta_1,\dots,\theta_6)$ should satisfy the Geometric Hypothesis. We examined the dihedral angles taking values from the set $\{0.01,0.26, 0.51, 0.76, 1.01, 1.26, 1.51, 1.76, 2.01, 2.26, 2.51, 2.76, 3.01\},$ evenly distributed in the space $[0,\pi]^6.$ There are in total  $104321$ values that form the set of dihedral angles of a hyperideal 
 hyperbolic tetrahedron, and $103213$ of them satisfy the Geometric Hypothesis, with a probability of $\sim98.94 \%.$ The same computations also suggest that as long as none of the $\theta_k$'s is as extremely small as $0.01,$ the Geometric Hypothesis is satisfied. 
\end{remark}




\subsection{Proof of Theorem \ref{vol}}

\begin{proof}[Proof of Theorem \ref{vol}] Let $\theta=\arg b.$ The case that  $\theta=0,$ i.e., $b$ is real, is proved in \cite[Theorem 1.5 and Theorem 3.18]{LMSWY2}. 
By Proposition \ref{conjugate}, it suffices to consider the case that $\theta \in (0, \frac{\pi}{2}),$ and the case that $\theta \in (-\frac{\pi}{2},0)$ follows directly. 

In the rest of the proof, we let 
\begin{equation}\label{W}
    W_{\boldsymbol\alpha}(\xi)\doteq e^{-2\mathbf i\theta}U_{\boldsymbol\alpha}(\xi)
\end{equation}
for each $\xi\in D_\theta.$ Then 
\begin{equation}\label{ImW}
\mathrm{Im}W_{\boldsymbol\alpha}(\xi) = \cos(2\theta)\mathrm{Im}U_{\boldsymbol\alpha}(\xi)-\sin(2\theta)\mathrm{Re}U_{\boldsymbol\alpha}(\xi).
  \end{equation}

\medskip

For (1) that $\theta \in (0,\frac{\pi}{4}],$ we  will consider the two cases  $\theta \in (0,\frac{\pi}{4})$ and  $\theta=\frac{\pi}{4}$ separately. 
\bigskip

For the case that $\theta\in (0,\frac{\pi}{4}),$ we let 
$$\mathbf v= \big(-\sin(2\theta),\mathbf \cos(2\theta)\big)$$
as a unit vector in $\mathbb R^2,$ and  by abuse of notations identify it with the complex number
$$\mathbf v= -\sin(2\theta)+\mathbf i\cos(2\theta).$$
 Then by (\ref{ImW}) and a direct computation,  the directional derivative 
\begin{equation}\label{pImW}
D_{\mathbf v} \mathrm{Im}W_{\boldsymbol\alpha}(\xi)=\frac{\partial \mathrm{Im}U_{\boldsymbol\alpha}(\xi)}{\partial\mathrm{Im}\xi}
\end{equation}
for each $\xi\in D_\theta.$

Let $\delta>0$ and $d>0$ both be sufficiently small so that the region 
 $$B_{d}=\Big\{\xi\in \mathbb C \ \Big|\ |\mathrm{Re}\xi - \xi^*| \leqslant d\tan(2\theta)+\delta\ \ 
 \text{and}\ \  |\mathrm{Im}\xi|\leqslant d \Big\}$$
 lies entirely in $D_{\theta,\delta},$ where $\xi^*$ is the critical point of $U_{\boldsymbol\alpha}$ given by Proposition \ref{critical1}. 
By (\ref{pImW} and Propositions \ref{period} and  \ref{limder2},  there is an $L>d$ such that 
\begin{equation}\label{Dv1}
D_{\mathbf v}\mathrm{Im}W_{\boldsymbol\alpha} (\xi^*+  t\mathbf v)<-2\pi\quad\text{and}\quad D_{\mathbf v}\mathrm{Im}W_{\boldsymbol\alpha} (\xi^* - t\mathbf v)>2\pi
\end{equation}
for $t>\frac{L}{\cos(2\theta)}.$ Then as depicted in Figure \ref{G1}, we consider the contour 
$$\Gamma^*=\Big\{\xi^*+t\mathbf v\ \Big|\ t\in \mathbb R\Big\};$$
and by (\ref{W}) and  (\ref{fk3}), we have
\begin{equation*}
\bigg\{\begin{matrix} a_1 & a_2 & a_3 \\ a_4 & a_5 & a_6 \end{matrix} \bigg\}_b=\frac{1}{\pi b}\int_{\Gamma^*}\exp\bigg(\frac{W_{\boldsymbol \alpha}(\xi)+\kappa_{\boldsymbol\alpha}(\xi)|b|^2+e^{2\mathbf i\theta}\nu_{\boldsymbol\alpha,b}(\xi)|b|^4}{2\pi \mathbf i|b|^2} \bigg)d\xi.
\end{equation*}

\begin{figure}[htbp]
\centering
\includegraphics[scale=0.4]{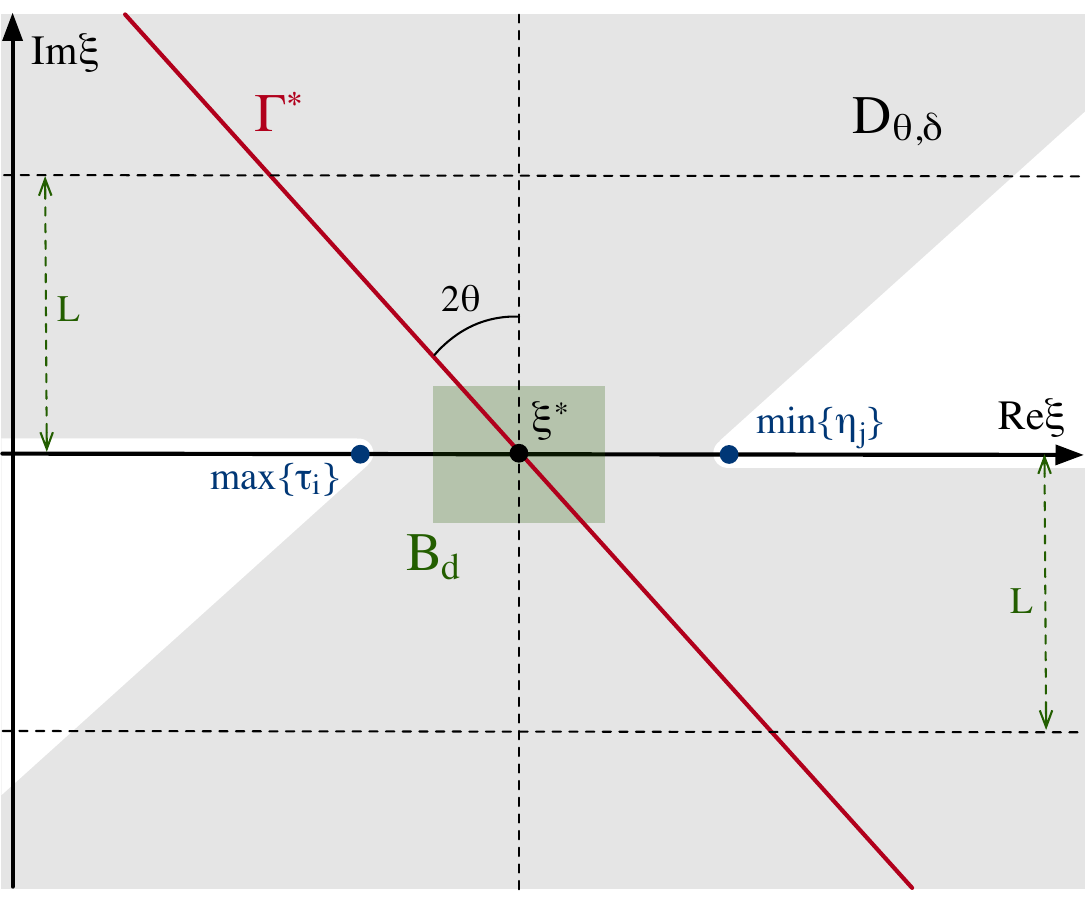}
\caption{The contour $\Gamma^*$ in the case that $\theta\in(0,\frac{\pi}{4}),$ where $\Gamma^*_d$ is the piece of $\Gamma^*$ lying in the light green region $B_d,$ and  $\Gamma^*_L$ is the piece of $\Gamma^*$ lying in between the two horizontal dotted lines.}
\label{G1}
\end{figure}

Let
$$\Gamma^*_d=\Gamma^*\cap B_d=\Big\{ \xi \in \Gamma^*\ \Big|\ |\mathrm{Im}\xi|<d\Big\}$$
and let 
$$\Gamma^*_L=\Big\{ \xi \in \Gamma^*\ \Big|\ |\mathrm{Im}\xi|\leqslant L\Big\}.$$
We will show that, as $b\to 0$ with $\arg b=\theta\in(0,\frac{\pi}{4}),$
\begin{enumerate}[(I)]
\item 
\begin{equation*}
    \begin{split}
      \frac{1}{\pi b}\int_{\Gamma^*_d}\exp\bigg(\frac{W_{\boldsymbol \alpha}(\xi)+\kappa_{\boldsymbol\alpha}(\xi)|b|^2+e^{2\mathbf i\theta}\nu_{\boldsymbol\alpha,b}(\xi)|b|^4}{2\pi \mathbf i|b|^2}  \bigg)d\xi =   \frac{e^{\frac{-\mathrm{Vol}(\Delta)}{\pi b^2}}}{2\sqrt[4]{-\det\mathrm{Gram}(\Delta)}}\Big(1+O\big(|b|^2\big)\Big),  
    \end{split}
\end{equation*}

\item $$\bigg|\frac{1}{\pi b}\int_{\Gamma^*_L\setminus \Gamma^*_d}\exp\bigg(\frac{W_{\boldsymbol \alpha}(\xi)+\kappa_{\boldsymbol\alpha}(\xi)|b|^2+e^{2\mathbf i\theta}\nu_{\boldsymbol\alpha,b}(\xi)|b|^4}{2\pi \mathbf i|b|^2}  \bigg)d\xi\bigg|< O\Big(e^{\frac{-\cos(2\theta)\mathrm{Vol}(\Delta)-\epsilon}{\pi |b|^2}}\Big),$$
and 

\item $$\bigg|\frac{1}{\pi b}\int_{\Gamma^*\setminus \Gamma^*_L}\exp\bigg(\frac{W_{\boldsymbol \alpha}(\xi)+\kappa_{\boldsymbol\alpha}(\xi)|b|^2+e^{2\mathbf i\theta}\nu_{\boldsymbol\alpha,b}(\xi)|b|^4}{2\pi \mathbf i|b|^2}  \bigg)d\xi\bigg|<O\Big(e^{\frac{-\cos(2\theta)\mathrm{Vol}(\Delta)-\epsilon}{\pi |b|^2}}\Big)$$
\end{enumerate}
for some $\epsilon>0,$ from which the result follows.
\medskip

For (I), we claim that  all the conditions of Proposition \ref{saddle} are satisfied by letting $\hbar=|b|^2,$ $D=B_d,$ $f=\frac{W_{\boldsymbol \alpha}}{2\pi \mathbf i},$ $g=\exp\big(\frac{\kappa_{\boldsymbol \alpha}}{2\pi \mathbf i }\big),$ $f_\hbar=\frac{W_{\boldsymbol \alpha}+e^{2\mathbf i\theta}\nu_{\boldsymbol \alpha,b}|b|^4}{2\pi \mathbf i},$ $\upsilon_h=\frac{e^{2\mathbf i\theta}\nu_{\boldsymbol \alpha,b}}{2\pi \mathbf i},$ $S=\Gamma^*_d$ and $c=\xi^*.$ Indeed, by Proposition \ref{critical1}, $\xi^*$ is a critical point of $f=\frac{W_{\boldsymbol \alpha}}{2\pi \mathbf i}=\frac{e^{-2\mathbf i \theta}U_{\boldsymbol \alpha}}{2\pi \mathbf i}$ in $B_{d},$  hence condition (i) is satisfied.
By (\ref{pImW}) and Proposition \ref{limder1}, $\xi^*$ is the unique maximum point of $\mathrm{Re}f=\frac{\mathrm{Im}W_{\boldsymbol\alpha}}{2\pi}$ on $\Gamma^*_d,$ hence condition (ii) is satisfied; and by Proposition \ref{Hess}, $\xi^*$ is a non-degenerate critical point of $f,$ and condition (iii) is satisfied. For condition (iv), since $\xi^*\neq \tau_i$ and $\xi^*\neq\eta_j$ for any $i$ and $j$ in $\{1,2,3,4\},$ $\kappa_{\boldsymbol \alpha}(\xi^*)$ is a finite value. As a consequence, $g(\xi^*)=\exp\big(\frac{\kappa_{\boldsymbol \alpha}(\xi^*)}{2\pi \mathbf i }\big)\neq 0,$ and condition (iv) is satisfied.
For condition (v), by Proposition \ref{bound2},  $|\upsilon_{\hbar}(\xi)|=\big|\frac{e^{2\mathbf i\theta}\nu_{\boldsymbol \alpha,b}(\xi)}{2\pi \mathbf i}\big|<\frac{N}{2\pi}$ on $B_{d}.$ For condition (vi), since $\Gamma^*$ is a straight line, it is smooth near $\xi^*.$  Finally, by Proposition \ref{saddle}, Proposition \ref{critical1}  and Proposition \ref{Hess}, we have as $b\to 0,$
\begin{equation*}
\begin{split}
\frac{1}{\pi b}\int_{\Gamma^*_d}\exp\bigg(\frac{U_{\boldsymbol \alpha,b}(\xi)}{2\pi \mathbf ib^2}\bigg) d\xi=& \frac{(2\pi |b|^2)^\frac{1}{2}}{\pi b} \frac{\exp\big(\frac{\kappa_{\boldsymbol \alpha}(\xi^*)}{2\pi \mathbf i}\big)}{\sqrt{-\frac{e^{-2\mathbf i\theta}U_{\boldsymbol \alpha}''(\xi^*)}{2\pi \mathbf i}}}e^{\frac{e^{-2\mathbf i\theta}U_{\boldsymbol\alpha}(\xi^*)}{2\pi \mathbf i |b|^2}}\Big(1+O\big(|b|^2\big)\Big)\\
=&\frac{1}{2}\frac{e^{\frac{-\mathrm{Vol}(\Delta)}{\pi b^2}}}{\sqrt[4]{-\det\mathrm{Gram}(\Delta)}}\Big(1+O\big(|b|^2\big)\Big).
\end{split}
\end{equation*}
This completes the proof of (I). 
\medskip

For  (II), by (\ref{ImW}), (\ref{pImW}) and Propositions \ref{critical1} and \ref{limder1}, there is an $\epsilon_1>0$ such that  
 \begin{equation}\label{ImU}
\mathrm{Im} W_{\boldsymbol\alpha}(\xi^*\pm t \mathbf v )< -2\cos(2\theta)\mathrm{Vol}(\Delta)-4\epsilon_1
\end{equation}
 for $t>d.$ Then by Proposition \ref{bound2},  there is a  $b_1>0$ such that 
\begin{equation}\label{last2}
\mathrm{Im}\big(e^{2\mathbf i\theta}\nu_{\boldsymbol\alpha,b}(\xi)\big)|b|^4\leqslant |\nu_{\boldsymbol\alpha,b}(\xi)||b|^4<N|b|^4<\epsilon_1 
\end{equation}
for all $|b|<b_1$ and for all $\xi\in\Gamma^*;$  and together with  (\ref{ImU}), we have
\begin{equation}\label{last3}
\mathrm{Im}W_{\boldsymbol \alpha}(\xi)+\mathrm{Im}  \big(e^{2\mathbf i\theta}\nu_{\boldsymbol \alpha,b}(\xi)\big)|b|^4<  -2\cos(2\theta)\mathrm{Vol}(\Delta) - 3\epsilon_1
\end{equation}
for all $|b|<b_1$ and  $\xi\in\Gamma^*\setminus \Gamma^*_d.$ By the compactness of ${\Gamma^*_L}\setminus \Gamma^*_d,$ there exists an $M_1>0$ such that 
\begin{equation}\label{Imk}
\mathrm{Im}\kappa_{\boldsymbol\alpha}(\xi)<M_1
\end{equation}
for all $\xi\in\Gamma^*_L\setminus \Gamma^*_d.$ As a consequence of (\ref{last2}) and (\ref{Imk}), we have
\begin{equation*}
\begin{split}
& \bigg|\frac{1}{\pi b}\int_{\Gamma_L^*\setminus \Gamma^*_d}\exp\bigg(\frac{W_{\boldsymbol \alpha}(\xi)+\kappa_{\boldsymbol\alpha}(\xi)|b|^2+e^{2\mathbf i\theta}\nu_{\boldsymbol\alpha,b}(\xi)|b|^4}{2\pi \mathbf i|b|^2}   \bigg)d\xi\bigg|\\
<&\frac{1}{\pi |b|}\int_{\Gamma^*_L\setminus \Gamma^*_d}\exp\bigg(\frac{\mathrm{Im}W_{\boldsymbol \alpha}(\xi)+\mathrm{Im}\kappa_{\boldsymbol \alpha}(\xi)|b|
^2+\mathrm{Im}\big(e^{2\mathbf i\theta}\nu_{\boldsymbol \alpha,b}(\xi)\big)|b|^4}{2\pi |b|^2} \bigg)|d\xi|\\
<  & \frac{1}{\pi |b|} \frac{2(L-d)}{\cos(2\theta)} e^{\frac{M_1}{2\pi}} \exp\bigg(\frac{-\cos(2\theta)\mathrm{Vol}(\Delta)-\epsilon_1}{\pi |b|^2}   \bigg )\\
< & O\Big(e^{\frac{-\cos(2\theta)\mathrm{Vol}(\Delta)-\epsilon}{\pi |b|^2}}\Big)
\end{split}
\end{equation*}
for any $\epsilon<\epsilon_1.$ This completes the proof of (II). 
\medskip

For (III), let $b_1$ be as in the proof of (II) above and let $\epsilon_1$ be as in (\ref{ImU}). Then there  is a $b_0\in (0, b_1)$ such that for all $|b|<b_0,$
\begin{equation}\label{ImK}
\mathrm{Im}\kappa_{\boldsymbol\alpha}\bigg(\xi^*\pm \frac{L}{\cos(2\theta)} \mathbf v\bigg) |b|^2<\epsilon_1,
\end{equation}
$K|b|^2<\epsilon_1$ and $N|b|^4<\epsilon_1,$ where $K$ and $N$ are respectively the constants in Propositions \ref{bound1} and \ref{bound2}. We claim that, for $\xi\in\Gamma^*\setminus \Gamma^*_L$ and $|b|<b_0,$ 
\begin{equation}\label{cl}
\begin{split}
    \mathrm{Im}W_{\boldsymbol \alpha}(\xi)&+\mathrm{Im}\kappa_{\boldsymbol \alpha}(\xi)|b|^2+\mathrm{Im}\big(e^{2\mathbf i\theta}\nu_{\boldsymbol \alpha,b}(\xi)\big)|b|^4\\<&-2\cos(2\theta)\mathrm{Vol}(\Delta)-2\epsilon_1-(2\pi-\epsilon_1)\bigg(|\xi-\xi^*|-\frac{L}{\cos(2\theta)}\bigg),
\end{split}
\end{equation}
as a consequence of which we have
\begin{equation}\label{CI}
\begin{split}
& \bigg|\frac{1}{\pi b}\int_{\Gamma^*\setminus \Gamma^*_L}\exp\bigg(\frac{W_{\boldsymbol \alpha}(\xi)+\kappa_{\boldsymbol\alpha}(\xi)|b|^2+e^{2\mathbf i\theta}\nu_{\boldsymbol\alpha,b}(\xi)|b|^4}{2\pi \mathbf i|b|^2}   \bigg)d\xi\bigg|\\
<
& \frac{1}{\pi |b|}\int_{\Gamma^*\setminus \Gamma^*_L}\exp\bigg(\frac{\mathrm{Im}W_{\boldsymbol \alpha}(\xi)+\mathrm{Im}\kappa_{\boldsymbol \alpha}(\xi)|b|^2+\mathrm{Im}\big(e^{2\mathbf i\theta}\nu_{\boldsymbol \alpha,b}(\xi)\big)|b|^4}{2\pi |b|^2} \bigg)|d\xi| \\
 < & \frac{1}{\pi |b|} \exp\bigg(\frac{-\cos(2\theta)\mathrm{Vol}(\Delta)-\epsilon_1}{\pi |b|^2}\bigg)\int_{\Gamma^*\setminus \Gamma^*_L}\exp\bigg(\frac{-(2\pi-\epsilon_1)\big(|\xi-\xi^*|-\frac{L}{\cos(2\theta)}\big)}{2\pi}\bigg) |d\xi|\\
< & O\Big(e^{\frac{-\cos(2\theta)\mathrm{Vol}(\Delta)-\epsilon}{\pi |b|^2}}\Big)
\end{split}
\end{equation}
for any $\epsilon<\epsilon_1.$ For the proof of the claim, by (\ref{Dv1}), Proposition \ref{bound1} and the choice of $b_0,$ for $t>\frac{L}{\cos(2\theta)}$ and $|b|<b_0,$ we have 
$$\frac{d}{dt}  \Big(\mathrm{Im}W_{\boldsymbol\alpha}(\xi^*\pm t\mathbf v)+\mathrm{Im}\kappa_{\boldsymbol\alpha}(\xi^*\pm t\mathbf v)|b|^2\Big)<-2\pi+K|b|^2<-2\pi+\epsilon_1.$$
Together with the Mean Value Theorem, (\ref{ImU}) and (\ref{ImK}), we have 
\begin{equation}\label{Bou}
\begin{split}
&\mathrm{Im}W_{\boldsymbol\alpha}(\xi)+ \mathrm{Im}\kappa_{\boldsymbol\alpha}(\xi)|b|^2\\
< & \mathrm{Im}W_{\boldsymbol\alpha}\bigg(\xi^*\pm \frac{L}{\cos(2\theta)}\mathbf v\bigg) + \mathrm{Im}
\kappa_{\boldsymbol\alpha}\bigg(\xi^*\pm \frac{L}{\cos(2\theta)}\mathbf v\bigg)|b|^2 - (2\pi-\epsilon_1) \bigg|  \xi - \bigg(\xi^*\pm \frac{L}{\cos(2\theta)}\mathbf v \bigg) \bigg|\\
< & -2\cos(2\theta)\mathrm{Vol}(\Delta)-3\epsilon_1  -(2\pi-\epsilon_1)\bigg(|\xi-\xi^*|-\frac{L}{\cos(2\theta)} \bigg)
\end{split}
\end{equation}
for all $\xi \in \Gamma^*\setminus \Gamma^*_L.$  Finally, putting (\ref{Bou}) and (\ref{last2}) together, we have  (\ref{cl}) and the second inequality in (\ref{CI}); and since 
$$|\xi-\xi^*|-\frac{L}{\cos(2\theta)}\to+\infty$$
as $\xi \in \Gamma^*\setminus \Gamma^*_L$ approaches $\infty,$ we have the last inequality in (\ref{CI}). This completes the proof of (III).
\smallskip

Putting (I), (II)  and (III) together, and as $\mathrm{Re}\Big(\frac{-\mathrm{Vol}(\Delta)}{\pi b^2}\Big)=\frac{-\cos(2\theta)\mathrm{Vol}(\Delta)}{\pi |b|^2},$ we have as $b\to 0$ with a fixed $\arg b \in (0,\frac{\pi}{4}),$ 
$$\bigg\{\begin{matrix} a_1 & a_2 & a_3 \\ a_4 & a_5 & a_6 \end{matrix} \bigg\}_b=\frac{1}{2}\frac{e^{\frac{-\mathrm{Vol}(\Delta)}{\pi b^2}}}{\sqrt[4]{-\det\mathrm{Gram}(\Delta)}}\Big(1+O\big(b^2\big)\Big). $$ 
\bigskip

For the case that $\theta=\frac{\pi}{4},$ by (\ref{W}) and (\ref{ImW}), we have
\begin{equation}\label{ImW-ReU1}
   W_{\boldsymbol\alpha}(\xi) = - \mathbf iU_{\boldsymbol\alpha}(\xi),
\end{equation}
and 
\begin{equation}\label{ImW-ReU}
    \mathrm{Im}W_{\boldsymbol\alpha}(\xi) = -\mathrm{Re}U_{\boldsymbol\alpha}(\xi).
\end{equation}
Let $d>0$ be  sufficiently small so that the region 
 $$B_{d}=\Big\{\xi\in \mathbb C \ \Big|\ |\mathrm{Re}\xi - \xi^*| \leqslant d\ \ 
 \text{and}\ \ |\mathrm{Im}\xi|\leqslant d \Big\}$$
 lies entirely in $D_{\theta},$ where $\xi^*$ is the critical point of $U_{\boldsymbol\alpha}$ given by Proposition \ref{critical1}.  Let $\alpha>0$ be sufficiently small so that 
 $$\max\{\tau_i\}+\alpha<\xi^*-d\quad\text{and}\quad \xi^*+d< \min\{\eta_j\}-\alpha,$$
 and let 
 $$\Gamma_1=\Big[\max\{\tau_i\}+\alpha, \min\{\eta_j\}-\alpha\Big].$$
 By Proposition \ref{ri} (2), there is a $l>0$ such that for each $t\in[0,l],$ 
 \begin{equation}\label{RI}
     -\frac{\partial \mathrm{Re}U_{\boldsymbol\alpha}}{\partial \mathrm{Im}\xi}\bigg|_{\xi=\max\{\tau_i\}+\alpha+\mathbf i t}<0\quad\text{and}\quad -\frac{\partial \mathrm{Re}U_{\boldsymbol\alpha}}{\partial \mathrm{Im}\xi}\bigg|_{\xi=\min\{\eta_j\}-\alpha-\mathbf i t}>0.
 \end{equation}
Let 
 $$\Gamma_2=\Big\{ \max\{\tau_i\}+\alpha+\mathbf i t\ \Big|\ t\in [0,l]\Big\}\cup \Big\{ \min\{\eta_j\}-\alpha-\mathbf i t\ \Big|\ t\in [0,l]\Big\}$$
 and 
 $$\Gamma_3=\Big\{t+\mathbf il \ \Big|\ t\in \big(-\infty,\max\{\tau_i\}+\alpha \big]\Big\}\cup \Big\{ t-\mathbf il \ \Big|\ t\in\big[\min\{\eta_j\}-\alpha,\infty\big) \Big\}.$$
As depicted in Figure \ref{G2} (a), let $$\Gamma=\Gamma_1\cup \Gamma_2\cup \Gamma_3,$$
and let $\delta>0$ be sufficiently small so that $\Gamma$ lies entirely in the region $D_{\theta,\delta}.$
\begin{figure}[htbp]
\centering
\includegraphics[scale=0.35]{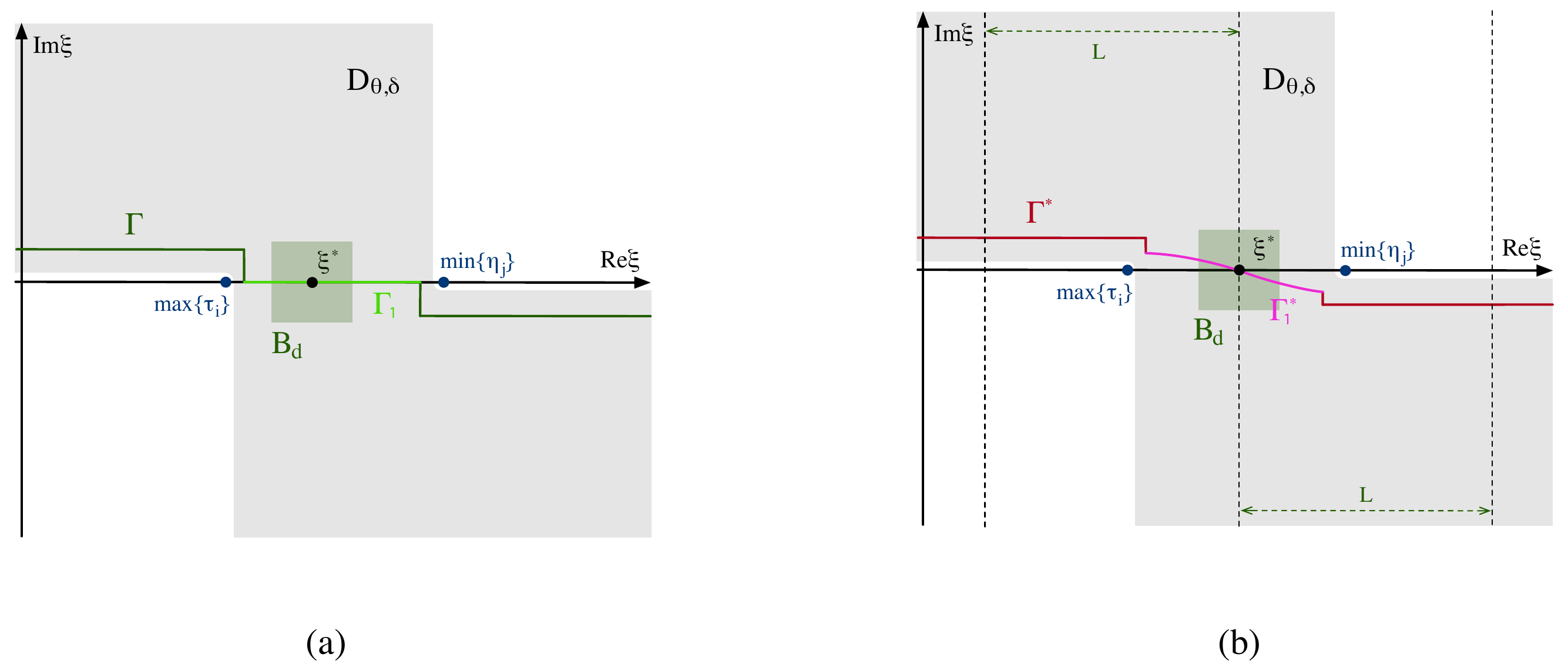}
\caption{The contour $\Gamma$ in (a) and the  contour $\Gamma^*$ in (b) in the case that $\theta=\frac{\pi}{4}.$ In (b), $\Gamma^*_d$ is the piece of $\Gamma^*$ lying in the light green region $B_d,$ and  $\Gamma^*_L$ is the piece of $\Gamma^*$ lying in between the two vertical dotted lines.}
\label{G2}
\end{figure}
Then  by (\ref{ImW-ReU}), (\ref{RI}), Proposition \ref{ri} (1) and Proposition \ref{limder1},  
we have 
$$\mathrm{Im}W_{\boldsymbol\alpha}(\xi)\leqslant \mathrm{Im}W_{\boldsymbol\alpha}(\xi^*)$$
for all $\xi\in \Gamma,$ and the equality holds if and only if $\xi\in \Gamma_1.$

Now, as depicted in Figure \ref{G2} (b), we let 
$\Gamma_1^*$ be the arc obtained from $\Gamma_1$ by following the flow of the vector field $\mathbf u=\big(0,\frac{\partial\mathrm{Re}U_{\boldsymbol\alpha}}{\partial\mathrm{Im}\xi}\big)$ for a sufficiently small time, and let $\Gamma^*$ be the contour obtained from $\Gamma$ by replacing $\Gamma_1$ and a piece of $\Gamma_2$ by $\Gamma_1^*.$ Then by Proposition \ref{ri} (2), we have 
\begin{equation}\label{absmax}
    \mathrm{Im}W_{\boldsymbol\alpha}(\xi)< \mathrm{Im}W_{\boldsymbol\alpha}(\xi^*)
\end{equation}
for all $\xi\in \Gamma^*\setminus\{\xi^*\}.$ We also see that $\Gamma^*$ lies entirely in $D_{\theta,\delta};$ and by (\ref{ImW-ReU1}) and  (\ref{fk3}),
\begin{equation*}
\bigg\{\begin{matrix} a_1 & a_2 & a_3 \\ a_4 & a_5 & a_6 \end{matrix} \bigg\}_b=\frac{1}{\pi b}\int_{\Gamma^*}\exp\bigg(\frac{W_{\boldsymbol \alpha}(\xi)+\kappa_{\boldsymbol\alpha}(\xi)|b|^2+\mathbf i\nu_{\boldsymbol\alpha,b}(\xi)|b|^4}{2\pi \mathbf i|b|^2} \bigg)d\xi.
\end{equation*}

We claim that  there is an $L_1>\xi^*-\max\{\tau_i\}-\alpha$ such that for each $t>L_1,$ 
 \begin{equation}\label{conv+}
     \mathrm{Im}W_{\boldsymbol\alpha}(\xi^*-t+\mathbf i l)\leqslant \mathrm{Im}W_{\boldsymbol\alpha}(\xi^*-L_1+\mathbf i l)-4\pi(t-L_1);
 \end{equation}
 and there is an $L_2>\min\{\eta_j\}-\alpha-\xi^*$ such that for each $t>L_2,$ 
  \begin{equation}\label{conv-}
  \mathrm{Im}W_{\boldsymbol\alpha}(\xi^*+t-\mathbf i l)\leqslant \mathrm{Im}W_{\boldsymbol\alpha}(\xi^*+L_2-\mathbf i l)-4\pi(t-L_2).
  \end{equation}
 Indeed, by (\ref{ImW-ReU}) and Proposition \ref{period}, the function 
 $f(t)\doteq \mathrm{Im}W_{\boldsymbol\alpha}(\xi^*-t+\mathbf il)+4\pi t$ is a $\pi$-periodic continuous function on $\mathbb R,$ and $L_1$ can be chosen as any of the  maximums of $f$ that is greater than $\xi^*-\max\{\tau_i\}-\alpha;$ and similarly, the function $g(t)\doteq \mathrm{Im}W_{\boldsymbol\alpha}(\xi^*+t-\mathbf il)+4\pi t$ is a $\pi$-periodic continuous function on $\mathbb R,$ and $L_2$ can be chosen as any of the  maximums of $g$ that is greater than $\min\{\eta_j\}-\alpha-\xi^*.$
 
Let
$$\Gamma^*_d=\Gamma^*\cap B_d=\Big\{ \xi \in \Gamma^*\ \Big|\ |\mathrm{Re}\xi-\xi^*|<d\ \ \text{and}\ \  |\mathrm{Im}\xi|<d\Big\},$$
and let 
$$\Gamma^*_L=\Big\{ \xi \in \Gamma^*\ \Big|\ \big|\mathrm{Re}\xi-\xi^*\big|\leqslant L\Big\},$$
where $L=\max\{L_1,L_2\}.$
See Figure \ref{G2} (b).
Then we will show that, as $b\to 0$ with $\arg b=\theta=\frac{\pi}{4},$
\begin{enumerate}[(I)]
\item 
\begin{equation*}
    \begin{split}
      \frac{1}{\pi b}\int_{\Gamma^*_d}\exp\bigg(\frac{W_{\boldsymbol \alpha}(\xi)+\kappa_{\boldsymbol\alpha}(\xi)|b|^2+\mathbf i\nu_{\boldsymbol\alpha,b}(\xi)|b|^4}{2\pi \mathbf i|b|^2}  \bigg)d\xi =  \frac{1}{2} \frac{e^{\frac{-\mathrm{Vol}(\Delta)}{\pi b^2}}}{\sqrt[4]{-\det\mathrm{Gram}(\Delta)}}\Big(1+O\big(|b|^2\big)\Big),  
    \end{split}
\end{equation*}

\item $$\bigg|\frac{1}{\pi b}\int_{\Gamma^*_L\setminus \Gamma^*_d}\exp\bigg(\frac{W_{\boldsymbol \alpha}(\xi)+\kappa_{\boldsymbol\alpha}(\xi)|b|^2+\mathbf i\nu_{\boldsymbol\alpha,b}(\xi)|b|^4}{2\pi \mathbf i|b|^2}  \bigg)d\xi\bigg|< O\Big(e^{\frac{-\epsilon}{\pi |b|^2}}\Big),$$
and 

\item $$\bigg|\frac{1}{\pi b}\int_{\Gamma^*\setminus \Gamma^*_L}\exp\bigg(\frac{W_{\boldsymbol \alpha}(\xi)+\kappa_{\boldsymbol\alpha}(\xi)|b|^2+\mathbf i\nu_{\boldsymbol\alpha,b}(\xi)|b|^4}{2\pi \mathbf i|b|^2}  \bigg)d\xi\bigg|<O\Big(e^{\frac{-\epsilon}{\pi |b|^2}}\Big)$$
\end{enumerate}
for some $\epsilon>0,$ from which the result follows.

The proof of (I), (II) and (III) follows very similarly to that of the previous case, with only minor changes.  More precisely:
\medskip

For (I), all the conditions of Proposition \ref{saddle} are satisfied by letting $\hbar=|b|^2,$ $D=B_d,$ $f=\frac{W_{\boldsymbol \alpha}}{2\pi \mathbf i},$ $g=\exp\big(\frac{\kappa_{\boldsymbol \alpha}}{2\pi \mathbf i }\big),$ $f_\hbar=\frac{W_{\boldsymbol \alpha}+\mathbf i\nu_{\boldsymbol \alpha,b}|b|^4}{2\pi \mathbf i},$ $\upsilon_h=\frac{\mathbf i\nu_{\boldsymbol \alpha,b}}{2\pi \mathbf i},$ $S=\Gamma^*_d$ and $c=\xi^*.$ By Proposition \ref{critical1} and (\ref{ImW-ReU1}), $\xi^*$ is a critical point of $f=\frac{W_{\boldsymbol \alpha}}{2\pi \mathbf i}=\frac{-\mathbf iU_{\boldsymbol \alpha}}{2\pi \mathbf i}$ in $B_{d},$  hence condition (i) is satisfied. By (\ref{absmax}), $\xi^*$ is the unique maximum point of $\mathrm{Re}f=\frac{\mathrm{Im}W_{\boldsymbol\alpha}}{2\pi}$ on $\Gamma^*_d,$ hence condition (ii) is satisfied; and by Proposition \ref{Hess}, $\xi^*$ is a non-degenerate critical point of $f,$ and condition (iii) is satisfied. For condition (iv), since $\xi^*\neq \tau_i$ and $\xi^*\neq\eta_j$ for any $i$ and $j$ in $\{1,2,3,4\},$ $\kappa_{\boldsymbol \alpha}(\xi^*)$ is a finite value. As a consequence, $g(\xi^*)=\exp\big(\frac{\kappa_{\boldsymbol \alpha}(\xi^*)}{2\pi \mathbf i }\big)\neq 0,$ and condition (iv) is satisfied. For condition (v), by Proposition \ref{bound2},  $|\upsilon_{\hbar}(\xi)|=\big|\frac{\mathbf i\nu_{\boldsymbol \alpha,b}(\xi)}{2\pi \mathbf i}\big|<\frac{N}{2\pi}$ on $B_{d}.$ 
For condition (vi), since $\Gamma_1^*$ is obtained from the straight line segment $\Gamma_1$ by following the flow of a smooth vector field, it is smooth near $\xi^*.$  Finally, by Proposition \ref{saddle}, Proposition \ref{critical1}  and Proposition \ref{Hess}, we have as $b\to 0,$
\begin{equation*}
\begin{split}
\frac{1}{\pi b}\int_{\Gamma^*_d}\exp\bigg(\frac{U_{\boldsymbol \alpha,b}(\xi)}{2\pi \mathbf ib^2}\bigg) d\xi=& \frac{(2\pi |b|^2)^\frac{1}{2}}{\pi b} \frac{\exp\big(\frac{\kappa_{\boldsymbol \alpha}(\xi^*)}{2\pi \mathbf i}\big)}{\sqrt{-\frac{\mathbf iU_{\boldsymbol \alpha}''(\xi^*)}{2\pi \mathbf i}}}e^{\frac{-\mathbf iU_{\boldsymbol\alpha}(\xi^*)}{2\pi \mathbf i |b|^2}}\Big(1+O\big(|b|^2\big)\Big)\\
=&\frac{1}{2}\frac{e^{\frac{-\mathrm{Vol}(\Delta)}{\pi b^2}}}{\sqrt[4]{-\det\mathrm{Gram}(\Delta)}}\Big(1+O\big(|b|^2\big)\Big).
\end{split}
\end{equation*}
This completes the proof of (I). 
\medskip

For (II), by (\ref{absmax}), (\ref{conv+}) and (\ref{conv-}), there is a $\epsilon_2>0$ such that  
 \begin{equation}\label{ImU2}
\mathrm{Im} W_{\boldsymbol\alpha}(\xi)<-4\epsilon_2
\end{equation}
 for all $\xi\in \Gamma^*\setminus \Gamma^*_d.$  Then by Proposition \ref{bound2},  there is a  $b_2>0$ such that 
\begin{equation}\label{last3}
\mathrm{Im}\big(\mathbf i\nu_{\boldsymbol\alpha,b}(\xi)\big)|b|^4\leqslant |\nu_{\boldsymbol\alpha,b}(\xi)||b|^4<N|b|^4<\epsilon_2 
\end{equation}
for all $|b|<b_2$ and for all $\xi\in\Gamma^*;$  and together with  (\ref{ImU2}), we have
\begin{equation}\label{last4}
\mathrm{Im}W_{\boldsymbol \alpha}(\xi)+\mathrm{Im}  \big(\mathbf i\nu_{\boldsymbol \alpha,b}(\xi)\big)|b|^4<- 3\epsilon_2
\end{equation}
for all $|b|<b_2$ and  $\xi\in\Gamma^*\setminus \Gamma^*_d.$ By the compactness of ${\Gamma^*_L}\setminus \Gamma^*_d,$ there exists an $M_2>0$ such that 
\begin{equation}\label{Imk2}
\mathrm{Im}\kappa_{\boldsymbol\alpha}(\xi)<M_2
\end{equation}
for all $\xi\in\Gamma^*_L\setminus \Gamma^*_d.$ As a consequence of (\ref{last4}) and (\ref{Imk2}), we have
\begin{equation*}
\begin{split}
& \bigg|\frac{1}{\pi b}\int_{\Gamma_L^*\setminus \Gamma^*_d}\exp\bigg(\frac{W_{\boldsymbol \alpha}(\xi)+\kappa_{\boldsymbol\alpha}(\xi)|b|^2+\mathbf i\nu_{\boldsymbol\alpha,b}(\xi)|b|^4}{2\pi \mathbf i|b|^2}   \bigg)d\xi\bigg|\\
<&\frac{1}{\pi |b|}\int_{\Gamma^*_L\setminus \Gamma^*_d}\exp\bigg(\frac{\mathrm{Im}W_{\boldsymbol \alpha}(\xi)+\mathrm{Im}\kappa_{\boldsymbol \alpha}(\xi)|b|
^2+\mathrm{Im}\big(\mathbf i\nu_{\boldsymbol \alpha,b}(\xi)\big)|b|^4}{2\pi |b|^2} \bigg)|d\xi|\\
<  & \frac{|\Gamma^*_L\setminus\Gamma^*_d| e^{\frac{M_2}{2\pi}} }{\pi |b|} \exp\bigg(\frac{-\epsilon_2}{\pi |b|^2}   \bigg )\\
< & O\Big(e^{\frac{-\epsilon}{\pi |b|^2}}\Big)
\end{split}
\end{equation*}
for any $\epsilon<\epsilon_2,$ where $|\Gamma^*_L\setminus\Gamma^*_d|$ is the length of $\Gamma^*_L\setminus\Gamma^*_d.$ This completes the proof of (II). 
\medskip

For (III), let $\epsilon_2$ and  $b_2$ be as in the proof of (II) above. Then there  is a $b_0\in (0, b_2)$ such that for all $|b|<b_0,$
\begin{equation}\label{ImK2}
\mathrm{Im}\kappa_{\boldsymbol\alpha}(\xi^*- L + \mathbf il) |b|^2<\epsilon_2\quad\text{and}\quad \mathrm{Im}\kappa_{\boldsymbol\alpha}(\xi^* + L -\mathbf il ) |b|^2<\epsilon_2,
\end{equation}
$K|b|^2<\epsilon_2$ and $N|b|^4<\epsilon_2,$ where $K$ and $N$ are respectively the constants in Propositions \ref{bound1} and \ref{bound2}. We claim that, for $\xi\in\Gamma^*\setminus \Gamma^*_L$ and $|b|<b_0,$ 
\begin{equation}\label{cl2}
\mathrm{Im}W_{\boldsymbol \alpha}(\xi)+\mathrm{Im}\kappa_{\boldsymbol \alpha}(\xi)|b|^2+\mathrm{Im}\big(\mathbf i\nu_{\boldsymbol \alpha,b}(\xi)\big)|b|^4<-2\epsilon_2-(4\pi-\epsilon_2)\big(|\mathrm{Re}\xi-\xi^*|-L\big),
\end{equation}
as a consequence of which we have
\begin{equation}\label{CI2}
\begin{split}
& \bigg|\frac{1}{\pi b}\int_{\Gamma^*\setminus \Gamma^*_L}\exp\bigg(\frac{W_{\boldsymbol \alpha}(\xi)+\kappa_{\boldsymbol\alpha}(\xi)|b|^2+\mathbf i\nu_{\boldsymbol\alpha,b}(\xi)|b|^4}{2\pi \mathbf i|b|^2}   \bigg)d\xi\bigg|\\
<& \frac{1}{\pi |b|}\int_{\Gamma^*\setminus \Gamma^*_L}\exp\bigg(\frac{\mathrm{Im}W_{\boldsymbol \alpha}(\xi)+\mathrm{Im}\kappa_{\boldsymbol \alpha}(\xi)|b|^2+\mathrm{Im}\big(\mathbf i\nu_{\boldsymbol \alpha,b}(\xi)\big)|b|^4}{2\pi |b|^2} \bigg)|d\xi| \\
 < & \frac{1}{\pi |b|} \exp\bigg(\frac{-\epsilon_2}{\pi |b|^2}\bigg)\int_{\Gamma^*\setminus \Gamma^*_L}\exp\bigg(\frac{-(4\pi-\epsilon_2)\big(|\mathrm{Re}\xi-\xi^*|-L\big)}{2\pi}\bigg) |d\xi|\\
< & O\Big(e^{\frac{-\epsilon}{\pi |b|^2}}\Big)
\end{split}
\end{equation}
for any $\epsilon<\epsilon_2.$ For the proof of the claim, by (\ref{conv+}) and (\ref{conv-}), we have for each $\xi\in\Gamma^*\setminus \Gamma^*_L,$ either 
$$\mathrm{Im}W_{\boldsymbol\alpha}(\xi)\leqslant \mathrm{Im}W_{\boldsymbol\alpha}(\xi^*-L_1+\mathbf i l)-4\pi\big(|\mathrm{Re}\xi-\xi^*|-L_1 \big)$$
or 
$$\mathrm{Im}W_{\boldsymbol\alpha}(\xi)\leqslant \mathrm{Im}W_{\boldsymbol\alpha}(\xi^*+L_2-\mathbf i l)-4\pi\big(|\mathrm{Re}\xi-\xi^*|-L_2 \big).$$
In either case, by  (\ref{ImU2}) and that $L=\max\{L_1,L_2\},$
we have 
\begin{equation}\label{446}
    \mathrm{Im}W_{\boldsymbol\alpha}(\xi)< -4\epsilon_2-4\pi\big(|\mathrm{Re}\xi-\xi^*|-L \big).
\end{equation}
By Proposition \ref{bound1} and the choice of $b_0,$ for $t>L$ and $|b|<b_0,$ we have 
$$\frac{d}{dt} \Big( \mathrm{Im}\kappa_{\boldsymbol\alpha}\big(\xi^*\pm ( t - \mathbf il)\big)|b|^2\Big)<K|b|^2<\epsilon_2.$$
Together with the Mean Value Theorem and (\ref{ImK2}), we have 
\begin{equation}\label{Bou2}
\begin{split}
\mathrm{Im}\kappa_{\boldsymbol\alpha}(\xi)|b|^2 < &  \mathrm{Im}
\kappa_{\boldsymbol\alpha}\big(\xi^*\pm (L - \mathbf il)\big)|b|^2 +\epsilon_2 \big |  \mathrm{Re}\xi - (\xi^*\pm L ) \big|\\
<& \epsilon_2 +\epsilon_2\big(|\mathrm{Re}\xi-\xi^*|-L \big)
\end{split}
\end{equation}
for all $\xi \in \Gamma^*\setminus \Gamma^*_L.$  Finally, putting (\ref{446}), (\ref{Bou2}) and (\ref{last3}) together, we have  (\ref{cl2}) and the second  inequality in (\ref{CI2}); and since 
$$|\mathrm{Re}\xi-\xi^*|-L\to+\infty$$
as $\xi \in \Gamma^*\setminus \Gamma^*_L$ approaches $+\infty,$ we have the last inequality in (\ref{CI2}). This completes the proof of (III).
\smallskip

Putting (I), (II)  and (III) together, and as $\mathrm{Re}\Big(\frac{-\mathrm{Vol}(\Delta)}{\pi b^2}\Big)=0,$ we have as $b\to 0$ with $\arg b =\frac{\pi}{4},$ 
$$\bigg\{\begin{matrix} a_1 & a_2 & a_3 \\ a_4 & a_5 & a_6 \end{matrix} \bigg\}_b=\frac{1}{2}\frac{e^{\frac{-\mathrm{Vol}(\Delta)}{\pi b^2}}}{\sqrt[4]{-\det\mathrm{Gram}(\Delta)}}\Big(1+O\big(b^2\big)\Big). $$ 
\bigskip

For (2) that $\theta\in (\frac{\pi}{4},\frac{\pi}{2})$ and $(\theta_1,\dots,\theta_6)$ satisfies the Geometric Hypothesis \ref{per-min}, we let $d>0$ be  sufficiently small so that the region 
 $$B_{d}=\Big\{\xi\in \mathbb C \ \Big|\ |\mathrm{Re}\xi -\xi^*| \leqslant d\ \ 
 \text{and}\ \  |\mathrm{Im}\xi|\leqslant d \Big\}$$
 lies entirely in $D_{\theta},$ where $\xi^*$ is the critical point of $U_{\boldsymbol\alpha}$ given by Proposition \ref{critical1}.  Let $\alpha>0$ be sufficiently small so that 
 $$\max\{\tau_i\}+\alpha<\xi^*-d\quad\text{and}\quad \xi^*+d< \min\{\eta_j\}-\alpha,$$
 and let 
 $$\Gamma_1=\Big[\max\{\tau_i\}+\alpha, \min\{\eta_j\}-\alpha\Big].$$
By Proposition \ref{ri} (2), $\mathrm{Im}U_{\boldsymbol\alpha}(\xi^*)<\mathrm{Im}U_{\boldsymbol\alpha}(\max\{\tau_i\}+\alpha)$ and $\mathrm{Im}U_{\boldsymbol\alpha}(\xi^*)<\mathrm{Im}U_{\boldsymbol\alpha}(\min\{\eta_j\}-\alpha).$  By Proposition \ref{period} and the Geometric Hypothesis \ref{per-min}, for $c>0$ sufficiently small 
 there is an $\beta>0$  such that for each $\xi\in \mathbb R\setminus \cup_{k\in \mathbb Z} [\xi^*-\beta+k\pi,\xi^*+\beta+k\pi],$ 
 \begin{equation}\label{2c}
 \mathrm{Im}U_{\boldsymbol\alpha}(\xi) > \mathrm{Im}U_{\boldsymbol\alpha}(\xi^*) + c.
 \end{equation}
We can choose $c$ to be so small that 
 \begin{equation}\label{cc}
     c < \min\bigg\{\frac{4\pi^2\sin(2\theta)}{
     \sin(2\theta)-\cos(2\theta)}
     , \mathrm{Im}U_{\boldsymbol\alpha}(\max\{\tau_i\}+\alpha)-\mathrm{Im}U_{\boldsymbol\alpha}(\xi^*),\mathrm{Im}U_{\boldsymbol\alpha}(\min\{\eta_j\}-\alpha)-\mathrm{Im}U_{\boldsymbol\alpha}(\xi^*)\bigg\}
 \end{equation}
 and $\beta<d.$
Then by Proposition \ref{period} and Proposition \ref{ri}, there is an $l>0$ such that:
\begin{enumerate}[(i)]
    \item For all $t\in \big(-\infty,\max\{\tau_i\}+\alpha \big],$
\begin{equation}\label{c1}
\mathrm{Im}U_{\boldsymbol\alpha}(t+\mathbf il) > \mathrm{Im}U_{\boldsymbol\alpha}(t) - c;
\end{equation}
and for all $t\in \big[\min\{\eta_j\}-\alpha,\infty\big),$
\begin{equation}\label{c2}
    \mathrm{Im}U_{\boldsymbol\alpha}(t-\mathbf il) > \mathrm{Im}U_{\boldsymbol\alpha}(t) - c.
    \end{equation}

    \item For all $t\in[0,l],$
    \begin{equation}\label{c3}
    \mathrm{Im}U_{\boldsymbol\alpha}(\max\{\tau_i\}+\alpha+\mathbf it) > \mathrm{Im}U_{\boldsymbol\alpha}(\max\{\tau_i\}+\alpha) - c,
\end{equation}
and 
\begin{equation}\label{c4}
\mathrm{Im}U_{\boldsymbol\alpha}(\min\{\eta_j\}-\alpha-\mathbf i t) > \mathrm{Im}U_{\boldsymbol\alpha}(\min\{\eta_j\}-\alpha) - c.
\end{equation}

    \item  For all $t\in[0,l],$
    \begin{equation}\label{RI3}
     -\frac{\partial \mathrm{Re}U_{\boldsymbol\alpha}}{\partial \mathrm{Im}\xi}\bigg|_{\xi=\max\{\tau_i\}+\alpha+\mathbf i t}<0\quad\text{and}\quad -\frac{\partial \mathrm{Re}U_{\boldsymbol\alpha}}{\partial \mathrm{Im}\xi}\bigg|_{\xi=\min\{\eta_j\}-\alpha-\mathbf i t}>0.
 \end{equation}

   \item For all $t\in [\xi^*-\beta,\xi^*+\beta],$
   \begin{equation}\label{c5}
\mathrm{Re}U_{\boldsymbol\alpha}(t\pm \mathbf i l) > \mathrm{Re}U_{\boldsymbol\alpha}(t) -c= -c.
\end{equation}

\end{enumerate}
Let 
 $$\Gamma_2=\Big\{ \max\{\tau_i\}+\alpha+\mathbf i t\ \Big|\ t\in [0,l]\Big\}\cup \Big\{ \min\{\eta_j\}-\alpha-\mathbf i t\ \Big|\ t\in [0,l]\Big\}$$
 and 
 $$\Gamma_3=\Big\{t+\mathbf il \ \Big|\ t\in \big(-\infty,\max\{\tau_i\}+\alpha \big]\Big\}\cup \Big\{ t-\mathbf il \ \Big|\ t\in\big[\min\{\eta_j\}-\alpha,\infty\big) \Big\}.$$
As depicted in Figure \ref{G3}, let $$\Gamma^*=\Gamma_1\cup \Gamma_2\cup \Gamma_3,$$
and let $\delta>0$ be sufficiently small so that $\Gamma^*$ lies entirely in the region $D_{\theta,\delta}.$
\begin{figure}[htbp]
\centering
\includegraphics[scale=0.4]{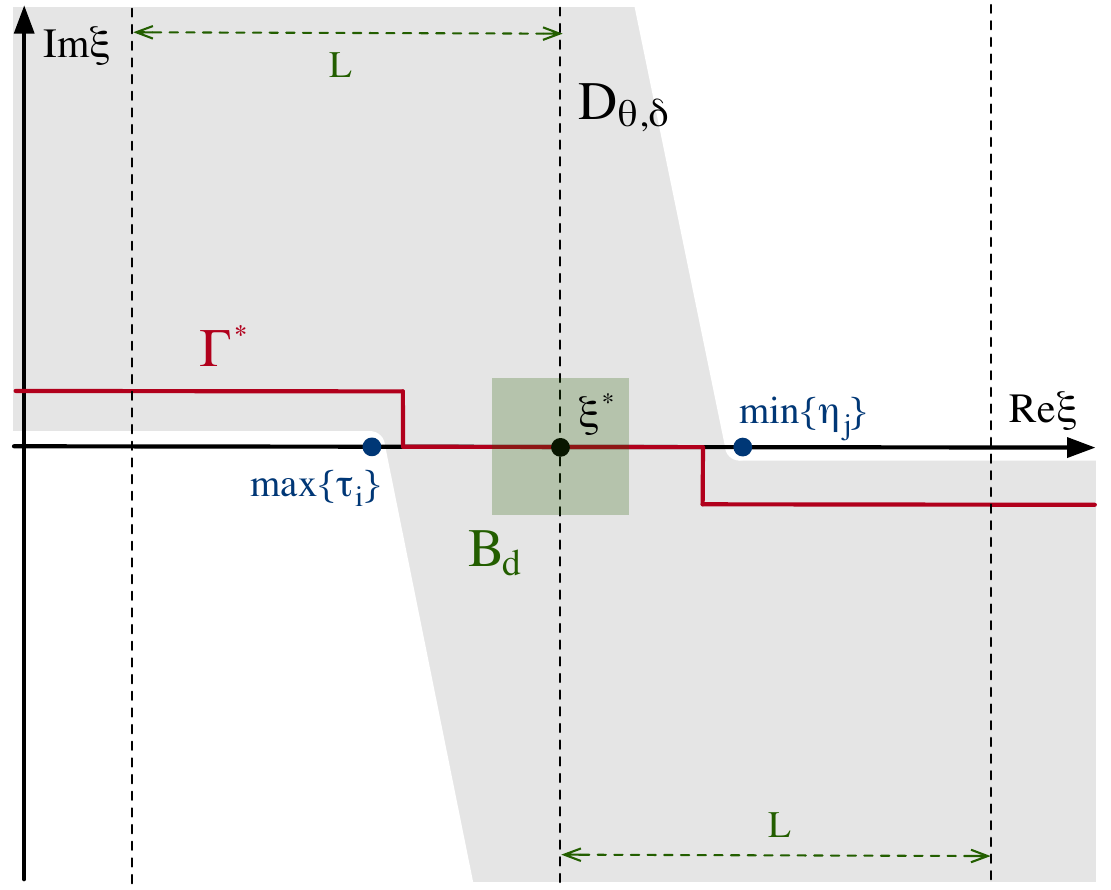}
\caption{The contour $\Gamma^*$ in the case that $\theta\in(\frac{\pi}{4},\frac{\pi}{2}),$ where $\Gamma^*_d$ is the piece of $\Gamma^*$ lying in the light green region $B_d,$ and  $\Gamma^*_L$ is the piece of $\Gamma^*$ lying in between the two vertical dotted lines.}
\label{G3}
\end{figure}

We claim that
\begin{equation}\label{absmax3}
    \mathrm{Im}W_{\boldsymbol\alpha}(\xi)< \mathrm{Im}W_{\boldsymbol\alpha}(\xi^*)
\end{equation}
for all $\xi\in \Gamma^*\setminus\{\xi^*\}.$ 

Indeed, for $\xi\in\Gamma_1\setminus \xi^*,$ by Proposition \ref{ri}, $\mathrm{Im}U_{\boldsymbol\alpha}(\xi)>\mathrm{Im}U_{\boldsymbol\alpha}(\xi^*)$ and $\mathrm{Re}U_{\boldsymbol\alpha}(\xi)=\mathrm{Re}U_{\boldsymbol\alpha}(\xi^*).$ Then by (\ref{ImW}) and that $\cos(2\theta)<0,$ (\ref{absmax3}) follows.

For $\xi\in \Gamma_2,$ by (\ref{cc}), (\ref{c3}) and (\ref{c4}), we have 
$$\mathrm{Im}U_{\boldsymbol\alpha}(\max\{\tau_i\}+\alpha+\mathbf it) > \mathrm{Im}U_{\boldsymbol\alpha}(\max\{\tau_i\}+\alpha)-c> \mathrm{Im}U_{\boldsymbol\alpha}(\xi^*)$$ 
and 
$$\mathrm{Im}U_{\boldsymbol\alpha}(\min\{\eta_j\}-\alpha-\mathbf it)>\mathrm{Im}U_{\boldsymbol\alpha}(\min\{\eta_j\}-\alpha)-c > \mathrm{Im}U_{\boldsymbol\alpha}(\xi^*)$$
for all $t\in [0,l].$ By (\ref{RI3}) and Proposition \ref{ri} (1), we have 
$$\mathrm{Re}U_{\boldsymbol\alpha}(\max\{\tau_i\}+\alpha+\mathbf it) > \mathrm{Re}U_{\boldsymbol\alpha}(\max\{\tau_i\}+\alpha)= \mathrm{Re}U_{\boldsymbol\alpha}(\xi^*)$$
and 
$$\mathrm{Re}U_{\boldsymbol\alpha}(\min\{\eta_j\}-\alpha-\mathbf it)>\mathrm{Re}U_{\boldsymbol\alpha}(\min\{\eta_j\}-\alpha)= \mathrm{Re}U_{\boldsymbol\alpha}(\xi^*)$$
for all $t\in [0,l].$ Then by (\ref{ImW}) and that $\cos(2\theta)<0$ and $\sin(2\theta)>0,$ (\ref{absmax3}) follows.  

For $\xi=t\pm \mathbf il\in \Gamma_3,$ we consider the following three cases.
\begin{enumerate}[(a)]
\item  If $t\notin [\xi^*-\beta+k\pi, \xi^*+\beta+k\pi]$ for any $k\in\mathbb Z,$ then by (\ref{2c}), (\ref{c1}) and (\ref{c2}), we have 
$$\mathrm{Im}U_{\boldsymbol\alpha}(t + \mathbf il) > \mathrm{Im}U_{\boldsymbol\alpha}(t) - c>\mathrm{Im}U_{\boldsymbol\alpha}(\xi^*)$$
for all $t\in \big(-\infty,\max\{\tau_i\}+\alpha \big],$
and
$$\mathrm{Im}U_{\boldsymbol\alpha}(t- \mathbf il) > \mathrm{Im}U_{\boldsymbol\alpha}(t) - c>\mathrm{Im}U_{\boldsymbol\alpha}(\xi^*)$$
for all $t\in\big[\min\{\eta_j\}-\alpha,\infty\big);$
and by Proposition \ref{limder1} and the previous case,
$$\mathrm{Re}U_{\boldsymbol\alpha}(t + \mathbf il) > \mathrm{Re}U_{\boldsymbol\alpha}(\max\{\tau_i\}+\alpha+\mathbf il) >\mathrm{Re}U_{\boldsymbol\alpha}(\xi^*)$$
for all $t\in \big(-\infty,\max\{\tau_i\}+\alpha \big],$
and 
$$\mathrm{Re}U_{\boldsymbol\alpha}(t - \mathbf il) > \mathrm{Re}U_{\boldsymbol\alpha}(\min\{\eta_j\}-\alpha-\mathbf il) >\mathrm{Re}U_{\boldsymbol\alpha}(\xi^*)$$
for all $t\in\big[\min\{\eta_j\}-\alpha,\infty\big).$
Then (\ref{absmax3}) follows by (\ref{ImW}) and that $\cos(2\theta)<0$ and $\sin(2\theta)>0$  in this case. 

\item If $t\in [\xi^*-\beta+k\pi, \xi^*+\beta+k\pi]$ for some $k<0,$ then by (\ref{c1}) and the Geometric Hypothesis \ref{per-min}, we have 
$$\mathrm{Im}U_{\boldsymbol\alpha}(t + \mathbf il) > \mathrm{Im}U_{\boldsymbol\alpha}(t) - c\geqslant\mathrm{Im}U_{\boldsymbol\alpha}(\xi^*)-c;$$
and by Proposition \ref{period} and (\ref{c5}), we have $$\mathrm{Re}U_{\boldsymbol\alpha}(t + \mathbf il) =  \mathrm{Re}U_{\boldsymbol\alpha}(t -k\pi+ \mathbf il) - 4k\pi^2> -c+4\pi^2.$$
Together with (\ref{ImW}) and (\ref{cc}), we have 
$$\mathrm{Im}W_{\boldsymbol\alpha}(t+\mathbf il)<\mathrm{Im}W_{\boldsymbol\alpha}(\xi^*)+\big(\sin(2\theta)-\cos(2\theta)\big)c-4\pi^2\sin(2\theta)<\mathrm{Im}W_{\boldsymbol\alpha}(\xi^*),$$
and (\ref{absmax3}) holds in this case.

\item If $t\in [\xi^*-\beta+k\pi, \xi^*+\beta+k\pi]$ for some $k>0,$ then by (\ref{c2}) and the Geometric Hypothesis \ref{per-min}, we have 
$$\mathrm{Im}U_{\boldsymbol\alpha}(t - \mathbf il) > \mathrm{Im}U_{\boldsymbol\alpha}(t) - c\geqslant \mathrm{Im}U_{\boldsymbol\alpha}(\xi^*)-c;$$
and by Proposition \ref{period} and (\ref{c5}), we have $$\mathrm{Re}U_{\boldsymbol\alpha}(t - \mathbf il) =  \mathrm{Re}U_{\boldsymbol\alpha}(t -k\pi- \mathbf il) + 4k\pi^2> -c+4\pi^2.$$
Together with (\ref{ImW}) and (\ref{cc}), we have 
$$\mathrm{Im}W_{\boldsymbol\alpha}(t-\mathbf il)<\mathrm{Im}W_{\boldsymbol\alpha}(\xi^*)+\big(\sin(2\theta)-\cos(2\theta)\big)c-4\pi^2\sin(2\theta)<\mathrm{Im}W_{\boldsymbol\alpha}(\xi^*),$$
and (\ref{absmax3}) holds in this case. 
\end{enumerate}

Now by (\ref{W}) and  (\ref{fk3}), we have
\begin{equation*}
\bigg\{\begin{matrix} a_1 & a_2 & a_3 \\ a_4 & a_5 & a_6 \end{matrix} \bigg\}_b=\frac{1}{\pi b}\int_{\Gamma^*}\exp\bigg(\frac{W_{\boldsymbol \alpha}(\xi)+\kappa_{\boldsymbol\alpha}(\xi)|b|^2+e^{2\mathbf i\theta}\nu_{\boldsymbol\alpha,b}(\xi)|b|^4}{2\pi \mathbf i|b|^2} \bigg)d\xi.
\end{equation*}

We claim that  there is an $L_1>\xi^*-\max\{\tau_i\}-\alpha$ such that for $t>L_1,$ 
 \begin{equation}\label{conv+3}
     \mathrm{Im}W_{\boldsymbol\alpha}(\xi^*-t+\mathbf i l)\leqslant \mathrm{Im}W_{\boldsymbol\alpha}(\xi^*-L_1+\mathbf i l)-4\pi\sin(2\theta)(t-L_1);
 \end{equation}
 and there is an $L_2>\min\{\eta_j\}-\alpha-\xi^*$ such that for  $t>L_2,$ 
  \begin{equation}\label{conv-3}
  \mathrm{Im}W_{\boldsymbol\alpha}(\xi^*+t-\mathbf i l)\leqslant \mathrm{Im}W_{\boldsymbol\alpha}(\xi^*+L_2-\mathbf i l)-4\pi\sin(2\theta)(t-L_2).
  \end{equation}
 Indeed, by (\ref{ImW}) and Proposition \ref{period}, the function 
 $f(t)\doteq \mathrm{Im}W_{\boldsymbol\alpha}(\xi^*-t+\mathbf il)+4\pi \sin(2\theta)t$ is a $\pi$-periodic continuous function on $\mathbb R,$ and $L_1$ can be chosen as any of the  maximums of $f$ that is greater than $\xi^*-\max\{\tau_i\}-\alpha;$ and similarly, the function $g(t)\doteq \mathrm{Im}W_{\boldsymbol\alpha}(\xi^*+t-\mathbf il)+4\pi \sin(2\theta)t$ is a $\pi$-periodic continuous function on $\mathbb R,$ and $L_2$ can be chosen as any of the  maximums of $g$ that is greater than $\min\{\eta_j\}-\alpha-\xi^*.$
 
Let
$$\Gamma^*_d=\Gamma^*\cap B_d=\Big\{ \xi \in \Gamma^*\ \Big|\ |\mathrm{Re}\xi-\xi^*|<d\ \ \text{and}\ \  |\mathrm{Im}\xi|<d\Big\},$$
and let 
$$\Gamma^*_L=\Big\{ \xi \in \Gamma^*\ \Big|\ \big|\mathrm{Re}\xi-\xi^*\big|\leqslant L\Big\},$$
where $L=\max\{L_1,L_2\}.$
See Figure \ref{G3}. Then we will show that, as $b\to 0$ with $\arg b=\theta\in(\frac{\pi}{4},\frac{\pi}{2}),$
\begin{enumerate}[(I)]
\item 
\begin{equation*}
    \begin{split}
      \frac{1}{\pi b}\int_{\Gamma^*_d}\exp\bigg(\frac{W_{\boldsymbol \alpha}(\xi)+\kappa_{\boldsymbol\alpha}(\xi)|b|^2+e^{2\mathbf i\theta}\nu_{\boldsymbol\alpha,b}(\xi)|b|^4}{2\pi \mathbf i|b|^2}  \bigg)d\xi =   \frac{1}{2}\frac{e^{\frac{-\mathrm{Vol}(\Delta)}{\pi b^2}}}{\sqrt[4]{-\det\mathrm{Gram}(\Delta)}}\Big(1+O\big(|b|^2\big)\Big),  
    \end{split}
\end{equation*}

\item $$\bigg|\frac{1}{\pi b}\int_{\Gamma^*_L\setminus \Gamma^*_d}\exp\bigg(\frac{W_{\boldsymbol \alpha}(\xi)+\kappa_{\boldsymbol\alpha}(\xi)|b|^2+e^{2\mathbf i\theta}\nu_{\boldsymbol\alpha,b}(\xi)|b|^4}{2\pi \mathbf i|b|^2}  \bigg)d\xi\bigg|< O\Big(e^{\frac{-\cos(2\theta)\mathrm{Vol}(\Delta)-\epsilon}{\pi |b|^2}}\Big),$$
and 

\item $$\bigg|\frac{1}{\pi b}\int_{\Gamma^*\setminus \Gamma^*_L}\exp\bigg(\frac{W_{\boldsymbol \alpha}(\xi)+\kappa_{\boldsymbol\alpha}(\xi)|b|^2+e^{2\mathbf i\theta}\nu_{\boldsymbol\alpha,b}(\xi)|b|^4}{2\pi \mathbf i|b|^2}  \bigg)d\xi\bigg|<O\Big(e^{\frac{-\cos(2\theta)\mathrm{Vol}(\Delta)-\epsilon}{\pi |b|^2}}\Big)$$
\end{enumerate}
for some $\epsilon>0,$ from which the result follows.

The proof of (I), (II) and (III) follows very similarly to that of the previous two cases, with only minor changes.  More precisely: 
\medskip

For (I), we claim that  all the conditions of Proposition \ref{saddle} are satisfied by letting $\hbar=|b|^2,$ $D=B_d,$ $f=\frac{W_{\boldsymbol \alpha}}{2\pi \mathbf i},$ $g=\exp\big(\frac{\kappa_{\boldsymbol \alpha}}{2\pi \mathbf i }\big),$ $f_\hbar=\frac{W_{\boldsymbol \alpha}+e^{2\mathbf i\theta}\nu_{\boldsymbol \alpha,b}|b|^4}{2\pi \mathbf i},$ $\upsilon_h=\frac{e^{2\mathbf i\theta}\nu_{\boldsymbol \alpha,b}}{2\pi \mathbf i},$ $S=\Gamma^*_d$ and $c=\xi^*.$ Indeed, by Proposition \ref{critical1}, $\xi^*$ is a critical point of $f=\frac{W_{\boldsymbol \alpha}}{2\pi \mathbf i}=\frac{e^{-2\mathbf i \theta}U_{\boldsymbol \alpha}}{2\pi \mathbf i}$ in $B_{d},$  hence condition (i) is satisfied.
By (\ref{absmax3}), $\xi^*$ is the unique maximum point of $\mathrm{Re}f=\frac{\mathrm{Im}W_{\boldsymbol\alpha}}{2\pi}$ on $\Gamma^*_d,$ hence condition (ii) is satisfied; and by Proposition \ref{Hess}, $\xi^*$ is a non-degenerate critical point of $f,$ and condition (iii) is satisfied. For condition (iv), since $\xi^*\neq \tau_i$ and $\xi^*\neq\eta_j$ for any $i$ and $j$ in $\{1,2,3,4\},$ $\kappa_{\boldsymbol \alpha}(\xi^*)$ is a finite value. As a consequence, $g(\xi^*)=\exp\big(\frac{\kappa_{\boldsymbol \alpha}(\xi^*)}{2\pi \mathbf i }\big)\neq 0,$ and condition (iv) is satisfied.
For condition (v), by Proposition \ref{bound2},  $|\upsilon_{\hbar}(\xi)|=\big|\frac{e^{2\mathbf i\theta}\nu_{\boldsymbol \alpha,b}(\xi)}{2\pi \mathbf i}\big|<\frac{N}{2\pi}$ on $B_{d}.$ For condition (vi), since $\Gamma^*$ is a straight line, it is smooth near $\xi^*.$  Finally, by Proposition \ref{saddle}, Proposition \ref{critical1}  and Proposition \ref{Hess}, we have as $b\to 0,$
\begin{equation*}
\begin{split}
\frac{1}{\pi b}\int_{\Gamma^*_d}\exp\bigg(\frac{U_{\boldsymbol \alpha,b}(\xi)}{2\pi \mathbf ib^2}\bigg) d\xi=& \frac{(2\pi |b|^2)^\frac{1}{2}}{\pi b} \frac{\exp\big(\frac{\kappa_{\boldsymbol \alpha}(\xi^*)}{2\pi \mathbf i}\big)}{\sqrt{-\frac{e^{-2\mathbf i\theta}U_{\boldsymbol \alpha}''(\xi^*)}{2\pi \mathbf i}}}e^{\frac{e^{-2\mathbf i\theta}U_{\boldsymbol\alpha}(\xi^*)}{2\pi \mathbf i |b|^2}}\Big(1+O\big(|b|^2\big)\Big)\\
=&\frac{1}{2}\frac{e^{\frac{-\mathrm{Vol}(\Delta)}{\pi b^2}}}{\sqrt[4]{-\det\mathrm{Gram}(\Delta)}}\Big(1+O\big(|b|^2\big)\Big).
\end{split}
\end{equation*}
This completes the proof of (I). 
\medskip

For (II), by (\ref{absmax3}), (\ref{conv+3}) and (\ref{conv-3}), there is a $\epsilon_3>0$ sufficiently small so that
\begin{equation}\label{E3}
    \epsilon_3<4\pi\sin(2\theta)
\end{equation}
and
 \begin{equation}\label{ImU3}
\mathrm{Im} W_{\boldsymbol\alpha}(\xi)<-2\cos(2\theta)\mathrm{Vol}(\Delta)-4\epsilon_3
\end{equation}
 for all $\xi\in \Gamma^*\setminus \Gamma^*_d.$ Then by Proposition \ref{bound2},  there is a  $b_3>0$ such that 
\begin{equation}\label{last6}
\mathrm{Im}\big(e^{2\mathbf i\theta}\nu_{\boldsymbol\alpha,b}(\xi)\big)|b|^4\leqslant |\nu_{\boldsymbol\alpha,b}(\xi)||b|^4<N|b|^4<\epsilon_3 
\end{equation}
for all $|b|<b_3$ and for all $\xi\in\Gamma^*;$  and together with  (\ref{ImU3}), we have
\begin{equation}\label{last7}
\mathrm{Im}W_{\boldsymbol \alpha}(\xi)+\mathrm{Im}  \big(e^{2\mathbf i\theta}\nu_{\boldsymbol \alpha,b}(\xi)\big)|b|^4<-2\cos(2\theta)\mathrm{Vol}(\Delta)- 3\epsilon_3
\end{equation}
for all $|b|<b_3$ and  $\xi\in\Gamma^*\setminus \Gamma^*_d.$ By the compactness of ${\Gamma^*_L}\setminus \Gamma^*_d,$ there exists an $M_3>0$ such that 
\begin{equation}\label{Imk3}
\mathrm{Im}\kappa_{\boldsymbol\alpha}(\xi)<M_3
\end{equation}
for all $\xi\in\Gamma^*_L\setminus \Gamma^*_d.$ As a consequence of (\ref{last6}) and (\ref{Imk3}), we have
\begin{equation*}
\begin{split}
& \bigg|\frac{1}{\pi b}\int_{\Gamma_L^*\setminus \Gamma^*_d}\exp\bigg(\frac{W_{\boldsymbol \alpha}(\xi)+\kappa_{\boldsymbol\alpha}(\xi)|b|^2+e^{2\mathbf i\theta}\nu_{\boldsymbol\alpha,b}(\xi)|b|^4}{2\pi \mathbf i|b|^2}   \bigg)d\xi\bigg|\\
<&\frac{1}{\pi |b|}\int_{\Gamma^*_L\setminus \Gamma^*_d}\exp\bigg(\frac{\mathrm{Im}W_{\boldsymbol \alpha}(\xi)+\mathrm{Im}\kappa_{\boldsymbol \alpha}(\xi)|b|
^2+\mathrm{Im}\big(e^{2\mathbf i\theta}\nu_{\boldsymbol \alpha,b}(\xi)\big)|b|^4}{2\pi |b|^2} \bigg)|d\xi|\\
<  & \frac{2(L+l-d) e^{\frac{M_3}{2\pi}} }{\pi |b|} \exp\bigg(\frac{-\cos(2\theta)\mathrm{Vol}(\Delta)-\epsilon_3}{\pi |b|^2}   \bigg )\\
< & O\Big(e^{\frac{-\cos(2\theta)\mathrm{Vol}(\Delta)-\epsilon}{\pi |b|^2}}\Big)
\end{split}
\end{equation*}
for any $\epsilon<\epsilon_3.$  This completes the proof of (II). 
\medskip

For (III), let $\epsilon_3$ and  $b_3$ be as in the proof of (II) above. Then there  is a $b_0\in (0, b_3)$ such that for all $|b|<b_0,$
\begin{equation}\label{ImK2}
\mathrm{Im}\kappa_{\boldsymbol\alpha}(\xi^*- L + \mathbf il) |b|^2<\epsilon_3\quad\text{and}\quad \mathrm{Im}\kappa_{\boldsymbol\alpha}(\xi^* + L -\mathbf il ) |b|^2<\epsilon_3,
\end{equation}
$K|b|^2<\epsilon_3$ and $N|b|^4<\epsilon_3,$ where $K$ and $N$ are respectively the constants in Propositions \ref{bound1} and \ref{bound2}. We claim that, for $\xi\in\Gamma^*\setminus \Gamma^*_L$ and $|b|<b_0,$ 
\begin{equation}\label{cl3}
\begin{split}
\mathrm{Im}W_{\boldsymbol \alpha}(\xi)&+\mathrm{Im}\kappa_{\boldsymbol \alpha}(\xi)|b|^2+\mathrm{Im}\big(e^{2\mathbf i\theta}\nu_{\boldsymbol \alpha,b}(\xi)\big)|b|^4\\<&-2\cos(2\theta)\mathrm{Vol}(\Delta)-2\epsilon_3-\big(4\pi\sin(2\theta)-\epsilon_3\big)\big(|\mathrm{Re}\xi-\xi^*|-L\big),
\end{split}
\end{equation}
as a consequence of which we have
\begin{equation}\label{CI3}
\begin{split}
& \bigg|\frac{1}{\pi b}\int_{\Gamma^*\setminus \Gamma^*_L}\exp\bigg(\frac{W_{\boldsymbol \alpha}(\xi)+\kappa_{\boldsymbol\alpha}(\xi)|b|^2+e^{2\mathbf i\theta}\nu_{\boldsymbol\alpha,b}(\xi)|b|^4}{2\pi \mathbf i|b|^2}   \bigg)d\xi\bigg|\\
<& \frac{1}{\pi |b|}\int_{\Gamma^*\setminus \Gamma^*_L}\exp\bigg(\frac{\mathrm{Im}W_{\boldsymbol \alpha}(\xi)+\mathrm{Im}\kappa_{\boldsymbol \alpha}(\xi)|b|^2+\mathrm{Im}\big(e^{2\mathbf i\theta}\nu_{\boldsymbol \alpha,b}(\xi)\big)|b|^4}{2\pi |b|^2} \bigg)|d\xi| \\
 < & \frac{1}{\pi |b|} \exp\bigg(\frac{-\cos(2\theta)\mathrm{Vol}(\Delta)-\epsilon_3}{\pi |b|^2}\bigg)\int_{\Gamma^*\setminus \Gamma^*_L}\exp\bigg(\frac{-(4\pi\sin(2\theta)-\epsilon_3)\big(|\mathrm{Re}\xi-\xi^*|-L\big)}{2\pi}\bigg) |d\xi|\\
< & O\Big(e^{\frac{-\cos(2\theta)\mathrm{Vol}(\Delta)-\epsilon}{\pi |b|^2}}\Big)
\end{split}
\end{equation}
for any $\epsilon<\epsilon_3.$ For the proof of the claim, by (\ref{conv+3}) and (\ref{conv-3}), we have for each $\xi\in\Gamma^*\setminus \Gamma^*_L,$ either 
$$\mathrm{Im}W_{\boldsymbol\alpha}(\xi)\leqslant \mathrm{Im}W_{\boldsymbol\alpha}(\xi^*-L_1+\mathbf i l)-4\pi\sin(2\theta)\big(|\mathrm{Re}\xi-\xi^*|-L_1 \big)$$
or 
$$\mathrm{Im}W_{\boldsymbol\alpha}(\xi)\leqslant \mathrm{Im}W_{\boldsymbol\alpha}(\xi^*+L_2-\mathbf i l)-4\pi\sin(2\theta)\big(|\mathrm{Re}\xi-\xi^*|-L_2 \big).$$
In either case, by  (\ref{ImU3}) and that $L=\max\{L_1,L_2\},$
we have 
\begin{equation}\label{466}
    \mathrm{Im}W_{\boldsymbol\alpha}(\xi)< -2\cos(2\theta)\mathrm{Vol}(\Delta)-4\epsilon_3-4\pi\sin(2\theta)\big(|\mathrm{Re}\xi-\xi^*|-L \big).
\end{equation}
By Proposition \ref{bound1} and the choice of $b_0,$ for $t>L$ and $|b|<b_0,$ we have 
$$\frac{d}{dt} \Big( \mathrm{Im}\kappa_{\boldsymbol\alpha}\big(\xi^*\pm ( t - \mathbf il)\big)|b|^2\Big)<K|b|^2<\epsilon_3.$$
Together with the Mean Value Theorem and (\ref{ImK2}), we have 
\begin{equation}\label{Bou3}
\begin{split}
\mathrm{Im}\kappa_{\boldsymbol\alpha}(\xi)|b|^2 < &  \mathrm{Im}
\kappa_{\boldsymbol\alpha}\big(\xi^*\pm (L - \mathbf il)\big)|b|^2 +\epsilon_3 \big |  \mathrm{Re}\xi - (\xi^*\pm L ) \big|\\
<& \epsilon_3 +\epsilon_3\big(|\mathrm{Re}\xi-\xi^*|-L \big)
\end{split}
\end{equation}
for all $\xi \in \Gamma^*\setminus \Gamma^*_L.$  Finally, putting (\ref{466}), (\ref{Bou3}) and (\ref{last6}) together, we have  (\ref{cl3}) and the second  inequality in (\ref{CI3}); and since 
$$|\mathrm{Re}\xi-\xi^*|-L\to+\infty$$
as $\xi \in \Gamma^*\setminus \Gamma^*_L$ approaches $+\infty$ and $4\pi\sin(2\theta)-\epsilon_3>0$ by (\ref{E3}), we have the last inequality in (\ref{CI3}). This completes the proof of (III).
\smallskip

Putting (I), (II)  and (III) together, and as $\mathrm{Re}\Big(\frac{-\mathrm{Vol}(\Delta)}{\pi b^2}\Big)=\frac{-\cos(2\theta)\mathrm{Vol}(\Delta)}{\pi |b|^2},$ we have as $b\to 0$ with a fixed $\arg b \in (\frac{\pi}{4},\frac{\pi}{2}),$ 
$$\bigg\{\begin{matrix} a_1 & a_2 & a_3 \\ a_4 & a_5 & a_6 \end{matrix} \bigg\}_b=\frac{1}{2}\frac{e^{\frac{-\mathrm{Vol}(\Delta)}{\pi b^2}}}{\sqrt[4]{-\det\mathrm{Gram}(\Delta)}}\Big(1+O\big(b^2\big)\Big). $$ 
\end{proof}

\end{document}